\documentclass[10pt]{article}
\usepackage{graphicx} % Required for inserting images
\usepackage[a4paper, total={6.5in, 9in}]{geometry} % Page dimensions
\usepackage{amsfonts}
\usepackage{parskip} % parskip package to paragraph settings
\usepackage{amsthm}
\usepackage{dsfont}
\usepackage[T1]{fontenc}
\usepackage{amsmath}
\usepackage[showonlyrefs]{mathtools} % Does not number unreferenced equations
\usepackage[style=ams,
sortcites = true, % Automatically cites in numerical order
backend=biber,
style=numeric, url = false, maxnames=50, 
giveninits=true % Only employs authors initials
]{biblatex}
\renewbibmacro{in:}{}
\usepackage[shortlabels]{enumitem} % For stylized enumerate
\usepackage{hyperref}
\usepackage{tikz}
\hypersetup{
    colorlinks=true,
    linkcolor=blue,
    filecolor=magenta,      
    urlcolor=cyan,
    pdftitle={A description of classical field equations using extensions of graded Poisson brackets},
    pdfpagemode=FullScreen,
    }
\usepackage{xcolor} % xcolor is used to change font color
\usepackage{quiver} % quiver for commutative diagrams
\usepackage{orcidlink}
\usepackage{authblk} % For affiliations
\usepackage{physics} % For derivative notation (and more)

\usepackage{commands}

\newenvironment{keywords}
{
\begin{center}
\textbf{Keywords}\\
\vspace{0.17cm}
\begin{minipage}{14.5cm}}
{\footnotesize
\end{minipage}
\end{center}}

\newenvironment{Abstract}
{
\begin{center}
\textbf{Abstract}\\
\vspace{0.25cm}
\begin{minipage}{14.5cm}}
{\footnotesize
\end{minipage}
\end{center}}

\newcommand{\Ima}{\operatorname{Im}}
\newcommand{\mail}[1]{\small\href{mailto:#1}{#1}}
\newtheorem{theorem}{Theorem}[subsection]
\newtheorem{proposition}{Proposition}[subsection]
\newtheorem{lemma}{Lemma}[subsection]

\theoremstyle{definition}
\newtheorem{example}{Example}[subsection]
\newtheorem{Def}{Definition}[subsection]
\newtheorem{remark}{Remark}[subsection]

\begin{document}
%-------------------------------------------
% Title, authors and affiliations
%-------------------------------------------
\title{\huge A description of classical field equations using extensions of graded Poisson brackets}

% Author 1
\author{Manuel de León \orcidlink{0000-0002-8028-2348} \\ \mail{mdeleon@icmat.es}}
\affil{Instituto de Ciencias Matemáticas, Campus Cantoblanco, Consejo Superior de Investigaciones Científicas, C/Nicolás Cabrera, 13–15, Madrid 28049, Spain}
\affil{Real Academia de Ciencias Exactas, Físicas y Naturales de España, C/Valverde, 22, Madrid 28004, Spain}

% Author 2
\author{ Rubén Izquierdo-López \orcidlink{0009-0007-8747-344X} \\ \mail{ruben.izquierdo@unir.net}}
\affil{Escuela Superior de Ingeniería y Tecnología, Universidad Internacional de La Rioja, Logroño, Spain}
\affil{Instituto de Ciencias Matemáticas, Campus Cantoblanco, Consejo Superior de Investigaciones Científicas, C/Nicolás Cabrera, 13–15, Madrid 28049, Spain}
\date{\today}

\maketitle

%-------------------------------------------
% Abstract
%-------------------------------------------
\begin{Abstract}
As it is well-known, Poisson brackets play a fundamental role both in mechanics and in classical field theories. In this paper we develop a theory of extensions of graded Poisson brackets in graded Dirac manifolds. We then show how these extensions can be used to obtain the field equations of a particular theory as well as the evolution of forms of arbitrary order, in a similar way that ordinary Poisson brackets provide in mechanics.

\end{Abstract}

%-------------------------------------------
% Keywords
%-------------------------------------------

%-------------------------------------------
% Table of contents
%-------------------------------------------
\tableofcontents

%-------------------------------------------
% Section 1: Introduction to the results
%-------------------------------------------
\section{Introduction}
\label{section:introduction}

The finite dimensional geometric setting for studying classical field theories is multisymplectic geometry, as has been established since the 1970s (\cite{kijowski_symplectic_1979, garcia_poincare-cartan_1974, goldschmidt_hamilton-cartan_1973}), just as symplectic geometry allowed the important developments of Hamiltonian dynamics. However, unlike in symplectic mechanics, in classical field theories we deal with differential forms of degree greater than 2, which heavily difficults its study (for instance, forms of higher degree lack a Darboux theorem, see \cite{gdrr2024rev.realacad.cienc.exactasfis.nat.ser.a-mat., ryvkin2025darbouxtypetheoremsmultisymplectic} for recent surveys on the matter). In particular, this allows us to introduce a \textit{graded Poisson bracket}.

As far as we know, these brackets were first defined in \cite{cantrijn_hamiltonian_1996}, just as a continuation of the study of abstract multisymplectic manifolds carried in \cite{cantrijn_geometry_1999}. An important property of these brackets, already discussed in the work by Cantrijn, Ibort and de León, is that it induces a graded Lie algebra structure on $\Omega_H(M)/ \dd \Omega(M)$, the space of Hamiltonian forms modulo exact forms. The study of these brackets is far from being finished; indeed, a lot of recent research has been devoted to understand its interplay with field equations
\cite{castrillon_lopez_remarks_2003, gay-balmaz_new_2024, Marco_PoissonPoincare, kanatchikov_field_1997, Kanatchikov_1998}. Up to now, this study is usually developed for particular cases, without identifying the essential properties of these brackets (as is usually done in Poisson geometry, for instance).

One possible approach for this generalization is that of $L_\infty$-algebras. In \cite{baez_categorified_2010}, carrying a careful study of the exact terms present in the Jacobiator, Baez, Hoffnung and Rogers identify that for the case of a multisymplectic manifold of order $2$ (the case where the form has order $3$), the graded Poisson bracket induces what is known as a Lie-$2$-algebra. This result was later  generalized to arbitrary multisymplectic manifolds in \cite{rogers_l_infty-algebras_2012}, proving that the Poisson bracket defines an $L_\infty$-algebra. However, in order to obtain this structure, one must particularize the bracket at the highest order of forms, only studying the Jacobiator when the forms have degree $n -1$, where $n+1$ is the degree of the multisymplectic form. A more complete study is developed in \cite{delgado2018multisymplecticstructureshighermomentum}, where the defect of the Jacobi identity is studied for forms of arbitrary degree. However, $L_\infty$ algebras lack a Leibniz identity, a property which is fundamental in classical mechanics and classical field theory.

In this vein, we recently study graded Poisson brackets on their own \cite{de_leon_graded_2025}, defining a graded Poisson bracket as a graded-skew-symmetric bracket on a suitable subfamily of differential forms, called Hamiltonian, that is required to satisfy a graded Jacobi identity up to an exact term, and two extra properties relating the bracket to exterior products and contractions, representing a graded version for the Leibniz identity, and invariance 
by symmetries, respectively. The first two  properties recover the usual $L_\infty$ algebra, and taking into consideration the last two has an important implication: the bracket is completely determined by a tensorial object. A manifold equipped with this ``tensor'' is called a \textit{graded Dirac manifold} (or \textit{graded Poisson manifold}, depending on the existence of geometric gauge as defined in \cite{gaset2022geometricgaugefreedommultisymplectic, AIHPA_1979__30_2_129_0}). In the present  paper, we continue the investigation of these structures and, in particular, a definition of dynamics in terms of these brackets is given.

Some preliminary results in this direction can be found, for instance, in \cite{kanatchikov_field_1997, Marco_PoissonPoincare, Kanatchikov_1998}, where the Poisson brackets are used to define dynamics, study evolution of observables, and reduce by symmetries. Nevertheless, the brackets presented and \textit{the Hamiltonian} are dependent on a choice of volume form or coordinates (more generally, a connection).

More recently, in \cite{gay-balmaz_new_2024}, Gay-Balmaz, Marrero and Martínez-Alba developed a bracket description of the dynamics which is independent of these choices. In the multisymplectic description of field theories, fields are interpreted as sections of a bundle $Y \longrightarrow X$, called the configuration bundle, where $X$ is the base manifold, usually representing spacetime. In this setting, the Hamiltonian is interpreted as another section of a particular line bundle (see \cite{cci1991diff.geom.appl.}):
\[h: \bigwedge^n_2 Y\big / \bigwedge^{n}_1 Y \longrightarrow \bigwedge^n_2 Y\,.\] The bracket that the authors introduce is of the form $\{\alpha, h\}$, where $\alpha$ is a Hamiltonian $(n-1)$-form and it yields a semi-basic $n$-form over the base manifold $X$. This bracket is defined so that it computes the evolution of $(n-1)$-forms, in fact, a section $\psi: X \longrightarrow \bigwedge^{n}_2Y\big / \bigwedge^n_1 Y$ will be solution of the Hamilton--De Donder--Weyl equations determined by $h$ if and only it satisfies 
\[\psi^\ast (\dd \alpha) = \{\alpha, h\} \circ \psi\,, \quad \forall \alpha \in \Omega^{n-1}_H\left (\bigwedge^n_2 Y \big / \bigwedge^n_1 Y\right)\,.\] Linearity of this bracket is replaced by it being affine (as the space of Hamiltonians is affine), and the Jacobi identity is substituted by it defining an affine representation. The introduction of the previous operation raises some important questions:
\begin{enumerate}[\rm (i)]
    \item The fact that the Hamiltonians are defined as sections of certain bundle presents some technical difficulties. The first one is how can one approach reduction by symmetries using such Hamiltonians. Would the reduced Hamiltonian be a different section of a different line bundle?
    \item It is the case that the usual Poisson bracket of Hamiltonian $(n-1)$-forms recovers the classical Poisson bracket of the (infinite) dimensional instantanous formalism via integration against a space-like slicing. How does this correspondence translate with the previous extension of the bracket?
    \item In order to identify the correct structures that allow for reduction, quantization and so on, could one define such a bracket in general?
\end{enumerate}
Although these are questions that one may in principle be able to answer in terms of the Hamiltonian section, we find that these are more naturally approached if the Hamiltonian is thought of as an $n$-form. Indeed, there is a one to one correspondence between sections of the previous line bundle and the family of $n$-forms $\Theta_h = h^\ast \Theta$, where $\Theta$ is the multisymplectic potential (also called the Liouville form) on $\bigwedge^n_2 Y$. Thus, if one identifies the Hamiltonian with the previous $n$-form $\mathcal{H} := \Theta_h$ the bracket reads as $\{\alpha, \mathcal{H}\}$, where now every object is a suitable differential form on $\bigwedge^n_2 Y\big / \bigwedge^n_1 Y$.

The first and main objective of the present paper is to carry this interpretation to general graded Dirac manifolds, so that both the Hamiltonian and the observables are in equal footing, that is, interpreted as the same kind of object (differential forms of possibly different degree). The usual graded Poisson bracket is defined for Hamiltonian forms of degree between $0$ and $n - 1$, so we focus on extending these bracket to arbitrary dregree, maintaining as much properties as possible. As a result, we obtain a family of possible extensions (dependent on a choice of extensions of the $\sharp$-morphisms) and we shall obtain a bracket
\[\Omega^{a}_H(M)  \otimes \Omega^{b}_H(M) \longrightarrow \Omega^{a+b- (n-1)}(M)\,,\]
where $a$ and $b$ are no longer confined to $0 \leq a, b \leq n -1$. The restriction of this bracket to the usual Hamiltonian forms and a particular subclass of Hamiltonian $n$-forms will yield the usual equations of motion as above. 

%These apparently arbitrary choices for the extension of the brackets may make the reader uncomfortable. After all, an important facet of Poisson brackets is that they are canonical. The interpretation that we give for the range of choices is that a determination of a bracket extension yields a solution to the field equations. Since PDEs have different solutions, there will be different possible extensions. This statement is made precise in \textcolor{red}{(Teorema)}.

Regarding the mathematical construction, extensions of Poisson brackets have been studied in symplectic geometry in \cite{Michor_extensions} by Michor and in general Poisson geometry in \cite{grabowski_z-graded_1997} by Grabowski. It should be noted that the bracket defined in the second cited paper is a modification of the brackets treated in the first, this modification is performed in order to obtain a Jacobi identity (not up to an exact term). We find more suitable for our purposes to generalize the one introduced by Michor, since it is of first order. Nevertheless, both constructions are built upon a generalization of the $\sharp$ mapping, and we follow this strategy closely.

When the extensions of the brackets are finished, we use them to study the field equations and evolution of forms of arbitrary order, an aspect which was missing in previous versions of the literature. In particular, we would like to highlight the following results:
\begin{enumerate}[\rm (i)]
    \item In Theorem \ref{thm:Properties_of_first_extension}, we build a first extension of the brackets and prove that it still satisfies the main properties of the original one. Afterwards, we give sufficient conditions on the chosen extension of the $\sharp$ mapping for the final extension of the brackets to satisfy the same properties in Theorem \ref{thm_properties_final_extension}.
    \item In Theorem \ref{thm:evolution_of_arbitrary_forms} we prove that for a given extension of the brackets, the Poisson bracket of a Hamiltonian form $\alpha \in \Omega^a_H(M)$ against the Hamiltonian $\mathcal{H} \in \Omega^n_H(M)[n]$ yields a possible evolution of $\alpha$, in the sense that there is an Ehresmann connection $h$ that solves the equations determined by $\mathcal{H}$ (equations that do not depend on the extension) and such that
    $
    h^\ast(\dd \alpha) = \dd \alpha + \{\alpha, \mathcal{H}\}\,.
    $
    \item In Theorem \ref{thm:Characterization_of_fibered_Extensions}, we prove that there is a one to one correspondence between all possible extensions of the brackets to Hamiltonians compatible with the fibered structure and affine maps $\gamma:\{\text{Hamiltonians}\} \longrightarrow \{\text{connections}\}$ such that $\gamma(\mathcal{H})$ solves the equations determined by $\mathcal{H}$.
    \item Finally, in Theorem \ref{thm:determined_evolution} we identify a subfamily of forms, that does not depend on the choice of Hamiltonian, which we call \textit{special Hamiltonian forms}, whose evolution is defined for any solution of the Hamilton--De Donder--Weyl equations so that, in particular, the bracket against $\mathcal{H}$ does not depend on the choice of extension, making them candidates for being conserved quantities, namely closed forms on solutions (such forms have gained interest in theoretical physics recently, see \cite{gomes_introduction_2023}, for instance). Under certain hypotheses, which certainly hold in the examples of this paper, this subfamily is proved to be a subalgebra under the Poisson bracket in Theorem \ref{thm_sufficient_condition_subalgebra}.
\end{enumerate}

The paper is structured as follows. Section \ref{section:fibered_graded_Dirac} is devoted to study graded Dirac structures on fibered manifolds, and how these structures restrict to submanifolds and quotient manifolds. In Section \ref{section:extensions}, we develop the main ingredients of our study, the possibility of extending the natural brackets on forms to any arbitrary order. So, in Section \ref{section:dynamics} we can apply the results in the previous section to discuss dynamics on graded Dirac structures on fibered manifolds. In Section \ref{section:Applications} we show how our theory can be used to study classical field theories by providing two examples: regular Lagrangians and Yang--Mills theories. Finally, Section \ref{section:conclusions} is devoted to summarize the main conclusions and indicate some further work.

\begin{center}{\bf Notation and conventions}
\end{center}
We will use these notations an conventions throughout the paper; to facilitate the reading, they are collected here.
\begin{enumerate}
    \item All manifolds are assumed $C^\infty$-smooth and finite dimensional.
    \item Einstein's summation convention is assumed throughout the text, unless stated otherwise.
    
    \item $\bigvee_p M$ denotes the vector bundle of $p$-vectors on $M$, $\bigwedge^p TM.$
    \item $\mathfrak{X}^p(M) = \Gamma\left(\bigvee_p M\right)$ denotes the space of all multivector fields of order $p$.
    \item $[\cdot, \cdot]$ denotes the Schouten-Nijenhuis bracket on multivector fields. We use the sign conventions of \cite{Marle1997}.
    \item $\bigwedge^aM$ denotes the vector bundle of $a$-forms on $M$, $\bigwedge^a T^\ast M.$
    \item $\Omega^a(M) = \Gamma\left (\bigwedge^a M\right )$ denotes the space of all $a$-forms on $M$.
    \item When dealing with multivector fields and differential forms using coordinates, we employ the alternate Kronecker delta notation:
    \[
    \delta^{\mu_1, \dots, \mu_p}_{\nu_1, \dots, \nu_p} = \begin{cases}
        (-1)^{\sigma}  & \text{if } \sigma(\mu_1) \dots\sigma(\mu_p) = \nu_1 \dots \nu_p, \text{ for some permutation } \sigma\\
        0 & \text{otherwise}
    \end{cases}\,.
    \]
    Furthermore, we employ the multi-index notation, which we will denote as $\widetilde \mu = \mu_1, \dots, \mu_p$. When multi-indices repeat as super-index and sub-index, summation is assumed for \textit{ordered} multi-index. For instance, $\dd^{a} x^{\widetilde \mu} \otimes \pdv{^a}{^ax^{\widetilde \mu}}$ denotes the identity on $\bigwedge^a M \otimes \bigvee_a M$, rather than $\frac{1}{a!}\dd^{a} x^{\widetilde \mu} \otimes \pdv{^a}{^ax^{\widetilde \mu}}$, which would be if summation was assumed for all multi-indices.
    \item $\pounds_U \alpha = \dd{\iota_U\alpha} - (-1)^p \iota \dd{\alpha}$, for $U \in \mathfrak{X}^p(M)$,
    $\alpha \in \Omega^a(M)$ denotes the Lie derivative along multivector fields. 
    For a proof of its main properties, we refer to \cite{Forger2003a}.
    \item For a subbundle $K \subseteq \bigvee_p M,$ and for $a \geq p,$ we denote by $$K^{\circ, a} = \{\alpha \in \bigwedge^a M: \iota_K \alpha = 0\}$$ the annihilator of order $a$ of $K$.
    \item Similarly, for a subbundle $S \subseteq \bigwedge^a M,$ and $p \leq a$ we denote by 
    $$S^{\circ, p} = \{U \in \bigvee_p M: \iota_U S = 0\}$$ the annihilator of order $p$ of $S$.
\end{enumerate}

%-------------------------------------------
% Section 2: Graded Dirac structures on fibered manifolds
%-------------------------------------------
\label{section:fibered_graded_Dirac}

\section{Graded Dirac structures on fibered manifolds}
In this section we give the basic definitions and constructions on graded Dirac manifolds and develop the notion of fibered graded Dirac manifolds, the setting that we will use to study classical field theories.

\subsection{Graded Dirac structures}

Let $M$ be a smooth manifold and suppose that we have a vector subbundle $S^n$ of $\bigwedge^n M$, the bundle of $n$-forms over $M$,
% https://q.uiver.app/#q=WzAsMyxbMCwwLCJTXm4iXSxbMSwwLCJcXGJpZ3dlZGdlXm4gTSJdLFswLDEsIk0iXSxbMCwyXSxbMSwyXSxbMCwxLCIiLDIseyJzdHlsZSI6eyJ0YWlsIjp7Im5hbWUiOiJob29rIiwic2lkZSI6InRvcCJ9fX1dXQ==
\[\begin{tikzcd}
	{S^n} & {\bigwedge^n M} \\
	M
	\arrow[hook, from=1-1, to=1-2]
	\arrow[from=1-1, to=2-1]
	\arrow[from=1-2, to=2-1]
\end{tikzcd}.\]
such that the subsequent contractions of $S^n$ by $\T M$,
\[
S^{n-1}:= \iota_{\T M} S^n \,, \quad S^{n - 2}:= \iota_{\T M} S^{n-1}\,,\quad \dots \quad S^1 := \iota_{\T M} S^2\,.
\] define vector subbundles of $\bigwedge^{n-1} M, \bigwedge^{n - 2}M, \dots, \T^\ast M,$ respectively. Here we are making a slight abuse of notation denoting 
\[\iota_{\T M} S^n = \langle \iota_u \alpha  : u \in \T M, \alpha \in S^n \rangle\,,\quad  \iota_{\T M} S^{n-1} = \langle \iota_u \alpha  : u \in \T M, \alpha \in S^{n-1} \rangle\,, \dots\]

\begin{Def}[Hamiltonian form] Let $U$ be an open subset of $M$. A \textbf{Hamiltonian form} of degree $a$ on $U$, $0 \leq a \leq n - 1$ with respect to the sequence $S^1, \dots, S^n$ determined by $S^n$ is a differential form $\alpha \in \Omega^a(M)$ such that $\dd \alpha$ takes values in $ S^{a+1}$. The space of Hamiltonian forms is denoted as $\Omega^a_H(U, S^n)$. If there is no danger confusion, we will just write $\Omega^a_H(U)$ instead of $\Omega^a(U, S^n)$.
\end{Def}

Let $\alpha$ be a differential form of arbitrary degree. Define 
\[ \deg_H \alpha := (n - 1) - \deg \alpha.
\] 
We now define our main object of study (see \cite{de_leon_graded_2025} for further details):

\begin{Def}[Graded Poisson bracket] 
\label{def:graded_Poisson_bracket}
A \textbf{graded Poisson bracket of order $n$} on $M$ is a bilinear bracket defined for every open subset $U$ of $M$
\[
\Omega^a_H(U) \otimes \Omega^b_H(U) \xrightarrow{\{\cdot, \cdot\}_U} \Omega^{a+ b - (n - 1)}_H(U)\,,
\]
for $0 \leq a, b \leq n - 1$ that satisfies the following properties. For every $\alpha \in \Omega^a_H(M), \beta \in \Omega^b_H(M)$ and $\gamma \in \Omega^c_H(M)$ we have

\begin{enumerate}[i)]
        \item \textit{It is graded:} $$\deg_H \{\alpha, \beta\}_U = \deg_H \alpha + \deg_H \beta;$$
        \item \textit{It is graded-skew-symmetric:} $$\{\alpha, \beta\}_U = -(-1)^{\deg_H \alpha \deg_H \beta} \{\beta,\alpha\}_U;$$
        \item \textit{It is local:} If $\dd{\alpha} |_x = 0,$ then $\{\alpha, \beta\}|_x = 0$;
        \item \textit{It satisfies a graded Leibniz identity:} Let $\beta^j \in \Omega^{b-1}_H(U), \gamma_j \in \Omega^{c-1}_H(U).$ If $\beta^j \wedge \dd{\gamma}_j \in \Omega^{b+c - 1}_H(U),$ then we have
        $$\{ \beta^j \wedge \dd{\gamma_j}, \alpha\}_U = \{\beta^j, \alpha\}_U \wedge \dd{\gamma}_j + (-1)^{n -\deg_H \beta^j} \dd{\beta}^j \wedge \{\gamma_j, \alpha\}_U;$$ for $\alpha \in \Omega^{n-1}_H(M)$; 
        \item \textit{It is invariant by symmetries:} If $X \in \mathfrak{X}(U)$ and $\pounds_X \alpha = 0,$ then 
        $\iota_X \alpha \in \Omega^{a-2}_H(U)$ and
        $$\{\iota_X \alpha, \beta\}_U=  (-1)^{\deg_H \beta} \iota_X \{\alpha, \beta\}_U;$$
        \item \textit{It satisfies the graded Jacobi identity (up to an exact term):} 
        $$(-1)^{\deg_H \alpha \deg_H \gamma}\{ \{\alpha, \beta \}_U , \gamma\}_U + \text{cyclic terms} = \text{exact form}.$$
    \end{enumerate}
    Finally, it satisfies the \textit{compatibility condition}, namely for $V \subseteq U$ open, the bracket on V, $\{\cdot, \cdot\}_V$,
    is the restriction of the bracket on $U$, $\{\cdot, \cdot\}_U.$
\end{Def}

\begin{remark}
The last condition implies that we have a well-defined global bracket \[
\Omega^a_H(M) \otimes \Omega^b_H(M) \xrightarrow{\{\cdot, \cdot\}} \Omega^{a+ b - (n - 1)}_H(M)\] and, making some abuse of notation, we will denote $\{\cdot, \cdot\}_U = \{\cdot, \cdot\}$. Let us explain why it is necessary to consider the bracket locally. This is due to the fact that $\Omega^a_H(M)$ could be the space of closed forms, even when $S^n$ is non zero. For instance, for $n = 1$, if we take the foliation of the torus where every leaf is dense and $S^1$ is the bundle of forms which are zero when evaluated on tangent vectors to this foliation, then $\dd f \in S^1$ only when $f$ is constant.
\end{remark}

Actually, the situation described above could still happen for arbitrarily small open subsets $U$, not due to topology, but due to integrability. Notice that in order to have a non-trivial space of Hamiltonian forms (where the trivial one would be defined as the space of closed forms), there must be a non-trivial solution to the following partial differential equation:
\[
\dd \alpha \in S^{a+1}\,, \quad \alpha \in \Omega_H^a(U)\,,
\]
where $\dd \alpha \in S^{a+1}$ means that $\dd \alpha$ takes values in $S^{a+1}$. Hence, if we are not careful, there could be an unique graded Poisson bracket for a given subbundle $S^n$, the trivial one. The same discussion applies to the graded Leibniz identity and invariance by symmetries. These properties could be meaningless if there are no Hamiltonian forms and vector fields satisfying the requirements. In order to avoid this, we can assume that $S^n$ satisfies the following conditions (the theories to be considered will certainly satisfy these hypotheses):

\begin{enumerate}[\rm (i)] 
        \item \label{enumerate:Flatness_1} Locally, for every $1 \leq a \leq n$, there exist Hamiltonian forms $\gamma_{ij} \in \Omega^{a-1}_H(U),$ and functions $f^j_i$ such that
        $S^{a+1} = \langle \dd{f}^j_i \wedge \dd{\gamma}_{ij}, i\rangle$.
        \item \label{enumerate:Flatness_2} Locally, 
        for each $1 \leq a \leq n $,  there exists a family of Hamiltonian forms forms ${\gamma}^j,$ and a family of vector fields $X^i$ such that 
        $S^a = \langle \dd{\gamma}^j\rangle, \,\,\pounds_{X^i} \gamma^j = 0$, and $S^{a-1} = \langle \dd{\iota_{X^i} \gamma^j}\rangle$.
\end{enumerate}

These hypotheses, besides guaranteeing that $S^a$ is locally generated by Hamiltonian forms, also guarantee that we can apply freely the Leibniz rule and the invariance by symmetries. In particular, we can prove the following:

\begin{theorem}[\cite{de_leon_graded_2025}]
\label{thm:bracket_tensor_equivalence}
Let $S^n$ be a vector subbundle of $\bigwedge^n M$ such that $S^{n-1} = \iota_{\T M} S^n, \dots, S^1 = \iota_{\T M} S^2$ are vector subbundles of $\bigwedge^{n-1}M, \dots, T^\ast M,$ respectively. Suppose that these subbundles satisfy the hypotheses above. Let $\{\cdot, \cdot\}$ be a graded Poisson bracket on the space of Hamiltonian forms. Then, there exists an unique sequence of vector bundle morphisms, $\sharp_a$ for $1 \leq a \leq n$ 
% https://q.uiver.app/#q=WzAsMyxbMCwwLCJTXmEiXSxbMSwwLCJcXGJpZ3ZlZV9hIE0vS19hIl0sWzAsMSwiTSJdLFswLDJdLFsxLDJdLFswLDEsIlxcc2hhcnBfYSJdXQ==
\[\begin{tikzcd}
	{S^a} & {\bigvee_{n+1-a} M/K_{n+1-a}} \\
	M
	\arrow["{\sharp_a}", from=1-1, to=1-2]
	\arrow[from=1-1, to=2-1]
	\arrow[from=1-2, to=2-1]
\end{tikzcd},\] where $K_{n+1-a} = \{U \in \bigvee_{n+1-a} M: \iota_U S^{n+1-a} = 0\},$
satisfying the following properties:
\begin{enumerate}[\rm i)]
        \item The maps $\sharp_a$ are \textit{skew-symmetric}, namely
        $$\iota_{\sharp_a(\alpha)} \beta = (-1)^{(n+1-a)(n+1-b)}\iota_{\sharp_b(\beta)} \alpha,$$ for all 
        $\alpha \in S^a,$ $\beta \in S^b.$
        \item It is \textit{integrable:} For $\alpha: M \longrightarrow S^a$, $\beta: M \longrightarrow S^b$ sections such that $a + b \leq 2n+1$, and $U, V$ multivectors
        of degree $p = n + 1 - a$, $q = n+1-b$, respectively such that $$\sharp_a(\alpha) = U + K_p, \,\, \sharp_b(\beta)= V + K_q,$$ we have
        that the $(a+ b - n)$-form 
        $$\theta := (-1)^{(p-1)q}\pounds_U \beta + (-1)^{q} \pounds_V \alpha - \frac{(-1)^q}{2} d \left( \iota_V \alpha + (-1)^{pq} \iota_U \beta\right)$$
        takes values in $S_{a+b-n},$ and $$\sharp_{a+b - n}(\theta) = [U, V] + K_{ p +q - 1}.$$
        \item For any Hamiltonian forms $\alpha \in \Omega^a_H(M)$ and $\beta \in \Omega^b_H(M)$, we have 
        \begin{equation}
            \label{Poisson_bracket_definition}
            \{\alpha, \beta\} = (-1)^{\deg_H \beta} \iota_{\sharp_{b+1}(\dd \beta)} \dd \alpha \,.
        \end{equation}
    \end{enumerate}
    Conversely, given a sequence of vector bundle maps satisfying the first two previous properties, the bracket defined by Eq. \eqref{Poisson_bracket_definition} defines a graded Poisson bracket of order $n$.
\end{theorem}

The previous result induces the following definition:
\begin{Def}[Regular graded Dirac structure of order $n$]
\label{def:graded_Dirac} Let $M$ be a smooth manifold, and  $S^1, \dots, S^n$ be a sequence of vector subbundles % https://q.uiver.app/#q=WzAsMyxbMCwwLCJTXmEiXSxbMSwwLCIgXFxiaWd3ZWRnZV5hTSJdLFswLDEsIk0iXSxbMCwxLCIiLDAseyJzdHlsZSI6eyJ0YWlsIjp7Im5hbWUiOiJob29rIiwic2lkZSI6InRvcCJ9fX1dLFswLDJdLFsxLDJdXQ==
\[\begin{tikzcd}
	{S^a} & { \bigwedge^aM} \\
	M
	\arrow[hook, from=1-1, to=1-2]
	\arrow[from=1-1, to=2-1]
	\arrow[from=1-2, to=2-1]
\end{tikzcd}\]
which are related through subsequent contractions, $S^{a} = \iota_{\T M} S^{a+1}$ (we do not necessarily require them to satisfy the integrability conditions mentioned above). Then, a \textbf{regular graded Dirac structure of order $n$} on $M$ is a sequence of vector bundle maps
\[\sharp_a: S^a \longrightarrow \bigvee_{n+1-a} M/ K_{k+1-a},\]
where $K_{n+1-a}$ is the annihilator of $S^{n+1-a}$ that satisfy the following conditions:
\begin{enumerate}[i)]
        \item The maps $\sharp_a$ are \textit{skew-symmetric}, that is, 
        $$\iota_{\sharp_a(\alpha)} \beta = (-1)^{(n+1-a)(n+1-b)}\iota_{\sharp_b(\beta)} \alpha,$$ for all 
        $\alpha \in S^a,$ $\beta \in S^b.$
        \item It is \textit{integrable:} For $\alpha: M \longrightarrow S^a$, $\beta: M \longrightarrow S^b$ sections such that $a + b \leq 2n+1$, and $U, V$ multivector fields
        of order $p = n + 1 - a$, $q = n+1-b$, respectively such that $$\sharp_a(\alpha) = U + K_p, \,\, \sharp_b(\beta)= V + K_q,$$ we have
        that the $(a+ b - n)$-form 
        $$\theta := (-1)^{(p-1)q}\pounds_U \beta + (-1)^{q} \pounds_V \alpha - \frac{(-1)^q}{2} d \left( \iota_V \alpha + (-1)^{pq} \iota_U \beta\right)$$
        takes values in $S_{a+b-n},$ and $$\sharp_{a+b - k}(\theta) = [U, V] + K_{ p +q - 1}.$$
    \end{enumerate}
\end{Def}

\begin{remark} Notice that a regular graded Dirac structure $(M, S^a, \sharp_a)$ induces a graded Poisson bracket on the space of Hamiltonian forms, by defining $\{\alpha, \beta\} = (-1)^{\deg_H \beta} \iota_{\sharp_{b+1}(\dd \beta)} \dd \alpha$.
\end{remark}

\begin{Def}[Flat regular graded Dirac structure] A regular graded Dirac structure of order $n$ will be called \textbf{flat} if the sequence of subbundles $S^1, \dots, S^n$ satisfy the integrability hypotheses \ref{enumerate:Flatness_1} and \ref{enumerate:Flatness_2}.
\end{Def}
\begin{remark} Equivalently, using Theorem \ref{thm:bracket_tensor_equivalence}, a flat regular graded Dirac structure of order $n$ can be defined by a graded Poisson bracket of order $n$ on the space of Hamiltonian forms $\Omega^a_H(M),$ for $0 \leq a \leq n-1.$
\end{remark}

\begin{remark}
    One could weaken the hypotheses of the objects in Definition \ref{def:graded_Dirac}, and work instead with vector-subbundles of \[
\bigvee_p M \oplus_M \bigwedge^{n+1-p} M,
\]
as in \cite{vankerschaver_geometry_2011, vankerschaver_hamilton-pontryagin_2012, bursztyn_higher_2019, de_leon_graded_2025}, corresponding to the image of the vector bundle morphism \[D^p = \{(u, \alpha) : \alpha \in S^{n+1-p}\,, \sharp_{n+1-p}(\alpha) = u + K_p\}\,,\]see \cite{de_leon_graded_2025}, Proposition III A.2. We prefer to work with regular graded Dirac structures because of its immediate connection to brackets, which is our main focus in the current paper. Nevertheless, the results presented may be extended to the general case with no difficulties.
\end{remark}

\begin{remark} The case $n = 1$ recovers the usual definition of Dirac structure. Indeed, a regular ``graded'' Dirac structure of order one on $M$ would correspond to a map $$\sharp: S \longrightarrow \T M/K,$$ where $S \subseteq \T^\ast M$ is a vector subbundle. Integrability implies that $K$ is an integrable distribution in the sense of Frobenius. Then, both skew-symmetry and integrability of $\sharp$ translates into\[
D :=\{(v, \alpha) \in \T M \oplus_M T^\ast M: \sharp(\alpha) = v + K\} \subset \T M \oplus_M T^\ast M
\]
defining a Dirac structure in the usual case present in the literature (see \cite{bursztyn_brief_2011, courant_dirac_1990, dorfman_dirac_1987, dorfman_dirac_1993}). Of course, this does not recover all possible Dirac structures, but only those $D \subset \T M \oplus_M T^\ast M$ whose projection onto $T^\ast M$ defines a vector subbundle. In any case, we can recover all possible Dirac structures if we do not ask for regularity.
\end{remark}

\begin{remark}
\label{remark:structure_at_n_determines_structure}
An important feature of these structures is that they are completely determined by the structure on degree $n$, since the vector bundle map $\sharp_n \colon S^n \longrightarrow \T M / K_1$ determines the rest defining $\sharp_a(\iota_U \alpha) := \sharp_n(\alpha) \wedge U + K_{n+1-a},$ where $\alpha \in S^n$ and $U \in \bigvee_{n - a} M$ (see \cite{de_leon_graded_2025}).
\end{remark}

\subsection{Restrictions and quotients of graded Dirac structures}

Now we will study induced graded Dirac structures by restrictions to sub-manifolds and projections to a quotient. More generally, we will define the \textit{pullback} and \textit{pushforward} of a graded Dirac structure by a smooth map (see \cite{bursztyn2008diracgeometryquasipoissonactions} for the definitions in Dirac geometry).

Let $N, M$ and $P$ be smooth manifolds, and let $M$ be endowed with a regular graded Dirac structure $(S^a, \sharp_a)$. Suppose we had an immersion
$i: N \longrightarrow M$ and a submersion $\pi: M \longrightarrow P$. One should define the induced structures by $i$ and $\pi$ on $N$ and $P$, respectively, as maximal regular Dirac structures such that both $i$ and $\pi$ preserve them. The preservation of the structure is given by the notion of a \textit{graded Dirac map}. In Dirac geometry (in some contexts called generalized geometry, see \cite{hitchin2010lecturesgeneralizedgeometry}), there are two notions, forward Dirac maps (generalizaing pushforward of multivectors) and backward Dirac maps (generalizing pullbacks of forms), see \cite{bursztyn2008diracgeometryquasipoissonactions}. Let us first define a general notion of graded Dirac map in order to define the generalizations of the previous concepts.

\begin{Def}
\label{def:graded_Dirac_map}
Let $(N, \widetilde S^a, \widetilde \sharp_a)$ and $(M, S^a, \sharp_a)$ be two regular graded Dirac manifolds and $f\colon N \longrightarrow M$ be a smooth map. Then, $f$ is said to be a \textbf{graded Dirac map} if
the following diagram is commutative for every $x  \in N$ and every $a = 1, \dots, n$:
% https://q.uiver.app/#q=WzAsNCxbMCwwLCIoZl5cXGFzdCApXnstMX0gXFxsZWZ0KFxcd2lkZXRpbGRlIFNeYSBcXGJpZyB8X3hcXHJpZ2h0KSBcXGJpZ2NhcCBTXmFcXGJpZyB8X3tmKHgpfSJdLFswLDEsIlxcYmlnd2VkZ2Vee24rMS1hIH1cXGxlZnQoXFxUX3tmKHgpfSBNXFxyaWdodCkgXFxiaWcgLyBcXGxlZnQoS197bisxLWF9ICsgZl9cXGFzdCggXFx3aWRldGlsZGUgS197bisxLWF9KSBcXHJpZ2h0KSJdLFsyLDEsIlxcYmlnd2VkZ2Vee24rMS1hfVxcbGVmdChcXFRfeCBOIFxccmlnaHQpIFxcYmlnIC8gXFx3aWRldGlsZGUgS197bisxLWF9Il0sWzIsMCwiXFx3aWRldGlsZGUgU15hIFxcYmlnIHxfeCJdLFswLDEsIlxcc2hhcnBfYSJdLFszLDIsIlxcd2lkZXRpbGRlIFxcc2hhcnBfYSIsMV0sWzAsMywiZl5cXGFzdCIsMV0sWzIsMSwiZl9cXGFzdCIsMV1d
\[\begin{tikzcd}
	{(f^\ast )^{-1} \left(\widetilde S^a \big |_x\right) \bigcap S^a\big |_{f(x)}} && {\widetilde S^a \big |_x} \\
	{\bigwedge^{n+1-a }\left(\T_{f(x)} M\right) \big / \left(K_{n+1-a} + f_\ast( \widetilde K_{n+1-a}) \right)} && {\bigwedge^{n+1-a}\left(\T_x N \right) \big / \widetilde K_{n+1-a}}
	\arrow["{f^\ast}"{description}, from=1-1, to=1-3]
	\arrow["{\sharp_a}", from=1-1, to=2-1]
	\arrow["{\widetilde \sharp_a}"{description}, from=1-3, to=2-3]
	\arrow["{f_\ast}"{description}, from=2-3, to=2-1]
\end{tikzcd}.\]
Here we are making some abuse of notation, the map on the left $\sharp_a$ should be interpreted as the composition \[
S^a \xrightarrow{\sharp_a} \bigvee_{n+1-a} M \big / K_{n+1-a} \longrightarrow \bigvee_{n+1-a} M \big / \left(K_{n+1-a} + f_\ast( \widetilde K_{n+1-a}) \right),
\]
where the second map is the projection onto the quotient. Hence, a smooth map $f\colon N \longrightarrow M$ is said to be a graded Dirac map if it satisfies $\sharp_a(\alpha |_{f(x)}) = f_\ast \left( \widetilde \sharp_{a}(f^\ast \alpha)\right)$ modulo $K_{n+1-a} + f_\ast( \widetilde K_{n+1-a})$ for every $\alpha \in S^a |_{f(x)}$ such that $f^\ast \alpha \in \widetilde S^a$.
\end{Def}

That Definition \ref{def:graded_Dirac_map} is the correct notion for a structure preserving map between graded Dirac manifolds is motivated by the following property:

\begin{proposition}
\label{prop:graded_dirac_maps_preserve_brackets}
Let $(N, \widetilde S^a, \widetilde \sharp_a)$ and $(M, S^a, \sharp_a)$ be two regular graded Dirac manifolds of order $n$ and $f: N \longrightarrow M$ be a graded Dirac map. Let $\alpha$ and $\beta$ be Hamiltonian forms on $M$ such that both $f^\ast \alpha$ and $f^\ast \beta$ are Hamiltonian on $N$. Then, 
\[\{f^\ast \alpha, f^\ast \beta\} = f^\ast \{\alpha, \beta\}\,.\]
\end{proposition}
\begin{proof} Indeed, we have
\begin{align*}
    f^\ast \{\alpha, \beta\} & = (-1) ^{\deg_H \beta}f^\ast \left( \iota_{\sharp_{b}(\dd \beta)} \dd \alpha \right)\\
    &= (-1)^{\deg_H \beta} f^\ast \left( \iota_{f_\ast \sharp_b(\dd f^\ast \beta)} \dd\alpha\right) = (-1)^{\deg_H \beta}\iota_{\sharp_b(\dd f^\ast \beta)} f^\ast \dd \alpha \\
    &= \{f^\ast \alpha, f^\ast \beta\}\,,
\end{align*}
which finishes the proof.
\end{proof}

Furthermore, in the case that $f$ is behaved \textit{well enough} with respect to both graded Dirac structures, Definition \ref{def:graded_Dirac_map} is not only sufficient but also necessary for a map $f: N \longrightarrow M$ to preserve brackets in the sense of Proposition \ref{prop:graded_dirac_maps_preserve_brackets}: 
\begin{proposition} Let $(N, \widetilde S^a, \widetilde \sharp_a)$ and $(M, S^a, \sharp_a)$ be two regular graded Dirac manifolds of order $n$ and $f: N \longrightarrow M$ be a smooth map such that for every $\alpha \in \Omega_H^a(N)$ there is $\gamma \in \Omega^a_H(M)$ with $f^\ast \gamma = \alpha$. Then, $f$ is a graded Dirac map if and only if it satisfies \[\{f^\ast \alpha, f^\ast \beta\} = f^\ast \{\alpha, \beta\}\,,\]
for every pair of Hamiltonian forms $\alpha$ and $\beta$ on $M$ such that both $f^\ast \alpha$ and $f^\ast \beta$ are Hamiltonian forms on $N$.
\end{proposition}
\begin{proof} Necessity follows from Proposition \ref{prop:graded_dirac_maps_preserve_brackets}. Now, for sufficiency, let $\alpha \in (f^\ast)^{-1} \left( \widetilde S^a\right) \cap S^a$. We need to show that $\sharp_a(\alpha) = f_\ast \sharp_a(f^\ast \alpha)$ $\operatorname{mod} K_{n+1-a} + f_\ast (\widetilde K_{n+1-a})$. Now, since
\begin{align*}
    \left((f^\ast)^{-1} \left( \widetilde S^{n+1-a}\right) \cap S^{n+1-a}\right)^{\circ, n+1-a} &= \left((f^\ast)^{-1} ( \widetilde S^{n+1-a})\right) ^{\circ, n+1-a} + (S^{n+1-a})^{\circ, n+1-a}\\
    &=  K_{n+1-a} + f_\ast (\widetilde K_{n+1-a})\,,
\end{align*}
it is enough to show that $\iota_{\sharp_a(\alpha)} \beta = \iota_{f_\ast \sharp_a(f^\ast \alpha)} \beta$, for every $\beta \in (f^\ast)^{-1} \left( \widetilde S^{n+1-a} \right) \cap S^{n+1-a}$, which clearly follows from the bracket preservation hypothesis.
\end{proof}

Now we can define the desired induced graded Dirac structures as maximal structures (with respect to inclusion and extension of the corresponding $\sharp$ morphisms) on $N$ and $P$ such that both $i$ and $\pi$ are graded Dirac maps. These are characterized by two particular kind of maps, those that we call \textit{backward} and \textit{forward}, generalizing the well known concepts in Dirac geometry \cite{bdn2022math.z.}:

\begin{Def}[Backward and forward maps] A graded Dirac map $f: N \longrightarrow M$ is said to be \textbf{backward} if
$\widetilde S^a \big |_{x} = \{f^\ast\alpha : \alpha \in S^a \big |_{f(x)} \text{ and } \sharp_a(\alpha) \text{ is tangent to $f$, modulo $K_{n+1-a}$}\},$ and it is said to be \textbf{forward} if $S^a \big |_{f(x)} = (f^{\ast})^{-1}\left( \widetilde S^a \big |_x\right)$, for every $a = 1, \dots, n$ and $x \in N$.
\end{Def}

\begin{Def}[Pullback and pushforward of graded Dirac structures]
\label{def:pullback_and_pushforward}
Let $(M, S^a, \sharp_a)$ be a regular graded Dirac manifold of order $n$ and $f: N \longrightarrow M$ be a smooth map (respectively, $\pi: M \longrightarrow P$, be a submersion). A graded Dirac structure on $N$, $(\widetilde S^a, \widetilde \sharp_a)$ (respectively, on $P$) is called the \textbf{pullback} of $(S^a, \sharp_a)$ (respectively, the \textbf{pushforward}) if $f$ is a backward (repectively, forward) graded Dirac map between these two structures.
\end{Def}

Some elementary properties are in order:

\begin{theorem}
\label{thm:uniqueness_of_pullback_and_pushforward}
When they exist, both pullback and pushforward of graded Dirac structures are unique.
\end{theorem}

\begin{proof} Let $(M, S^a, \sharp_a)$ be a graded Dirac manifold, $f: N \longrightarrow M$ be a smooth map and $\pi: M \longrightarrow P$ be a submersion. Clearly, if there are graded Dirac structures on $N$ and $P$ such that $f$ is a backward graded Dirac map and $\pi$ is a forward graded Dirac map, the correspoding subbundles $\widetilde S^a$ are unique, by definition. Therefore, it is sufficient to check that the choice of the $\sharp$ maps is unique as well. Let us study both cases separatedly:

\begin{enumerate}[\rm (i)]
    \item \underline{Pullbacks}: Notice that
    $(f^\ast)^{-1}\left(\widetilde S^a \right) \cap S^a \xrightarrow{f^\ast} \widetilde S^a$
    is surjective. Suppose there were two maps
    \[\widetilde \sharp_a \,, \widetilde{ \sharp}'_a : \widetilde S^a \longrightarrow \bigvee_{n+1-a} N \big / \widetilde K_{n+1-a}\]
    such that $f_\ast (\widetilde \sharp_a(f^\ast \alpha)) = f_\ast (\widetilde \sharp_a'(f^\ast \alpha)) = \sharp_a(\alpha)$ modulo $K_{n+1-a} + f_\ast \widetilde K_{n+1-a},$ for every $\alpha \in (f^\ast)^{-1}\left(\widetilde S^a \right) \cap S^a$. Then, for every $\alpha \in S^a$, we would have \[f_\ast (\widetilde \sharp_a(f^\ast \alpha)) - f_\ast (\widetilde \sharp_a'(f^\ast \alpha)) \in (K_{n+1-a} + f_\ast \widetilde K_{n+1-a}) \cap f_\ast\left(\, \bigvee_{n+1-a} N\right) \subseteq f_\ast \widetilde K_{n+1-a}\,,\] so that $f_\ast (\widetilde \sharp_a(f^\ast \alpha)-\sharp_a'(f^\ast \alpha)) \in f_\ast (\widetilde K_{n+1-a})$. Hence, $ \sharp_a(f^\ast \alpha)-\sharp_a'(f^\ast \alpha) \in (f_\ast)^{-1} \left(f_\ast (\widetilde K_{n+1-a}) \right)$. Since this last subspace is easily seen to be $\widetilde K_{n+1-a}$, we have $\sharp_a(f^\ast \alpha)-\sharp_a'(f^\ast \alpha) \in \widetilde K_{n+1-a}$, finishing the proof.
    \item \underline{Pushforwards}: Uniqueness in the case of pushforwards follows from the observation that $\pi_\ast (K_{n+1-a}) \subseteq \widetilde K_{n+1-a}$, given that $\pi^\ast$ is injective (since $\pi$ is a submersion). Therefore, if $\pi$ is a forward graded Dirac map, we necessarily have a commutative diagram
    % https://q.uiver.app/#q=WzAsNCxbMCwwLCIoXFxwaV5cXGFzdCleey0xfSBcXGxlZnQoIFNeYVxccmlnaHQpIl0sWzIsMCwiU15hIl0sWzAsMSwiXFxiaWd2ZWVfe24rMS1hfSBQIFxcYmlnIC8gXFx3aWRldGlsZGUgS197bisxLWF9Il0sWzIsMSwiXFxiaWd2ZWVfe24rMS1hfSBNIFxcYmlnIC8gS197bisxLWF9Il0sWzAsMSwiXFxwaV5cXGFzdCIsMV0sWzAsMiwiXFx3aWRldGlsZGUgXFxzaGFycF9hIiwyXSxbMSwzLCJcXHNoYXJwX2EiXSxbMywyLCJcXHBpX1xcYXN0IiwxXV0=
\[\begin{tikzcd}
	{(\pi^\ast)^{-1} \left( S^a\right)} && {S^a} \\
	{\bigvee_{n+1-a} P \big / \widetilde K_{n+1-a}} && {\bigvee_{n+1-a} M \big / K_{n+1-a}}
	\arrow["{\pi^\ast}"{description}, from=1-1, to=1-3]
	\arrow["{\widetilde \sharp_a}"', from=1-1, to=2-1]
	\arrow["{\sharp_a}", from=1-3, to=2-3]
	\arrow["{\pi_\ast}"{description}, from=2-3, to=2-1]
\end{tikzcd},\]
from which uniqueness follows.
\end{enumerate}
\end{proof}

\begin{remark} Actually, from the analysis done in Theorem \ref{thm:uniqueness_of_pullback_and_pushforward} we can explicitly compute the pullback and pushforward of regular graded Dirac structures. This will result useful in the sequel.
\end{remark}

We now turn onto the problem of determining the conditions for existence of both pullback and pushforward of graded Dirac structures. The matter of existence of pullbacks is purely a matter of regularity:

\begin{proposition} 
\label{prop:existence_of_pullback}
Let $f: N \longrightarrow M$ be a smooth map and let $M$ be endowed with a regular graded Dirac structure $(S^a,\sharp_a)$. Suppose 
\[ \widetilde S^{a} := \{f^\ast \alpha : \sharp_a(\alpha) \text{ is tangent to } f\}\]
defines a vector subbundle of $\bigwedge^{a} N$, for every $1 \leq a \leq n$. Then, there exists the pullback of the graded Dirac structure by $f$.
\end{proposition}
\begin{proof} Skew symmetry of the maps follows easily from the definition. Integrability follows from the naturality of the Schouten--Nijenhuis bracket. 
\end{proof}

In particular, if an embedding $i: N \longrightarrow M$ is well behaved with respect to the graded Dirac structure on $M$, then we can restrict said structure to $N$. However, the existence of pushforwards is only guaranteed under the following hypotheses:

\begin{proposition}
\label{prop:existence_of_pushforward}
Let $\pi : M \longrightarrow P$ be a submersion and let $M$ be endowed with a graded Dirac structure $(S^a, \sharp_a)$. Then, the pushforward of $(S^a, \sharp_a)$ exists if and only if for each pair of points $x$ and $y \in M$ such that $\pi(x) = \pi(y)$ we have
\begin{enumerate}[\rm (i)]
    \item $\{\alpha \in \bigwedge^{a} \T\,^\ast_{\pi(x)} P : (\dd_x\pi)^\ast \alpha \in S^a |_x\} = \{\alpha \in \bigwedge^{a} \T\,^\ast_{\pi(x)} P : (\dd_y\pi)^\ast \alpha \in S^a |_y\},$
    \item Given $\alpha \in \bigwedge^a \T\,^\ast_{\pi(x)} P$ such that $(\dd_x\pi)^\ast \alpha \in S^a |_x$ and $(\dd_\pi)^\ast \alpha \in S^a |_y$, then
    \[(\dd_x \pi)_\ast \sharp_a ((\dd_x \pi)^\ast \alpha) = (\dd_y \pi)_\ast \sharp_a ((\dd_y \pi)^\ast \alpha)\,.\]
\end{enumerate}
\end{proposition}
\begin{proof} Clearly, the previous hypotheses allow us to define induced maps $\widetilde \sharp_a : \widetilde S^a \longrightarrow \bigvee_{n+1-a} P / \widetilde K_{n+1-a}$ that are skew-symmetric. Integrability follows from the naturality of the Lie derivative and Schouten--Nijenhuis bracket.
\end{proof}

Such hypotheses are guaranteed, for instance, when dealing with a quotient by a Lie group action that preserves the graded Dirac structure, so that the space of orbits inherits a natural graded Dirac structure, namely the pushforward of Proposition \ref{prop:existence_of_pushforward}. However, an in-depth study of reduction is not the objective of this paper.

%Recall (see \cref{remark:structure_at_n_determines_structure}) that graded Dirac structures are completely determined by the structure at degree $n$. Hence, it is natural to ask wether $f$ is a graded Dirac map if and only if it satisfies the three conditions of \cref{def:graded_Dirac_map} but only in degree $n$. The answer to this question is positive, under certain technical hypotheses which will be discussed later.

\subsection{Fibered graded Dirac manifolds}

Now, we shall focus on describing graded Dirac structures compatible with a fibered manifold $\tau: M \longrightarrow X$, which is the structure that we use to describe classical field theories.

Let $\tau: M \rightarrow X$ be a fibered manifold, namely, a surjective submersion. Let us first introduce some bit of terminology that we will use constantly throughout our study:

\begin{Def}[Horizontal, basic, semi-basic forms] A differential $a$-form $\alpha \in \Omega^a(M)$ is called:
\begin{enumerate}[\rm (i)]
    \item \textbf{Semi-basic} over $X$ with respect to the projection $\tau: M \longrightarrow X$ if $\iota_v \alpha = 0$ for every vector $v \in  \T M$ such that $\dd \tau \cdot v = 0$.
    \item \textbf{Basic} over $X$ with respect to the projection $\tau: M \longrightarrow X$ if $\alpha = \tau^\ast \beta$, for certain $a$-form $\beta$ (which will be unique) on $X$.
    \item  \textbf{$r$-horizontal} over $X$ with respect to the projection $\tau: M \longrightarrow X$, for $0 \leq r \leq a$ if $\iota_{v_1, \dots, v_{a+1-r}} \alpha = 0$ for every sequence of vectors $v_1, \dots, v_{a+1-r} \in  \T M$ such that $\dd \tau \cdot v_i = 0$.
\end{enumerate}
\end{Def}

\begin{remark} Notice that the notion of semi-basic $a$-form and $a$-horizontal $a$-form coincide.
\end{remark}

\begin{Def}[Fibered graded Dirac manifolds]
\label{def:fibered_graded_Dirac}
Let $\tau \colon M \longrightarrow X$ be a fibered manifold over an $n$-dimensional manifold $X$. A graded Dirac structure on $M$ of order $n$, $(S^a, \sharp_a)$ is called \textbf{fibered over $X$} if the following two conditions hold:
\begin{enumerate}[(i)]
    \item  $S^a$ consists of {$(a-1)$-horizontal} forms over $X$.
    \item Every semi-basic $a$-form over $X$ is in $S^a$.
\end{enumerate}
\end{Def}

\begin{remark}
Let us motivate this definition. First, as we will see in Theorem \ref{thm:defined_evolution_implies_Hamiltonian}, Hamiltonian forms on a graded Dirac manifold will be the forms that are candidates to have defined evolution. Since we want this evolution to be linear on the future Hamiltonian, we ask for the $a$-forms to be $(a-1)$-horizontal (or, equivalently, at most $1$-vertical). Second, it is only natural to ask for basic forms to have defined evolution. This issue is represented by the second condition. 
\end{remark}

We now proceed to study the relationship of fibered structures to the previous constructions of graded Dirac maps, pullbacks and pushforwards:

\begin{proposition} Let $\tau_1 \colon N \longrightarrow X$, $\tau_2\colon P \longrightarrow X$ and $\tau\colon M \longrightarrow X$ be fibered manifolds and let $\tau \colon M \longrightarrow X$ be endowed with a fibered graded Dirac structure. Then, both pullback and pushforwards (if they are well defined) by smooth maps $f: N \longrightarrow M$ and submersions $\pi\colon M \longrightarrow P$, respectively, are fibered when $f$ and $\pi$ are fibered maps.
\end{proposition}

\begin{proof} Both are easily check from the definition of backward and forward graded Dirac maps.
\end{proof}

Let us study an example which will be of relevance for our purposes:

\begin{example}
\label{example:multisymplectic_field_theory}
Let $\pi \colon Y \longrightarrow X$ be a fibered manifold over an $n$-dimensional base, $X$, and let $\widetilde M := \bigwedge^n_2 Y \subseteq \bigwedge^n Y$ denote the vector subbundle of $(n-1)$-horizontal forms, that is, the space consisting of those $\alpha \in \bigwedge^n_2 Y$ such that $\iota_{e_1 \wedge e_2} \alpha = 0$, for every $e_i \in \T Y$ with $\dd \pi (e_i) = 0$. If $(x^\mu, y^i)$ denote fibered coordinates on $Y$ so that $\pi(x^\mu, y^i) = (x^\mu)$, the previous vector bundle is generated by the forms 
\[
\dd^n x\,, \quad \dd y^i \wedge \dd ^{n-1}x_\mu\,,
\]
where $\dd^n x$ and $\dd^{n-1}x_\mu$ denote $\dd x^1 \wedge \cdots \wedge \dd x^n$ and $\iota_{\pdv{x^\mu}} \dd^{n}x$, respectively. Therefore, the subbundle $\bigwedge^n_2 Y$ admits natural coordinates $(x^\mu, y^i, p, p^\mu_i)$ representing the form
$\alpha = p \dd^n x + p^\mu_i \dd y^i \wedge \dd^{n-1}x_\mu$.
There is a canonical form, called the Liouville form (or the canonical multisymplectic potential) 
\[
\Theta = p \dd^n x + p^\mu_i \dd y^i \wedge \dd^{n-1}x_\mu\,,
\]
which induces a canonical multisymplectic form on this manifold, which is given by the closed non-degenerate form
\[
\Omega = -\dd \Theta = -\dd p \wedge \dd ^n x - \dd p^\mu_i \wedge \dd y^i \wedge \dd^{n-1}x_\mu\,.
\]
This closed form induces a map $\flat: \T \widetilde M \longrightarrow \bigwedge^n \widetilde M$ given by $\flat(v) := \iota_v \Omega$. This map defines a monomorphism, so that its inverse $\sharp_n: = \flat^{-1}$ is a well defined map between $S^n := \Ima \flat$ and $\T M$: 
\[
\sharp_n \colon S^n \longrightarrow \T M\,, \quad \sharp_n(\alpha) = \flat ^{-1}(\alpha)\,.
\]
In coordinates, this map reads as follows
\begin{align*}
    &\sharp_n(\dd^n x) = -\pdv{p}\,, \quad &\sharp_n(\dd y^i \wedge \dd^{n-1}x_\mu) = -\pdv{p^\mu_i}\,,\\
    &\sharp_n(\dd p^\mu_i \wedge \dd^{n-1}x_\mu) =  \pdv{y^i}\,, \quad & \sharp_n (\dd p \wedge \dd^{n-1}x_\nu - \dd p^\mu_i \wedge \dd y^i \wedge \dd ^{n-1}x_{\mu\nu}) =  \pdv{x^\nu}\,. 
\end{align*}
The reader may check that this is a graded Dirac structure but, as we can see, it is not fibered over $X$. This setting (usually done via the form $\Omega$) is called the \textit{extended} covariant formalism of classical field theory. Then, there is the \textit{reduced} formalism, which takes place in the manifold obtained by quotienting $\widetilde M$ by semi-basic forms, $\bigwedge^n_1 Y$ so that the manifold is $M := \bigwedge^n_2 Y \big / \bigwedge^n_1 Y$ with induced coordinates $(x^\mu, y^i, p^\mu_i)$. When working in the reduced formalism, the reader may find one of two options in the literature:
\begin{enumerate}[\rm (i)]
    \item The first option, given a section (representing a Hamiltonian) $h: M \longrightarrow \widetilde M$, to work with the multisymplectic structure induced by $\Omega_h= h^\ast \Omega$, see \cite{Rom_n_Roy_2009, cci1991diff.geom.appl.}, for instance.
    \item The second option, to identify those basic Hamiltonian forms (with respect to the projection $\widetilde M \longrightarrow M$) whose bracket is still basic, and work with the reduced bracket on $M$, see \cite{Marco_PoissonPoincare, castrillon_lopez_remarks_2003}, among others.
\end{enumerate}
The first method is uncomfortable if one aims to develop a bracket theory independent of the Hamiltonian, so we will work with the second \cite{castrillon_lopez_remarks_2003}. This actually corresponds to our notion of pushforward of a graded Dirac structure. The induced structure is easily seen to be $S^n 
 = \langle\dd^n x, \dd y^i \wedge \dd^{n-1}x_\mu, \dd p^\mu_i \wedge \dd^{n-1}x_\mu \rangle$ with $\sharp_n$ given by
\[
 \sharp_n(\dd^n x) = 0\,, \quad \sharp_n(\dd y^i \wedge \dd^{n-1}x_\mu) = -\pdv{p^\mu_i}\,, \quad
    \sharp_n(\dd p^\mu_i \wedge \dd^{n-1}x_\mu) = \pdv{y^i}\,.
\]
We immediately see that this is a fibered graded Dirac manifold over $X$.
\end{example}

\begin{remark} In the extended formalism there are sveral and equivalent ways of determinining the dynamics, see \cite{wagner2024notesequivalentformulationshamiltonian} for a recent treatment of the subject.
\end{remark}

%-------------------------------------------
% Section 3: Extensions of graded Poisson brackets
%-------------------------------------------
\section{Extensions of graded Poisson brackets}
\label{section:extensions}

We now turn to the main objective of our study: defining extensions of graded Poisson brackets on a regular graded Dirac manifold. More particularly, we are interested in finding a bracket 
\[
\Omega_H^a(U) \otimes \Omega_H^b(U) \longrightarrow \Omega^{a+b - (n-1)}_H(U),
\]
that satisfies as much properties of Definition \ref{def:graded_Poisson_bracket} as possible, where $a, b \geq 0$ are now bounded from above only by the dimension of $M$. As we will immediately see, this bracket will depend upon certain choices. Nevertheless, there will be pairs of forms such that the bracket is independent of all choices made. These pairs will then be used to define dynamics.

We will use the techniques employed in \cite{grabowski_z-graded_1997, Michor_extensions}, that is, we start by defining a suitable extension of the $\sharp$ mappings to a wider family of forms, and then the brackets will be defined by Equation \eqref{Poisson_bracket_definition} up to a constant, namely
\[\{\alpha, \beta\} := (-1)^{\deg_H \beta} C \iota_{\sharp(\dd \beta)} \dd \alpha\,,\] where $\alpha$ and $\beta$ are suitable forms of arbitrary order and $\sharp(\dd \beta)$ is the aforementioned extension. Our study is divided in three steps:
\begin{enumerate}[\rm (i)]
    \item \underline{First step}: In order to motivate the construction in the two subsequent steps, we show that the maps \[\sharp_a: S^a \longrightarrow \bigvee_{n+1-a} M / K_{n+1-a}\] for $1 \leq a \leq n$ may be extended to maps $\sharp_a : S^b \longrightarrow S^{b-a} \otimes \bigvee_{n+1-a} M / K_{n+1-a}$. Furthermore, one may define the brackets using any possible extension, since we will show that there is a constant $C_{a,b}$ such that $\iota_{\sharp_b(\beta)} \gamma = C_{a, b}  \cdot \iota_{\sharp_a(\beta)}  \gamma$, for every $\beta \in S^b$, $b \geq a$ and $\gamma \in S^c$ with $c \geq n + 1 - a$.
    \item \underline{Second step}: We show how to extend the first $\sharp$-mapping
    \[ \sharp_1 \colon S^1 \longrightarrow \bigvee_n M / K_n\] to a map $\widetilde \sharp_1:  (S^1)^{\wedge a} \longrightarrow \bigwedge^{a-1} M \otimes \bigvee_{n} M / K_n$ in a canonical way, for arbitrary $a$. This extension will allow us to introduce a bracket of Hamiltonian $(n-1)$-forms and arbitrary forms $\Theta$ with $\dd \Theta \in (S^1)^{\wedge a}$ using Equation \eqref{Poisson_bracket_definition}. This is the bracket that we will use to define dynamics.
    \item \underline{Last step}: We deal with the issue of extending the rest of the $\sharp$-mappings to forms of arbitrary degree. Compatible extensions of $\widetilde \sharp_j$ with the map $\widetilde \sharp_1$ presented in the second step will only exist for a particular subfamily of forms, those whose exterior derivative takes values in a particular subbundle which we denote by $S^a[j]$. Thus, we will look for extensions $\widetilde \sharp_j : S^a[j] \longrightarrow \bigwedge^{a-j} M \otimes \bigvee_{n+1-j} M / K_j\,$, and study their main properties.
\end{enumerate}

\subsection{First step: Extension of \texorpdfstring{$\sharp_a$}{} to \texorpdfstring{$S^b$}{}, for \texorpdfstring{$n \geq b \geq a$}{}}
Just as a motivation, let us first notice that we may extend the mapping
\[
\sharp_a: S^a \longrightarrow \bigvee_{n+1-a}M/ K_{n+1-a}
\]
to every subbundle $S^b,$ with $n \geq b \geq a$ as follows.
\begin{proposition}
\label{prop:Extension_to_Hamiltonian}
For $1 \leq a \leq n$ and $a \leq b \leq n$ there is an unique extension \[
 \sharp_a: S^b \longrightarrow S^{b - a} \otimes \bigvee_{n+1-a} M/ K_{n+1-a}\,,
\]
such that, if $\beta \in S^b$ and $\gamma \in S^c$ with $c \geq n+1-a$, then we have
\begin{equation}
\label{equation_condition_of_extension}
    \iota_{\sharp_a(\beta)}\gamma = \binom{b + c - (n+1)}{b - a} \iota_{\sharp_b(\beta)} \gamma = \frac{(b+c- n+1)!}{ (b-a)! (a + c - (n+1))!} \iota_{\sharp_b(\beta)} \gamma\,.
\end{equation}
Here, contraction by a multivector valued form is defined as $\iota_{\theta \otimes u} \gamma := \theta \wedge \iota_u \gamma$, generalizing the usual contraction by vector valued forms (see \cite{grabowski_z-graded_1997}).
\end{proposition}
\begin{proof}
Particularizing Equation \eqref{equation_condition_of_extension} for $\gamma \in S^{n+1-a},$ we must have
$
\iota_{\sharp_a(\beta)} \gamma = \iota_{\sharp_b(\beta)} \gamma.
$
The second term defines a linear mapping
\[
S^{b} \otimes S^{n +1 - a} \longrightarrow S^{b - a}.
\]
Dualizing the previous mapping, we see that there exists an unique map
\[
\sharp_a\colon S^b \longrightarrow S^{b-a} \otimes \bigvee_{n+1-a}M/K_{n+1-a}
\]
that satisfies Equation \eqref{equation_condition_of_extension} for $c = n+1-a$. It only remains to check that such a map will also satisfy it for arbitrary $n \geq c \geq n +1 - a$. This will be a consequence of the following lemmas:
\begin{lemma}
\label{lemma1}
Let $V$ be a finite dimensional vector space and $u \in \bigwedge^p V.$ For $a \geq p$ we have an induced map \[
\bigwedge^a V^\ast \longrightarrow \bigwedge^{a-p} V^\ast, \, \alpha \mapsto \iota_u \alpha.
\]
Then, interpreting this linear map as an element $
C_u \in \bigwedge^{a - p} V^\ast\otimes \bigwedge^a V$, we have \[
C_u = (-1)^{p (a- p)} \mathds{1}_{a-p} \wedge u = u \wedge \mathds{1}_{a-p},
\]
where $\mathds{1}_{a-p} \in \bigwedge^{a-p}
V^\ast \otimes \bigwedge^{a-p} V$ denotes the identity, and the exterior product is given by $(\theta \otimes v) \wedge u = \theta \otimes (v \wedge u)$ or $u \wedge (\theta \otimes v) = \theta \otimes (u \wedge v)$.
\end{lemma}
\begin{proof} Let $\{e_i\}$ denote a basis of $V$ and suppose $u = u^{i_1 \dots i_p} e_{i_1} \wedge \cdots \wedge e_{i_p}$. Then, \[
\iota_u e^{j_1} \wedge \cdots\wedge  e^{j_a} = \delta^{j_1 \dots j_a}_{i_1 \dots i_p k_{p+1} \dots k_a} u^{i_1 \dots i_p} e^{k_{p+1}} \wedge \cdots \wedge e^{k_a}\,.
\]
Hence, we have 
\[C_u = \delta^{j_1 \dots j_a}_{i_1 \dots i_p k_{p+1} \dots k_a} u^{i_1 \dots i_p} e^{k_{p+1}} \wedge \cdots \wedge e^{k_a} \otimes e_{j_1} \wedge \cdots\wedge  e_{j_a}.\]
Furthermore, $\mathds{1}_{a-p} = e^{k_{p+1}} \wedge \cdots \wedge e^{k_a} \otimes e_{k_{p+1}} \wedge \cdots \wedge e_{k_a}$ and we have
\begin{align*}
    \mathds{1}_{a-p} \wedge u &=  u^{i_1 \dots i_p} e^{k_{p+1}} \wedge \cdots \wedge e^{k_a} \otimes e_{k_{p+1}} \wedge \cdots \wedge e_{k_a} \wedge e_{i_1} \wedge \cdots \wedge e_{i_p}\\
    &= \delta^{j_1 \dots j_a}_{k_{p+1} \dots k_ai_1 \dots i_p}u^{i_1 \dots i_p}e^{k_{p+1}} \wedge \cdots \wedge e^{k_a} \otimes e_{j_1} \wedge \cdots \wedge e_{j_a} \\
    &= (-1)^{p (a-p)} \delta^{j_1 \dots j_a}_{i_1 \dots i_p k_{p+1} \dots k_a} u^{i_1 \dots i_p} e^{k_{p+1}} \wedge \cdots \wedge e^{k_a} \otimes e_{j_1} \wedge \cdots\wedge  e_{j_a}\,,
\end{align*}
where in the last equality we have switched $k_{p+1} \dots k_ai_1 \dots i_p$ to $i_1 \dots i_p k_{p+1} \dots k_a$. Comparing both expressions we conclude that $C_u = (-1)^{p (a-p)} \mathds{1}_{a-p} \wedge u$, finishing the proof.
\end{proof}

\begin{lemma}
\label{lemma2}
Let $\mathds{1}_a \in \bigwedge^a V^\ast \otimes \bigwedge^a V$ denote the identity. Then, for $b \geq a$, and $\beta \in \bigwedge^b V^\ast,$ we have
\[
\iota_{\mathds{1}_a} \beta = \binom{b}{a}\beta = \frac{b!}{a! (b-a)!} \beta\,.
\]
\end{lemma}
\begin{proof} Again, let $\{e_i\}$ denote a basis of $V$ and suppose $\beta = \beta_{i_1 \dots i_b} e^{i_1}\wedge \cdots e^{i_b}$. Now, since \[\mathds{1}_a = e^{j_1} \wedge \cdots \wedge e^{j_a} \otimes e_{j_1} \wedge \cdots \wedge e_{j_a}\,\] we have 
\begin{align*}
    \iota_{\mathds{1}_{a}} \beta &= \beta_{i_1 \dots i_b}e^{j_1} \wedge \cdots \wedge e^{j_a}  \wedge \left( \iota_{e_{j_1} \wedge \cdots \wedge e_{j_a}} e^{i_1} \wedge \cdots \wedge e^{i_b}\right)\\
    &=\delta_{j_1 \dots j_a k_{a+1} \dots k_b}^{ i_1 \dots i_b} \beta_{i_1 \dots i_b}  e^{j_1} \wedge \cdots \wedge e^{j_a} \wedge e^{k_{a+1}} \wedge \cdots \wedge e^{k_b} \\
    &= \binom{b}{a}\beta_{i_1 \dots i_b} e^{i_1} \wedge \cdots \wedge e^{i_b} = \binom{b}{a}\beta,
\end{align*}
which finishes the proof.
\end{proof}

Using Lemma \ref{lemma1} it follows that for $\beta \in S^b$ with $b \geq a$, we may write
\[
\sharp_a(\beta) =  \sharp_b(\beta) \wedge \mathds{1}_{b - a} .
\]
In general, notice that for a multivector valued form $\theta \otimes v \in \bigwedge^{a} V^\ast \otimes \bigwedge^q V$, and $u \in \bigwedge^p V$ we have that
\[
\iota_{u \wedge (\theta \otimes v)} \gamma = \iota_{\theta \otimes (u \wedge v)} \gamma = \theta \wedge \iota_{u \wedge v} \gamma = 
\theta \wedge \iota_v \iota_u \gamma = \iota_{\theta \otimes v} \iota_u \gamma,
\]
for every $\gamma \in \bigwedge^c V^\ast$.
Hence, for $\gamma \in S^c$, with $n \geq c \geq n + 1 - a$,
\[
\iota_{\sharp_a(\beta)} \gamma = \iota_{\mathds{1}_{b-a}} \iota_{\sharp_b(\beta)} \gamma\,.
\]
Now, since $c \geq n+1 - a$, we have $c - (n+1 - b) \geq b - a$ and, using Lemma \ref{lemma2}, we conclude that
\[
\iota_{\sharp_a(\beta)} \gamma = \binom{b + c - (n+1)}{b - a}\iota_{\sharp_b(\beta)}\gamma,
\]
which finishes the proof of Proposition \ref{prop:Extension_to_Hamiltonian}.
\end{proof}

\begin{remark} This implies that in the definition of graded Poisson brackets we can use any $\sharp$ mapping, as long as we are careful with the orders of the forms to be considered and the coefficients appearing in Proposition \ref{prop:Extension_to_Hamiltonian}. Indeed we know that
\[
\{\gamma, \beta\} = (-1)^{\deg_H \beta} {\binom{b + c - (n+1)}{b - a}} ^{-1}{\iota_{\sharp_a(\dd \beta)} \dd \gamma,} = (-1)^{\deg_H \beta} \frac{(b-a)! (c + a -( n+1))!}{ (b+c - (n+1))!} \iota_{\sharp_{a} (\dd \beta)} \dd \gamma\,,\]
 for every $0 \leq a \leq b,$ and $n +1 - a \geq c$, since we have
 \[
  \iota_{\sharp_a(\dd \beta)} \dd \gamma= \binom{b + c - (n+1)}{b - a}\iota_{\sharp_b(\dd \beta)} \dd \gamma\,.
 \]
\end{remark}
\begin{remark}
\label{remark:formula_for_first_extension}
Also notice, as it follows from Lemma \ref{lemma1}, that we have an explicit expression for $\sharp_a(\beta),$ namely
\[
\sharp_a(\beta) =  \sharp_b(\beta) \wedge \mathds{1}_{a-b}\,.
\]
This will result useful in the sequel, being the condition we will impose for compatibility among extensions.
%A similar expression will be of fundamental importance later on.

\end{remark}

\subsection{Second step: Extension of \texorpdfstring{$\sharp_1$}{} to \texorpdfstring{$(S^1)^{\wedge a}$}{}
, for \texorpdfstring{$a \geq 1$}{}}

We are now ready to start defining the extensions of the $\sharp$-mappings.

\begin{theorem}
\label{thm:extension_of_sharp_1}
Let $(M, S^a, \sharp_a)$ be a regular graded Dirac manifold of order $n$. Then, there exists an unique extension of $\sharp_1$, namely,
\[
(S^1)^{\wedge a} \xrightarrow{\widetilde \sharp_1} \bigwedge^{a - 1}M \otimes_M \bigvee_n M/K_n\,,
\]
for all possible values of $a \geq 1$, such that for any $\alpha \in S^n$ and $\beta \in (S^1)^{\wedge a}$ we have
\[
\iota_{\sharp_n(\alpha)} \beta  = (-1)^{(n + 1 - a)} \iota_{\widetilde \sharp_1(\beta)} \alpha\,.
\]

\end{theorem}
\begin{proof} Notice that the map $\alpha \otimes \beta \mapsto (-1)^{(n+1-a)}\iota_{\sharp_n(\alpha)} \beta$ for $\alpha \in S^n$, $\beta \in (S^1)^{\wedge a}$ defines a map 
$S^n \otimes (S^1)^{\wedge a} \longrightarrow \bigwedge^{a - 1 } M$.
Dualizing, we obtain $\widetilde  \sharp_1 \colon(S^1)^{\wedge a} \xrightarrow{\widetilde \sharp_1} \bigwedge^{a-1} M \otimes \bigvee _n M / K_n$. By construction, it satisfies the required equality.

Furthermore, the equality $\iota_{\sharp_n(\alpha)} \beta  = (-1)^{ (n + 1 - a)} \iota_{\widetilde \sharp_1(\beta)} \alpha$ just means that $\widetilde\sharp_1$ is the dual (in the previous sense) of the first map defined.
\end{proof}

\begin{remark} Signs are chosen throughout the text in the $\sharp$ mappings so that the following equality holds: $\iota_{\sharp(\alpha)} \beta = (-1)^{(n+1-a)(n+1-b)} \iota_{\sharp(\beta)} \alpha$, maintaining skew-symmetry as in Definition \ref{def:graded_Dirac}. Notice that this is the condition imposed in Theorem \ref{thm:extension_of_sharp_1}.
\end{remark}

\begin{remark}
\label{remark:anti_derivation} 
The previous extension of $\sharp_1$ recovers the construction by Michor in \cite{Michor_extensions} and by Grabowski in \cite{grabowski_z-graded_1997}. Indeed, it is easily shown that the map from Proposition \ref{prop:Extension_to_Hamiltonian} satisfies 
\[\widetilde \sharp_1(\theta_1 \wedge \cdots \wedge \theta_a) = (-1)^{a+1}\sum_{j = 1}^{a} (-1)^{j+1} \theta_1 \wedge \cdots \wedge \hat{\theta}_j \wedge \cdots \wedge \theta_a \otimes \sharp_1(\theta_j)\,,\]
so that it defines an anti-derivation, which is the main property used in both cited papers for their respective construction. We prefer the point of view of the proof of Proposition \ref{prop:Extension_to_Hamiltonian} because of its immediate generalization to different degrees, which is not so easily accomplished from the point of view of derivations.
\end{remark}

As a corollary of Proposition \ref{prop:Extension_to_Hamiltonian}, we can define an extension of brackets for Hamiltonian $(n-1)-$forms and forms of arbitrary order in a suitable subspace: Let $$\Omega_H^a(U)[1] := \{ \alpha \in \Omega^a(U): \dd \alpha  \in (S^1)^{\wedge(a+1)}\}\,,$$ where $U$ is an open subset of $M$. The extension of $\sharp_1$ defines a bracket 
\[
\Omega_H^{n-1}(U) \otimes \Omega^{a}_H(U)[1] \xrightarrow{ \{ \cdot, \cdot \}} \Omega^{a}(M)
\]
by putting
\begin{equation}
    \label{eq:bracket_sharp_1}
    \{\alpha, \Theta\} := (-1)^{\deg_H \Theta} \iota_{\widetilde \sharp_1(\dd \Theta)} \dd \alpha\,.
\end{equation}
By Proposition \ref{prop:Extension_to_Hamiltonian}, this bracket coincides with the original one for Hamiltonian forms $\Theta \in \Omega^b_H(M)$, $0 \leq b \leq n-1.$

\begin{remark}
The diagram below shows the achieved the extension with this first step: There is a point on the coordinate plane $(\order \alpha, \order \beta)$ if we can evaluate a bracket on a suitable pair of forms of such an order. In the gray square we have represented the original bracket of Hamiltonian forms, and we have extended it to the case $(n-1, a)$ or $(a, n-1),$ for arbitrary $a \geq 0.$

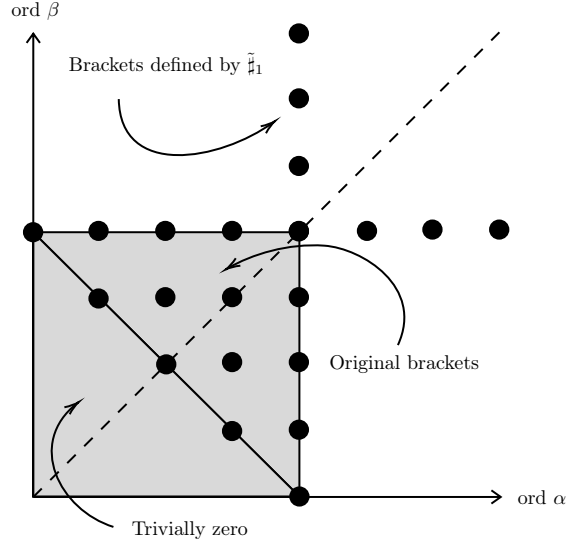
\begin{center}
\tikzset{every picture/.style={line width=0.75pt}} %set default line width to 0.75pt        

\begin{tikzpicture}[x=0.5pt,y=0.5pt,yscale=-1,xscale=1]
%uncomment if require: \path (0,460); %set diagram left start at 0, and has height of 460

%Shape: Square [id:dp1905891204130199] 
\draw  [fill={rgb, 255:red, 217; green, 217; blue, 217 }  ,fill opacity=1 ] (50.85,200.45) -- (250.25,200.45) -- (250.25,399.85) -- (50.85,399.85) -- cycle ;
%Shape: Circle [id:dp7770347199924486] 
\draw  [fill={rgb, 255:red, 0; green, 0; blue, 0 }  ,fill opacity=1 ] (243,199.45) .. controls (243,195.67) and (246.07,192.6) .. (249.85,192.6) .. controls (253.63,192.6) and (256.7,195.67) .. (256.7,199.45) .. controls (256.7,203.23) and (253.63,206.3) .. (249.85,206.3) .. controls (246.07,206.3) and (243,203.23) .. (243,199.45) -- cycle ;
%Shape: Circle [id:dp31327571020367273] 
\draw  [fill={rgb, 255:red, 0; green, 0; blue, 0 }  ,fill opacity=1 ] (243,249.45) .. controls (243,245.67) and (246.07,242.6) .. (249.85,242.6) .. controls (253.63,242.6) and (256.7,245.67) .. (256.7,249.45) .. controls (256.7,253.23) and (253.63,256.3) .. (249.85,256.3) .. controls (246.07,256.3) and (243,253.23) .. (243,249.45) -- cycle ;
%Shape: Circle [id:dp044507100126703225] 
\draw  [fill={rgb, 255:red, 0; green, 0; blue, 0 }  ,fill opacity=1 ] (243,298.45) .. controls (243,294.67) and (246.07,291.6) .. (249.85,291.6) .. controls (253.63,291.6) and (256.7,294.67) .. (256.7,298.45) .. controls (256.7,302.23) and (253.63,305.3) .. (249.85,305.3) .. controls (246.07,305.3) and (243,302.23) .. (243,298.45) -- cycle ;
%Shape: Circle [id:dp6306025955000303] 
\draw  [fill={rgb, 255:red, 0; green, 0; blue, 0 }  ,fill opacity=1 ] (243,349.45) .. controls (243,345.67) and (246.07,342.6) .. (249.85,342.6) .. controls (253.63,342.6) and (256.7,345.67) .. (256.7,349.45) .. controls (256.7,353.23) and (253.63,356.3) .. (249.85,356.3) .. controls (246.07,356.3) and (243,353.23) .. (243,349.45) -- cycle ;
%Shape: Circle [id:dp1141546871262955] 
\draw  [fill={rgb, 255:red, 0; green, 0; blue, 0 }  ,fill opacity=1 ] (93,199.45) .. controls (93,195.67) and (96.07,192.6) .. (99.85,192.6) .. controls (103.63,192.6) and (106.7,195.67) .. (106.7,199.45) .. controls (106.7,203.23) and (103.63,206.3) .. (99.85,206.3) .. controls (96.07,206.3) and (93,203.23) .. (93,199.45) -- cycle ;
%Shape: Circle [id:dp05948980891523892] 
\draw  [fill={rgb, 255:red, 0; green, 0; blue, 0 }  ,fill opacity=1 ] (143,199.45) .. controls (143,195.67) and (146.07,192.6) .. (149.85,192.6) .. controls (153.63,192.6) and (156.7,195.67) .. (156.7,199.45) .. controls (156.7,203.23) and (153.63,206.3) .. (149.85,206.3) .. controls (146.07,206.3) and (143,203.23) .. (143,199.45) -- cycle ;
%Shape: Circle [id:dp8994141232888382] 
\draw  [fill={rgb, 255:red, 0; green, 0; blue, 0 }  ,fill opacity=1 ] (143,249.45) .. controls (143,245.67) and (146.07,242.6) .. (149.85,242.6) .. controls (153.63,242.6) and (156.7,245.67) .. (156.7,249.45) .. controls (156.7,253.23) and (153.63,256.3) .. (149.85,256.3) .. controls (146.07,256.3) and (143,253.23) .. (143,249.45) -- cycle ;
%Shape: Circle [id:dp8483205424096769] 
\draw  [fill={rgb, 255:red, 0; green, 0; blue, 0 }  ,fill opacity=1 ] (193,199.45) .. controls (193,195.67) and (196.07,192.6) .. (199.85,192.6) .. controls (203.63,192.6) and (206.7,195.67) .. (206.7,199.45) .. controls (206.7,203.23) and (203.63,206.3) .. (199.85,206.3) .. controls (196.07,206.3) and (193,203.23) .. (193,199.45) -- cycle ;
%Shape: Circle [id:dp22836979147841419] 
\draw  [fill={rgb, 255:red, 0; green, 0; blue, 0 }  ,fill opacity=1 ] (193,249.45) .. controls (193,245.67) and (196.07,242.6) .. (199.85,242.6) .. controls (203.63,242.6) and (206.7,245.67) .. (206.7,249.45) .. controls (206.7,253.23) and (203.63,256.3) .. (199.85,256.3) .. controls (196.07,256.3) and (193,253.23) .. (193,249.45) -- cycle ;
%Shape: Circle [id:dp24790511846957264] 
\draw  [fill={rgb, 255:red, 0; green, 0; blue, 0 }  ,fill opacity=1 ] (193,298.45) .. controls (193,294.67) and (196.07,291.6) .. (199.85,291.6) .. controls (203.63,291.6) and (206.7,294.67) .. (206.7,298.45) .. controls (206.7,302.23) and (203.63,305.3) .. (199.85,305.3) .. controls (196.07,305.3) and (193,302.23) .. (193,298.45) -- cycle ;
%Straight Lines [id:da2980554772351536] 
\draw    (50.85,200.45) -- (250.25,399.85) ;
%Shape: Circle [id:dp6685074710091325] 
\draw  [fill={rgb, 255:red, 0; green, 0; blue, 0 }  ,fill opacity=1 ] (243,50.45) .. controls (243,46.67) and (246.07,43.6) .. (249.85,43.6) .. controls (253.63,43.6) and (256.7,46.67) .. (256.7,50.45) .. controls (256.7,54.23) and (253.63,57.3) .. (249.85,57.3) .. controls (246.07,57.3) and (243,54.23) .. (243,50.45) -- cycle ;
%Shape: Circle [id:dp18640724259856678] 
\draw  [fill={rgb, 255:red, 0; green, 0; blue, 0 }  ,fill opacity=1 ] (243,99.45) .. controls (243,95.67) and (246.07,92.6) .. (249.85,92.6) .. controls (253.63,92.6) and (256.7,95.67) .. (256.7,99.45) .. controls (256.7,103.23) and (253.63,106.3) .. (249.85,106.3) .. controls (246.07,106.3) and (243,103.23) .. (243,99.45) -- cycle ;
%Shape: Circle [id:dp3890329483197599] 
\draw  [fill={rgb, 255:red, 0; green, 0; blue, 0 }  ,fill opacity=1 ] (243,150.45) .. controls (243,146.67) and (246.07,143.6) .. (249.85,143.6) .. controls (253.63,143.6) and (256.7,146.67) .. (256.7,150.45) .. controls (256.7,154.23) and (253.63,157.3) .. (249.85,157.3) .. controls (246.07,157.3) and (243,154.23) .. (243,150.45) -- cycle ;
%Shape: Circle [id:dp8187589187687252] 
\draw  [fill={rgb, 255:red, 0; green, 0; blue, 0 }  ,fill opacity=1 ] (243,199.85) .. controls (243,196.07) and (246.07,193) .. (249.85,193) .. controls (253.63,193) and (256.7,196.07) .. (256.7,199.85) .. controls (256.7,203.63) and (253.63,206.7) .. (249.85,206.7) .. controls (246.07,206.7) and (243,203.63) .. (243,199.85) -- cycle ;
%Shape: Circle [id:dp8523635313110536] 
\draw  [fill={rgb, 255:red, 0; green, 0; blue, 0 }  ,fill opacity=1 ] (294,199.45) .. controls (294,195.67) and (297.07,192.6) .. (300.85,192.6) .. controls (304.63,192.6) and (307.7,195.67) .. (307.7,199.45) .. controls (307.7,203.23) and (304.63,206.3) .. (300.85,206.3) .. controls (297.07,206.3) and (294,203.23) .. (294,199.45) -- cycle ;
%Shape: Circle [id:dp3406807953849085] 
\draw  [fill={rgb, 255:red, 0; green, 0; blue, 0 }  ,fill opacity=1 ] (343,198.45) .. controls (343,194.67) and (346.07,191.6) .. (349.85,191.6) .. controls (353.63,191.6) and (356.7,194.67) .. (356.7,198.45) .. controls (356.7,202.23) and (353.63,205.3) .. (349.85,205.3) .. controls (346.07,205.3) and (343,202.23) .. (343,198.45) -- cycle ;
%Shape: Circle [id:dp6115802608299312] 
\draw  [fill={rgb, 255:red, 0; green, 0; blue, 0 }  ,fill opacity=1 ] (393,198.45) .. controls (393,194.67) and (396.07,191.6) .. (399.85,191.6) .. controls (403.63,191.6) and (406.7,194.67) .. (406.7,198.45) .. controls (406.7,202.23) and (403.63,205.3) .. (399.85,205.3) .. controls (396.07,205.3) and (393,202.23) .. (393,198.45) -- cycle ;
%Shape: Right Triangle [id:dp25884064404261653] 
\draw  [fill={rgb, 255:red, 217; green, 217; blue, 217 }  ,fill opacity=1 ] (50.85,200.45) -- (250.25,399.85) -- (50.85,399.85) -- cycle ;
%Shape: Circle [id:dp923109190474837] 
\draw  [fill={rgb, 255:red, 0; green, 0; blue, 0 }  ,fill opacity=1 ] (44,200.45) .. controls (44,196.67) and (47.07,193.6) .. (50.85,193.6) .. controls (54.63,193.6) and (57.7,196.67) .. (57.7,200.45) .. controls (57.7,204.23) and (54.63,207.3) .. (50.85,207.3) .. controls (47.07,207.3) and (44,204.23) .. (44,200.45) -- cycle ;
%Shape: Circle [id:dp907016083359832] 
\draw  [fill={rgb, 255:red, 0; green, 0; blue, 0 }  ,fill opacity=1 ] (93,250.45) .. controls (93,246.67) and (96.07,243.6) .. (99.85,243.6) .. controls (103.63,243.6) and (106.7,246.67) .. (106.7,250.45) .. controls (106.7,254.23) and (103.63,257.3) .. (99.85,257.3) .. controls (96.07,257.3) and (93,254.23) .. (93,250.45) -- cycle ;
%Shape: Circle [id:dp8001963309796896] 
\draw  [fill={rgb, 255:red, 0; green, 0; blue, 0 }  ,fill opacity=1 ] (143.7,300.15) .. controls (143.7,296.37) and (146.77,293.3) .. (150.55,293.3) .. controls (154.33,293.3) and (157.4,296.37) .. (157.4,300.15) .. controls (157.4,303.93) and (154.33,307) .. (150.55,307) .. controls (146.77,307) and (143.7,303.93) .. (143.7,300.15) -- cycle ;
%Shape: Circle [id:dp7618375158920645] 
\draw  [fill={rgb, 255:red, 0; green, 0; blue, 0 }  ,fill opacity=1 ] (193,350.45) .. controls (193,346.67) and (196.07,343.6) .. (199.85,343.6) .. controls (203.63,343.6) and (206.7,346.67) .. (206.7,350.45) .. controls (206.7,354.23) and (203.63,357.3) .. (199.85,357.3) .. controls (196.07,357.3) and (193,354.23) .. (193,350.45) -- cycle ;
%Shape: Circle [id:dp6658678933686535] 
\draw  [fill={rgb, 255:red, 0; green, 0; blue, 0 }  ,fill opacity=1 ] (243.4,399.85) .. controls (243.4,396.07) and (246.47,393) .. (250.25,393) .. controls (254.03,393) and (257.1,396.07) .. (257.1,399.85) .. controls (257.1,403.63) and (254.03,406.7) .. (250.25,406.7) .. controls (246.47,406.7) and (243.4,403.63) .. (243.4,399.85) -- cycle ;
%Curve Lines [id:da2033894503289626] 
\draw    (111,422.8) .. controls (70.41,409.93) and (45.5,359.62) .. (83.82,329.51) ;
\draw [shift={(85,328.6)}, rotate = 143.13] [color={rgb, 255:red, 0; green, 0; blue, 0 }  ][line width=0.75]    (10.93,-3.29) .. controls (6.95,-1.4) and (3.31,-0.3) .. (0,0) .. controls (3.31,0.3) and (6.95,1.4) .. (10.93,3.29)   ;
%Curve Lines [id:da32764444031156326] 
\draw    (324,285.6) .. controls (344,241.6) and (297,212.6) .. (269,210.6) .. controls (241.84,208.66) and (214.68,217.07) .. (196.65,227.62) ;
\draw [shift={(195,228.6)}, rotate = 328.57] [color={rgb, 255:red, 0; green, 0; blue, 0 }  ][line width=0.75]    (10.93,-3.29) .. controls (6.95,-1.4) and (3.31,-0.3) .. (0,0) .. controls (3.31,0.3) and (6.95,1.4) .. (10.93,3.29)   ;
%Shape: Axis 2D [id:dp029705070273106315] 
\draw  (51,400.2) -- (401,400.2)(51,50.6) -- (51,400.2) -- cycle (394,395.2) -- (401,400.2) -- (394,405.2) (46,57.6) -- (51,50.6) -- (56,57.6)  ;
%Curve Lines [id:da07897622673278937] 
\draw    (115,100) .. controls (115.99,159.99) and (190.49,146.87) .. (230.79,117.5) ;
\draw [shift={(232,116.6)}, rotate = 143.13] [color={rgb, 255:red, 0; green, 0; blue, 0 }  ][line width=0.75]    (10.93,-3.29) .. controls (6.95,-1.4) and (3.31,-0.3) .. (0,0) .. controls (3.31,0.3) and (6.95,1.4) .. (10.93,3.29)   ;
%Straight Lines [id:da3034911560932234] 
\draw  [dash pattern={on 4.5pt off 4.5pt}]  (51,400.2) -- (400,49.6) ;

% Text Node
\draw (123,416.6) node [anchor=north west][inner sep=0.75pt]  [xscale=0.75,yscale=0.75] [align=left] {Trivially zero};
% Text Node
\draw (271,292) node [anchor=north west][inner sep=0.75pt]  [xscale=0.75,yscale=0.75] [align=left] {Original brackets};
% Text Node
\draw (76,63) node [anchor=north west][inner sep=0.75pt]  [xscale=0.75,yscale=0.75] [align=left] {Brackets defined by $\displaystyle \tilde{\sharp }_{1}$};
% Text Node
\draw (412,393.4) node [anchor=north west][inner sep=0.75pt]  [xscale=0.75,yscale=0.75]  {$\operatorname{ord\ \alpha }$};
% Text Node
\draw (33,25.4) node [anchor=north west][inner sep=0.75pt]  [xscale=0.75,yscale=0.75]  {$\operatorname{ord\ \beta }$};

\end{tikzpicture}
\end{center}
\end{remark}

We now turn into studying elementary properties of the previous extension, dealing with the problem of how many properties of Definition \ref{def:graded_Poisson_bracket} remain.

We will find useful the fact that we may compute $\widetilde \sharp_1(\{\alpha, \Theta\})$, where $\alpha \in \Omega^{n-1}_H(M)$ and $\Theta \in \Omega^a_H(M)[1]$. Indeed, we have the following:

\begin{proposition}
\label{prop:sharp_of_bracket}
Let $\Theta \in \Omega^{a}_H(M)[1],$ and $\alpha \in \Omega^{n-1}_H(M).$ Then,
\[
\widetilde \sharp_1(\dd \{\alpha, \Theta\}) = -\pounds_{\sharp_n (\dd \alpha)} \left(\widetilde{\sharp}_1(\dd \Theta)  \right)\,. 
\]
\end{proposition}
Before proving the previous proposition, let us first show that the bracket is closed:
\begin{lemma} Given $\alpha \in \Omega^{n-1}_H(M)$ and $\Theta \in \Omega^{a}(M)[1]$, we have $\{\alpha, \Theta\}\in \Omega^a(M)[1]$.
\end{lemma}

\begin{proof}Let $\alpha \in \Omega^{n-1}_H(M)$ and $\Theta \in \Omega^a_H(M)[1]$. We need to check that $\{\alpha, \Theta\} \in \Omega ^{a}_H(M)[1]$, which is equivalent to $\dd\iota_{\sharp_n(\dd\alpha)}\dd\Theta \in S^{a}$. Now, the last requirement is satisfied if and only if $\iota_{K_1} \dd\iota_{\sharp_n(\dd\alpha)}\dd\Theta = 0$, where $K_1 = (S^1)^{\circ, 1}$. Indeed, let $X$ be a vector field taking values in $K_1$. Then, 
\begin{align*}
    \iota_X \dd \iota_{\sharp_n(\dd\alpha)} \dd\Theta &= \pounds_X \dd\iota_{\sharp_n(\dd \alpha)} \dd \Theta - \dd \iota_X \iota_{\sharp_n (\dd \alpha)} \dd \Theta  =  \pounds_X \dd\iota_{\sharp_n(\dd \alpha)} \dd \Theta\\
    &= \dd  \pounds_X \iota_{\sharp_n(\dd \alpha)} \dd \Theta = \dd \iota_{[X, \sharp_n(\dd \alpha)]} \dd \Theta + \dd \iota_{\sharp_n(\dd \alpha)} \pounds_{X} \dd \Theta \\
    &= \dd \iota_{[X, \sharp_n(\dd \alpha)]} \dd \Theta,
\end{align*}
which is zero, as $[X, \sharp_n(\dd \alpha)]$ takes values in $K_1$, since the graded Dirac structure is integrable.
\end{proof}

\begin{proof}[Proof of Proposition \ref{prop:sharp_of_bracket}] In order to prove the equality, let us check that contraction against an arbitrary form $\beta$ taking values in $S^n$ by both sides of the desired equality yields the same result. Indeed, using the defining property of $\widetilde \sharp_1$ (see Theorem \ref{thm:extension_of_sharp_1}) and that $\dd \Theta \in S^{n+1}$,  we get
\begin{align*}
    \iota_{\widetilde \sharp_1(\dd \{\alpha, \Theta\})} \beta &= (-1)^{n-a} \iota_{\sharp_n(\beta)} \dd \{\alpha, \Theta\} = - \iota_{\sharp_n(\beta)} \dd \iota_{\widetilde \sharp_1(\dd \Theta) }\dd \alpha\\
    &= - \pounds_{\sharp_n(\beta)} \iota_{\widetilde \sharp_1(\dd \Theta)} \dd \alpha + \dd \iota_{\sharp_n(\beta)} \iota_{\widetilde \sharp_1(\dd \Theta)} \dd \alpha\\
    &= - (-1)^{(n-a)}\pounds_{\sharp_n(\beta)} \iota_{ \sharp_n(\dd \alpha)} \dd \Theta + \dd \iota_{\sharp_n(\beta)} \iota_{\widetilde \sharp_1(\dd \Theta)} \dd \alpha\,.
\end{align*}
Now, rewriting the leftmost term employing the formula $\pounds_X \iota_Y \omega  = \iota_{[X, Y]}\omega + \iota_Y \pounds_X \omega$, we get
\begin{align}
    \iota_{\widetilde \sharp_1(\dd \{\alpha, \Theta\})} \beta &= - (-1)^{n-a} \left(\iota_{[\sharp_n(\beta), \sharp_n(\dd \alpha)]} \dd \Theta + \iota_{\sharp_n(\dd \alpha)}  \pounds_{\sharp_n(\beta)} \dd \Theta \right)+\dd \iota_{\sharp_n(\beta)} \iota_{\widetilde \sharp_1(\dd \Theta)} \dd \alpha\,.
    \label{eq:contraction_unsing_Lie_bracket}
\end{align}
Recall that $(S^a, \sharp_a)$ is a graded Dirac structure, so that we have
$[\sharp_n(\beta), \sharp_n(\dd \alpha)] = \sharp_n(\eta)$, where 
\begin{equation}
    \eta = \pounds_{\sharp_n(\beta)} \dd \alpha - \pounds_{\sharp_n(\dd \alpha)} \beta + \frac{1}{2} \dd \left( \iota_{\sharp_n(\dd \alpha)} \beta - \iota_{\sharp_n(\beta)} \dd \alpha\right) = - \pounds_{\sharp_n(\dd \alpha)} \beta\,.
    \label{eq:description_of_eta}
\end{equation}
Using Eq. \eqref{eq:contraction_unsing_Lie_bracket} together with Eq. \eqref{eq:description_of_eta} we obtain
\begin{align*}
    \iota_{\widetilde \sharp_1(\dd \{\alpha, \Theta\})} \beta &=- (-1)^{n-a} \left(\iota_{[\sharp_n(\beta), \sharp_n(\dd \alpha)]} \dd \Theta + \iota_{\sharp_n(\dd \alpha)}  \pounds_{\sharp_n(\beta)} \dd \Theta \right)+\dd \iota_{\sharp_n(\beta)} \iota_{\widetilde \sharp_1(\dd \Theta)} \dd \alpha\\
    &= - (-1)^{n-a} \left( \iota_{\sharp_n(\eta)}  \dd \Theta + \iota_{\sharp_n(\dd \alpha)}  \pounds_{\sharp_n(\beta)} \dd \Theta\right)+\dd \iota_{\sharp_n(\beta)} \iota_{\widetilde \sharp_1(\dd \Theta)} \dd \alpha\\
    &= - (-1)^{n-a} \left( (-1)^{n-a} \iota_{\widetilde \sharp_1(\dd \Theta)} \eta + \iota_{\sharp_n(\dd \alpha)}  \pounds_{\sharp_n(\beta)} \dd \Theta\right) +\dd \iota_{\sharp_n(\beta)} \iota_{\widetilde \sharp_1(\dd \Theta)} \dd \alpha\\
    &=  \iota_{\widetilde \sharp_1(\dd \Theta)} \pounds_{\sharp_n(\dd \alpha)} \beta - (-1)^{n-a} \iota_{\sharp_n(\dd \alpha)} \pounds_{\sharp_n(\beta)} \dd \Theta + \dd \iota_{\sharp_n(\beta)} \iota_{\widetilde \sharp_1(\dd \Theta)} \dd \alpha \,.
\end{align*}
Now it only remains to rewrite $\iota_{\widetilde \sharp_1(\dd \alpha)} \pounds_{\sharp_n(\dd \alpha)} \beta =  -\iota_{\pounds_{\sharp_n(\dd \alpha)}\widetilde \sharp_1(\dd \Theta)} \beta +\pounds_{\sharp_n(\dd \alpha)}\iota_{\widetilde \sharp_1 (\dd \Theta)} \beta$ to get
\begin{align*}
    \iota_{\widetilde \sharp_1(\dd \{\alpha, \Theta\})} \beta =&- \iota_{\pounds_{\sharp_n(\dd \alpha)}\widetilde \sharp_1(\dd \Theta)} \beta + \pounds_{\sharp_n(\dd \alpha)}\iota_{\widetilde \sharp_1 (\dd \Theta)} \beta\\
    & - (-1)^{n-a} \iota_{\sharp_n(\dd \alpha)} \pounds_{\sharp_n(\beta)} \dd \Theta + \dd \iota_{\sharp_n(\beta)} \iota_{\widetilde \sharp_1(\dd \Theta)} \dd \alpha\\
    &= - \iota_{\pounds_{\sharp_n(\dd \alpha)}\widetilde \sharp_1(\dd \Theta)} \beta\,,
\end{align*}
finishing the proof.
\end{proof}

Now we turn into studying its properties under contraction by vector fields:

\begin{proposition}
\label{prop:sharp_1_contraction}
Let $\alpha$ be an $a$-form taking values in $S^a[1]$ and let $X \in \mathfrak{X}(M)$. Then, $\iota_X \alpha \in S^{a-1}[1]$ and $\widetilde{\sharp}_1(\iota_X \alpha) = - \iota_{X} \widetilde \sharp_1(\alpha)$, where contraction is interpreted as $\iota_X (\theta \otimes u) = (\iota_X \theta) \otimes u$.
\end{proposition}
\begin{proof} Indeed, let $\alpha = \theta_1\wedge \cdots \wedge \theta_a$, where $\theta_i \in S^1$, for every $i = 1, \dots, a$. Then
$
    \widetilde \sharp_1(\theta_1 \wedge \cdots \wedge \theta_a) = (-1)^{a+1}\sum_{j = 1}^{a} (-1)^{j+1} \theta_1 \wedge \cdots \wedge \hat{\theta}_j \wedge \cdots \wedge \theta_a \otimes \sharp_1(\theta_j)\,,
$
so that
\begin{align*}
    \iota_X\sharp_1(\theta_1 \wedge \cdots \wedge \theta_a)  = & (-1)^{a+1} \sum_{j = 1}^a \sum_{i< j}\theta_i (X) \theta_1 \wedge \cdots \wedge \hat{\theta}_i \wedge \cdots \wedge \hat{\theta_j} \wedge \cdots \wedge \theta_a \otimes \sharp_1(\theta_j)\\
    &+ (-1)^{a} \sum_{j = 1}^a \sum_{i> j}\theta_i (X) \theta_1 \wedge \cdots \wedge \hat{\theta}_j \wedge \cdots \wedge \hat{\theta_i} \wedge \cdots \wedge \theta_a \otimes \sharp_1(\theta_j)\,.
 \end{align*}
 Now, $\iota_X \left( \theta_1 \wedge \cdots \wedge \theta_a\right) = \sum_{i = 1}^{a} (-1)^{j+1} \theta_i(X) \theta_1 \wedge \cdots \wedge \hat{\theta_i} \wedge \cdots \wedge \theta_a$ and
 \begin{align*}
      \widetilde \sharp_1\left(\iota_X \left(\theta_1 \wedge \cdots \wedge  \theta_a \right)\right) = &(-1)^a \sum_{i = 1}^a \sum_{j > i }\theta_i (X) \theta_1 \wedge \cdots \wedge \hat{\theta}_i \wedge \cdots \wedge \hat{\theta_j} \wedge \cdots \wedge \theta_a \otimes \sharp_1(\theta_j) \\
      &+ (-1)^{a+1}\sum_{i = 1}^a \sum_{j < i }\theta_i (X) \theta_1 \wedge \cdots \wedge \hat{\theta}_j \wedge \cdots \wedge \hat{\theta_i} \wedge \cdots \wedge \theta_a \otimes \sharp_1(\theta_j)\,,
 \end{align*}
 which finishes the proof.
\end{proof}

\begin{theorem}
\label{thm:Properties_of_first_extension}
The extension of the graded Poisson bracket on a graded Dirac manifold defined by the extension $\widetilde \sharp_1$ satisfies the following properties:
\begin{enumerate}[\rm (i)]
    \item It is graded-skew-symmetric.
    \item It satisfies invariance by symmetries: If $\pounds_X \Theta = 0$, for some vector field $X \in \mathfrak{X}(M)$ and differential form $\Theta \in \Omega^a_H(M)[1]$, then $\{\alpha, \iota_X \Theta \} = - \iota_X \{\alpha, \Theta\}$.
    \item It satisfies the (graded) Jacobi identity for $\alpha, \beta \in \Omega_H^{n-1}(M),$ and $\Theta \in \Omega^a_H(M)[1]$ up to an exact term:
    \[\{\alpha, \{\beta, \Theta\}\} + \{\Theta, \{\alpha, \beta\}\} + \{\beta,\{\Theta, \alpha\}\} = \text{exact form}\,.\]
    \item It satisfies the Leibniz identity, for $\alpha \in \Omega^{n-1}_H(M)$, $\Theta_1 \in \Omega^{a_1}_H(M)[1]$ and $\Theta_2 \in \Omega^{a_2}_H(M)[1]$, we have
    \[
    \{\alpha, \Theta_1 \wedge \dd \Theta_2\} = \{\alpha, \Theta_1\} \wedge \dd \Theta_2 + (-1)^{a_1 +1} \dd \Theta_1 \wedge \{\alpha, \Theta_2\}\,.
    \]
\end{enumerate}
\end{theorem}
\begin{proof}
\begin{enumerate}[\rm (i)]
    \item It is satisfied by definition.
    \item Indeed, suppose $\pounds_X \Theta = \dd \iota_X \Theta + \iota_X \dd \Theta = 0$. Then, for every $\alpha \in \Omega^{n-1}_H(M)$, we have 
    \begin{align*}
        \{\alpha, \iota_X\Theta\} &= (-1)^{\deg_H \Theta + 1} \iota_{\widetilde \sharp_1(\dd \iota_X \Theta)} \dd \alpha = (-1)^{\deg_H \Theta} \iota_{\widetilde\sharp_1(\iota_X \dd \Theta)} \dd \alpha \\
        &= -(-1)^{\deg_H \Theta} \iota_{\iota_X \widetilde \sharp_1\left( 
 \dd \Theta\right)} \dd \alpha = -(-1)^{\deg_H \Theta} \iota_X  \iota_{\widetilde \sharp_1\left( 
 \dd \Theta\right)} \dd \alpha \\
 &= - \iota_X \{\alpha, \Theta\}\,,
    \end{align*}
    where in the third equality we have used Proposition \ref{prop:sharp_1_contraction}.
\item Finally, for $\alpha, \beta \in \Omega^{n-1}_H(M)$ and $\Theta \in \Omega^a_H(M)[1]$, we have:
\begin{align*}
   \{\{\alpha, \beta\}, \Theta\}  &= (-1)^{n-1 -a} \iota_{\widetilde \sharp_1(\dd \Theta)} \dd\{ \alpha, \beta\} = - \iota_{\sharp_n(\dd \{\alpha, \beta\})} \dd \Theta  = \iota_{[\sharp_n(\dd \alpha), \sharp_n(\dd \beta)]} \dd \Theta \\
   &= \pounds_{\sharp_n(\dd \alpha)} \iota_{\sharp_n(\dd \beta)} \dd \Theta -\iota_{\sharp_n(\dd \beta)} \pounds_{\sharp_n(\dd \alpha)} \dd \Theta \\
   &=  \dd \iota_{\sharp_n(\dd \alpha)} \iota_{\sharp_n(\dd \beta)} \dd \Theta + \iota_{\sharp_n(\dd \alpha)} \dd \iota_{\sharp_n(\dd \beta)} \dd \Theta - \iota_{\sharp_n(\dd \beta)}\dd  \iota_{\sharp_n(\dd \alpha)} \dd \Theta\\
   &=  \dd \iota_{\sharp_n(\dd \alpha)} \iota_{\sharp_n(\dd \beta)} \dd \Theta + \{\{\Theta, \beta\}, \alpha\} - \{\{\Theta, \alpha\}, \beta\}\,,
\end{align*}
frow which Jacobi identity follows.
\end{enumerate}
\item First notice, as $\dd \Theta_i \in (S^1)^{\wedge({a_i+1})}$ that we have $\dd \Theta_1 \wedge \dd \Theta_2 \in (S^{1})^{\wedge(a_1+a_2+2)}$. Now, we have 
\begin{align*}
    \{\alpha, \Theta_1 \wedge \Theta_2\} &= (-1)^{\deg_H(\Theta_1 \wedge \dd \Theta_2)} \iota_{\widetilde \sharp_1(\dd \Theta_1 \wedge \dd \Theta_2)} \dd \alpha = - \iota_{\sharp_n(\dd \alpha)} (\dd \Theta_1 \wedge \dd \Theta_2)\\
    &= - \{\Theta_1, \alpha\} \wedge \dd \Theta_2 + (-1)^{a_1} \dd \Theta_1 \wedge \{ \Theta_2, \alpha\}\\
    &= \{\alpha, \Theta_1\} \wedge \dd \Theta_2 + (-1)^{a_1 +1} \dd \Theta_1 \wedge \{\alpha, \Theta_2\}\,.
\end{align*}
\end{proof}

\begin{remark}
Notice that we cannot use the values of $\widetilde \sharp_1$ to define a bracket on
\[
\Omega^a_H(M) \otimes \Omega^b_H{M}[1] \xrightarrow{\{\cdot, \cdot\}} \Omega^{a + b - (n - 1)}(M),
\]
for $a < n -1$, as it would be trivially zero.
\end{remark}

\subsection{Last step: Extension of \texorpdfstring{$\sharp_a$}{} to a suitable family of forms}
\label{subsection:last_step}
For the sake of exposition, let us first discuss how to extend the brackets of the previous step to a bracket $\{\alpha, \Theta\}$, where $\alpha \in \Omega^{n-2}_H(M)$ and $\Theta \in \Omega^a_H(M)[1]$. Following the proposed approach, we would need to find an extension
\[
\widetilde{\sharp}_2: (S^1)^{\wedge a} \longrightarrow \bigwedge^{a - 2} M \otimes \bigvee_{n-1} M/K_{n- 1}.
\]
In order to define the new extension by Equation \eqref{eq:bracket_sharp_1}, we need to ask that for compatibility with $\widetilde \sharp_1$, so that we require $\widetilde \sharp_1 = \widetilde \sharp_2\wedge \mathds{1}_1 $ (following the same idea behind the formula in Remark \ref{remark:formula_for_first_extension}), where the exterior product is understood component wise. An analogous computation to that in Lemma \ref{lemma1} shows that the previous requirement is equivalent to $\widetilde \sharp_2$ satisfying
\[
\iota_{\widetilde \sharp_2(\Theta)} \alpha =  \iota_{\widetilde \sharp_1(\Theta)} \alpha\,.
\]
It could happen that such extension does not exists and thus, we need to restrict to the following subspace, which we shall assume to be a subbundle:
\[
 S^a[2] = \left\{\Theta \in (S^1)^{\wedge a}: \exists\, \widetilde {\sharp}_2(\Theta) \in \bigwedge^{a - 2}M \otimes \bigwedge_{n - 1}M/ K_{n - 1} \text{ with } \iota_{\widetilde \sharp_2(\Theta)}\alpha = \iota_{\widetilde \sharp_1(\Theta)} \alpha,\, \forall \alpha \in S^n\right\}. 
\]
Following the notation of the previous section, we obtain a bracket
\[
\Omega^a_H(M) \otimes \Omega^{b}_H(M)[2] \xrightarrow{\{\cdot, \cdot\}_{\widetilde \sharp_2}} \Omega^{a+b- (n-1)}(M)\,,
\]
where \[
\Omega^a_H(M)[2] := \{\alpha \in \Omega^{a}(M): \dd \alpha \in S^{a+1}[2]\}\,, \quad \{\alpha, \Theta\}_{\widetilde \sharp_2} = (-1)^{\deg_H \Theta} \iota_{\widetilde \sharp_2(\Theta)} \dd \alpha\,.
\]

We now proceed iteratively and define \[
 S^a[j] := \left\{\Theta \in (S^1)^{\wedge a}: \exists\, \widetilde {\sharp}_j(\Theta) \in \bigwedge^{a - j}M \otimes \bigwedge_{n + 1 - j}M/ K_{n + 1 - j} \text{ with } \iota_{\widetilde \sharp_j(\Theta)}\alpha = {}\iota_{\widetilde \sharp_{1}(\Theta)} \alpha,\, \forall \alpha \in S^n\right\}\,,
\]
and then choose a map \[
\widetilde \sharp_j : S^a[j] \longrightarrow \bigwedge^{a - j} M \otimes \bigvee_{n+1-j} M \big / K_{n+1-j}
\]
that is compatible with all the previous choices $\widetilde \sharp_1, \widetilde \sharp_2, \dots, \widetilde \sharp_{j-1}$ (compatible meaning $\widetilde \sharp_i = C(i,j)\widetilde \sharp_j \wedge \mathds{1}_{j-i}$, for a suitable constant $C(i,j)$ to be specified later) to then extend the bracket to
\[
\Omega^{n-j}_H(M) \otimes \Omega^{b}_H(M)[j] \xrightarrow{\{\cdot, \cdot\}} \Omega^{a + b - (n - 1)}(M)\,, \quad \{\alpha, \Theta\}_{\widetilde \sharp_j} := (-1)^{\deg_H \Theta} \iota_{\widetilde \sharp_j(\dd \Theta)} \dd \alpha \,.
\]
where 
\[
\Omega^{a}_H(M)[j]= \{\Theta \in \Omega^a(M): \dd \Theta \in S^{a+1}[j]\}
\]
with $a \geq n - j$.

Then, such a choice of vector bundle map satisfying the compatibility condition gives a bracket defined as 
\begin{equation}
\label{eq:bracket_using_extension_sharp_j}
    \{\alpha, \Theta\}_{\widetilde \sharp_j} = (-1)^{\deg_H \Theta} \binom{j + a - ( n+1 ) }{ j - 1} \iota_{\widetilde \sharp_j (\dd \Theta)} \dd \alpha,
\end{equation}
The following diagram represents the pairs of forms on which we can evaluate  the bracket given the extension $\widetilde \sharp_j$, extending by skew-symmetry and using Eq. \eqref{eq:bracket_using_extension_sharp_j}: 
\begin{center}

\tikzset{every picture/.style={line width=0.75pt}} %set default line width to 0.75pt        

\begin{tikzpicture}[x=0.5pt,y=0.5pt,yscale=-1,xscale=1]
%uncomment if require: \path (0,460); %set diagram left start at 0, and has height of 460

%Shape: Square [id:dp1905891204130199] 
\draw  [fill={rgb, 255:red, 217; green, 217; blue, 217 }  ,fill opacity=1 ] (50.85,200.45) -- (250.25,200.45) -- (250.25,399.85) -- (50.85,399.85) -- cycle ;
%Shape: Circle [id:dp7770347199924486] 
\draw  [fill={rgb, 255:red, 0; green, 0; blue, 0 }  ,fill opacity=1 ] (243,199.45) .. controls (243,195.67) and (246.07,192.6) .. (249.85,192.6) .. controls (253.63,192.6) and (256.7,195.67) .. (256.7,199.45) .. controls (256.7,203.23) and (253.63,206.3) .. (249.85,206.3) .. controls (246.07,206.3) and (243,203.23) .. (243,199.45) -- cycle ;
%Shape: Circle [id:dp31327571020367273] 
\draw  [fill={rgb, 255:red, 0; green, 0; blue, 0 }  ,fill opacity=1 ] (243,249.45) .. controls (243,245.67) and (246.07,242.6) .. (249.85,242.6) .. controls (253.63,242.6) and (256.7,245.67) .. (256.7,249.45) .. controls (256.7,253.23) and (253.63,256.3) .. (249.85,256.3) .. controls (246.07,256.3) and (243,253.23) .. (243,249.45) -- cycle ;
%Shape: Circle [id:dp044507100126703225] 
\draw  [fill={rgb, 255:red, 0; green, 0; blue, 0 }  ,fill opacity=1 ] (243,298.45) .. controls (243,294.67) and (246.07,291.6) .. (249.85,291.6) .. controls (253.63,291.6) and (256.7,294.67) .. (256.7,298.45) .. controls (256.7,302.23) and (253.63,305.3) .. (249.85,305.3) .. controls (246.07,305.3) and (243,302.23) .. (243,298.45) -- cycle ;
%Shape: Circle [id:dp6306025955000303] 
\draw  [fill={rgb, 255:red, 0; green, 0; blue, 0 }  ,fill opacity=1 ] (243,349.45) .. controls (243,345.67) and (246.07,342.6) .. (249.85,342.6) .. controls (253.63,342.6) and (256.7,345.67) .. (256.7,349.45) .. controls (256.7,353.23) and (253.63,356.3) .. (249.85,356.3) .. controls (246.07,356.3) and (243,353.23) .. (243,349.45) -- cycle ;
%Shape: Circle [id:dp1141546871262955] 
\draw  [fill={rgb, 255:red, 0; green, 0; blue, 0 }  ,fill opacity=1 ] (93,199.45) .. controls (93,195.67) and (96.07,192.6) .. (99.85,192.6) .. controls (103.63,192.6) and (106.7,195.67) .. (106.7,199.45) .. controls (106.7,203.23) and (103.63,206.3) .. (99.85,206.3) .. controls (96.07,206.3) and (93,203.23) .. (93,199.45) -- cycle ;
%Shape: Circle [id:dp05948980891523892] 
\draw  [fill={rgb, 255:red, 0; green, 0; blue, 0 }  ,fill opacity=1 ] (143,199.45) .. controls (143,195.67) and (146.07,192.6) .. (149.85,192.6) .. controls (153.63,192.6) and (156.7,195.67) .. (156.7,199.45) .. controls (156.7,203.23) and (153.63,206.3) .. (149.85,206.3) .. controls (146.07,206.3) and (143,203.23) .. (143,199.45) -- cycle ;
%Shape: Circle [id:dp8994141232888382] 
\draw  [fill={rgb, 255:red, 0; green, 0; blue, 0 }  ,fill opacity=1 ] (143,249.45) .. controls (143,245.67) and (146.07,242.6) .. (149.85,242.6) .. controls (153.63,242.6) and (156.7,245.67) .. (156.7,249.45) .. controls (156.7,253.23) and (153.63,256.3) .. (149.85,256.3) .. controls (146.07,256.3) and (143,253.23) .. (143,249.45) -- cycle ;
%Shape: Circle [id:dp8483205424096769] 
\draw  [fill={rgb, 255:red, 0; green, 0; blue, 0 }  ,fill opacity=1 ] (193,199.45) .. controls (193,195.67) and (196.07,192.6) .. (199.85,192.6) .. controls (203.63,192.6) and (206.7,195.67) .. (206.7,199.45) .. controls (206.7,203.23) and (203.63,206.3) .. (199.85,206.3) .. controls (196.07,206.3) and (193,203.23) .. (193,199.45) -- cycle ;
%Shape: Circle [id:dp22836979147841419] 
\draw  [fill={rgb, 255:red, 0; green, 0; blue, 0 }  ,fill opacity=1 ] (193,249.45) .. controls (193,245.67) and (196.07,242.6) .. (199.85,242.6) .. controls (203.63,242.6) and (206.7,245.67) .. (206.7,249.45) .. controls (206.7,253.23) and (203.63,256.3) .. (199.85,256.3) .. controls (196.07,256.3) and (193,253.23) .. (193,249.45) -- cycle ;
%Shape: Circle [id:dp24790511846957264] 
\draw  [fill={rgb, 255:red, 0; green, 0; blue, 0 }  ,fill opacity=1 ] (193,298.45) .. controls (193,294.67) and (196.07,291.6) .. (199.85,291.6) .. controls (203.63,291.6) and (206.7,294.67) .. (206.7,298.45) .. controls (206.7,302.23) and (203.63,305.3) .. (199.85,305.3) .. controls (196.07,305.3) and (193,302.23) .. (193,298.45) -- cycle ;
%Straight Lines [id:da2980554772351536] 
\draw    (50.85,200.45) -- (250.25,399.85) ;
%Shape: Circle [id:dp6685074710091325] 
\draw  [fill={rgb, 255:red, 0; green, 0; blue, 0 }  ,fill opacity=1 ] (243,50.45) .. controls (243,46.67) and (246.07,43.6) .. (249.85,43.6) .. controls (253.63,43.6) and (256.7,46.67) .. (256.7,50.45) .. controls (256.7,54.23) and (253.63,57.3) .. (249.85,57.3) .. controls (246.07,57.3) and (243,54.23) .. (243,50.45) -- cycle ;
%Shape: Circle [id:dp18640724259856678] 
\draw  [fill={rgb, 255:red, 0; green, 0; blue, 0 }  ,fill opacity=1 ] (243,99.45) .. controls (243,95.67) and (246.07,92.6) .. (249.85,92.6) .. controls (253.63,92.6) and (256.7,95.67) .. (256.7,99.45) .. controls (256.7,103.23) and (253.63,106.3) .. (249.85,106.3) .. controls (246.07,106.3) and (243,103.23) .. (243,99.45) -- cycle ;
%Shape: Circle [id:dp3890329483197599] 
\draw  [fill={rgb, 255:red, 0; green, 0; blue, 0 }  ,fill opacity=1 ] (243,150.45) .. controls (243,146.67) and (246.07,143.6) .. (249.85,143.6) .. controls (253.63,143.6) and (256.7,146.67) .. (256.7,150.45) .. controls (256.7,154.23) and (253.63,157.3) .. (249.85,157.3) .. controls (246.07,157.3) and (243,154.23) .. (243,150.45) -- cycle ;
%Shape: Circle [id:dp8187589187687252] 
\draw  [fill={rgb, 255:red, 0; green, 0; blue, 0 }  ,fill opacity=1 ] (243,199.85) .. controls (243,196.07) and (246.07,193) .. (249.85,193) .. controls (253.63,193) and (256.7,196.07) .. (256.7,199.85) .. controls (256.7,203.63) and (253.63,206.7) .. (249.85,206.7) .. controls (246.07,206.7) and (243,203.63) .. (243,199.85) -- cycle ;
%Shape: Circle [id:dp8523635313110536] 
\draw  [fill={rgb, 255:red, 0; green, 0; blue, 0 }  ,fill opacity=1 ] (294,199.45) .. controls (294,195.67) and (297.07,192.6) .. (300.85,192.6) .. controls (304.63,192.6) and (307.7,195.67) .. (307.7,199.45) .. controls (307.7,203.23) and (304.63,206.3) .. (300.85,206.3) .. controls (297.07,206.3) and (294,203.23) .. (294,199.45) -- cycle ;
%Shape: Circle [id:dp3406807953849085] 
\draw  [fill={rgb, 255:red, 0; green, 0; blue, 0 }  ,fill opacity=1 ] (343,198.45) .. controls (343,194.67) and (346.07,191.6) .. (349.85,191.6) .. controls (353.63,191.6) and (356.7,194.67) .. (356.7,198.45) .. controls (356.7,202.23) and (353.63,205.3) .. (349.85,205.3) .. controls (346.07,205.3) and (343,202.23) .. (343,198.45) -- cycle ;
%Shape: Circle [id:dp6115802608299312] 
\draw  [fill={rgb, 255:red, 0; green, 0; blue, 0 }  ,fill opacity=1 ] (393,198.45) .. controls (393,194.67) and (396.07,191.6) .. (399.85,191.6) .. controls (403.63,191.6) and (406.7,194.67) .. (406.7,198.45) .. controls (406.7,202.23) and (403.63,205.3) .. (399.85,205.3) .. controls (396.07,205.3) and (393,202.23) .. (393,198.45) -- cycle ;
%Shape: Right Triangle [id:dp25884064404261653] 
\draw  [fill={rgb, 255:red, 217; green, 217; blue, 217 }  ,fill opacity=1 ] (50.85,200.45) -- (250.25,399.85) -- (50.85,399.85) -- cycle ;
%Shape: Circle [id:dp923109190474837] 
\draw  [fill={rgb, 255:red, 0; green, 0; blue, 0 }  ,fill opacity=1 ] (44,200.45) .. controls (44,196.67) and (47.07,193.6) .. (50.85,193.6) .. controls (54.63,193.6) and (57.7,196.67) .. (57.7,200.45) .. controls (57.7,204.23) and (54.63,207.3) .. (50.85,207.3) .. controls (47.07,207.3) and (44,204.23) .. (44,200.45) -- cycle ;
%Shape: Circle [id:dp907016083359832] 
\draw  [fill={rgb, 255:red, 0; green, 0; blue, 0 }  ,fill opacity=1 ] (93,250.45) .. controls (93,246.67) and (96.07,243.6) .. (99.85,243.6) .. controls (103.63,243.6) and (106.7,246.67) .. (106.7,250.45) .. controls (106.7,254.23) and (103.63,257.3) .. (99.85,257.3) .. controls (96.07,257.3) and (93,254.23) .. (93,250.45) -- cycle ;
%Shape: Circle [id:dp8001963309796896] 
\draw  [fill={rgb, 255:red, 0; green, 0; blue, 0 }  ,fill opacity=1 ] (143.7,300.15) .. controls (143.7,296.37) and (146.77,293.3) .. (150.55,293.3) .. controls (154.33,293.3) and (157.4,296.37) .. (157.4,300.15) .. controls (157.4,303.93) and (154.33,307) .. (150.55,307) .. controls (146.77,307) and (143.7,303.93) .. (143.7,300.15) -- cycle ;
%Shape: Circle [id:dp7618375158920645] 
\draw  [fill={rgb, 255:red, 0; green, 0; blue, 0 }  ,fill opacity=1 ] (193,350.45) .. controls (193,346.67) and (196.07,343.6) .. (199.85,343.6) .. controls (203.63,343.6) and (206.7,346.67) .. (206.7,350.45) .. controls (206.7,354.23) and (203.63,357.3) .. (199.85,357.3) .. controls (196.07,357.3) and (193,354.23) .. (193,350.45) -- cycle ;
%Shape: Circle [id:dp6658678933686535] 
\draw  [fill={rgb, 255:red, 0; green, 0; blue, 0 }  ,fill opacity=1 ] (243.4,399.85) .. controls (243.4,396.07) and (246.47,393) .. (250.25,393) .. controls (254.03,393) and (257.1,396.07) .. (257.1,399.85) .. controls (257.1,403.63) and (254.03,406.7) .. (250.25,406.7) .. controls (246.47,406.7) and (243.4,403.63) .. (243.4,399.85) -- cycle ;
%Shape: Axis 2D [id:dp029705070273106315] 
\draw  (51,400.2) -- (401,400.2)(51,50.6) -- (51,400.2) -- cycle (394,395.2) -- (401,400.2) -- (394,405.2) (46,57.6) -- (51,50.6) -- (56,57.6)  ;
%Straight Lines [id:da3034911560932234] 
\draw  [dash pattern={on 4.5pt off 4.5pt}]  (51,400.2) -- (400,49.6) ;
%Shape: Circle [id:dp0277635986446938] 
\draw  [fill={rgb, 255:red, 0; green, 0; blue, 0 }  ,fill opacity=1 ] (193,50.45) .. controls (193,46.67) and (196.07,43.6) .. (199.85,43.6) .. controls (203.63,43.6) and (206.7,46.67) .. (206.7,50.45) .. controls (206.7,54.23) and (203.63,57.3) .. (199.85,57.3) .. controls (196.07,57.3) and (193,54.23) .. (193,50.45) -- cycle ;
%Shape: Circle [id:dp7537837449256739] 
\draw  [fill={rgb, 255:red, 0; green, 0; blue, 0 }  ,fill opacity=1 ] (193,99.45) .. controls (193,95.67) and (196.07,92.6) .. (199.85,92.6) .. controls (203.63,92.6) and (206.7,95.67) .. (206.7,99.45) .. controls (206.7,103.23) and (203.63,106.3) .. (199.85,106.3) .. controls (196.07,106.3) and (193,103.23) .. (193,99.45) -- cycle ;
%Shape: Circle [id:dp33442241812131845] 
\draw  [fill={rgb, 255:red, 0; green, 0; blue, 0 }  ,fill opacity=1 ] (193,150.45) .. controls (193,146.67) and (196.07,143.6) .. (199.85,143.6) .. controls (203.63,143.6) and (206.7,146.67) .. (206.7,150.45) .. controls (206.7,154.23) and (203.63,157.3) .. (199.85,157.3) .. controls (196.07,157.3) and (193,154.23) .. (193,150.45) -- cycle ;
%Shape: Circle [id:dp09078454919635237] 
\draw  [fill={rgb, 255:red, 0; green, 0; blue, 0 }  ,fill opacity=1 ] (143,49.45) .. controls (143,45.67) and (146.07,42.6) .. (149.85,42.6) .. controls (153.63,42.6) and (156.7,45.67) .. (156.7,49.45) .. controls (156.7,53.23) and (153.63,56.3) .. (149.85,56.3) .. controls (146.07,56.3) and (143,53.23) .. (143,49.45) -- cycle ;
%Shape: Circle [id:dp8485910505069212] 
\draw  [fill={rgb, 255:red, 0; green, 0; blue, 0 }  ,fill opacity=1 ] (143,98.45) .. controls (143,94.67) and (146.07,91.6) .. (149.85,91.6) .. controls (153.63,91.6) and (156.7,94.67) .. (156.7,98.45) .. controls (156.7,102.23) and (153.63,105.3) .. (149.85,105.3) .. controls (146.07,105.3) and (143,102.23) .. (143,98.45) -- cycle ;
%Shape: Circle [id:dp4009768556114235] 
\draw  [fill={rgb, 255:red, 0; green, 0; blue, 0 }  ,fill opacity=1 ] (143,149.45) .. controls (143,145.67) and (146.07,142.6) .. (149.85,142.6) .. controls (153.63,142.6) and (156.7,145.67) .. (156.7,149.45) .. controls (156.7,153.23) and (153.63,156.3) .. (149.85,156.3) .. controls (146.07,156.3) and (143,153.23) .. (143,149.45) -- cycle ;
%Shape: Circle [id:dp8447072367054838] 
\draw  [fill={rgb, 255:red, 0; green, 0; blue, 0 }  ,fill opacity=1 ] (294,250.45) .. controls (294,246.67) and (297.07,243.6) .. (300.85,243.6) .. controls (304.63,243.6) and (307.7,246.67) .. (307.7,250.45) .. controls (307.7,254.23) and (304.63,257.3) .. (300.85,257.3) .. controls (297.07,257.3) and (294,254.23) .. (294,250.45) -- cycle ;
%Shape: Circle [id:dp6109004736268389] 
\draw  [fill={rgb, 255:red, 0; green, 0; blue, 0 }  ,fill opacity=1 ] (343,249.45) .. controls (343,245.67) and (346.07,242.6) .. (349.85,242.6) .. controls (353.63,242.6) and (356.7,245.67) .. (356.7,249.45) .. controls (356.7,253.23) and (353.63,256.3) .. (349.85,256.3) .. controls (346.07,256.3) and (343,253.23) .. (343,249.45) -- cycle ;
%Shape: Circle [id:dp09643857037976211] 
\draw  [fill={rgb, 255:red, 0; green, 0; blue, 0 }  ,fill opacity=1 ] (393,249.45) .. controls (393,245.67) and (396.07,242.6) .. (399.85,242.6) .. controls (403.63,242.6) and (406.7,245.67) .. (406.7,249.45) .. controls (406.7,253.23) and (403.63,256.3) .. (399.85,256.3) .. controls (396.07,256.3) and (393,253.23) .. (393,249.45) -- cycle ;
%Shape: Circle [id:dp24194837923856505] 
\draw  [fill={rgb, 255:red, 0; green, 0; blue, 0 }  ,fill opacity=1 ] (293,299.45) .. controls (293,295.67) and (296.07,292.6) .. (299.85,292.6) .. controls (303.63,292.6) and (306.7,295.67) .. (306.7,299.45) .. controls (306.7,303.23) and (303.63,306.3) .. (299.85,306.3) .. controls (296.07,306.3) and (293,303.23) .. (293,299.45) -- cycle ;
%Shape: Circle [id:dp2690543945784374] 
\draw  [fill={rgb, 255:red, 0; green, 0; blue, 0 }  ,fill opacity=1 ] (342,298.45) .. controls (342,294.67) and (345.07,291.6) .. (348.85,291.6) .. controls (352.63,291.6) and (355.7,294.67) .. (355.7,298.45) .. controls (355.7,302.23) and (352.63,305.3) .. (348.85,305.3) .. controls (345.07,305.3) and (342,302.23) .. (342,298.45) -- cycle ;
%Shape: Circle [id:dp27839222472608616] 
\draw  [fill={rgb, 255:red, 0; green, 0; blue, 0 }  ,fill opacity=1 ] (392,298.45) .. controls (392,294.67) and (395.07,291.6) .. (398.85,291.6) .. controls (402.63,291.6) and (405.7,294.67) .. (405.7,298.45) .. controls (405.7,302.23) and (402.63,305.3) .. (398.85,305.3) .. controls (395.07,305.3) and (392,302.23) .. (392,298.45) -- cycle ;

% Text Node
\draw (119,0) node [anchor=north west][inner sep=0.75pt]  [xscale=0.75,yscale=0.75] [align=left] {Brackets defined by $\displaystyle \tilde{\sharp }_{j}$};
% Text Node
\draw (412,393.4) node [anchor=north west][inner sep=0.75pt]  [xscale=0.75,yscale=0.75]  {$\operatorname{ord\ \alpha }$};
% Text Node
\draw (33,25.4) node [anchor=north west][inner sep=0.75pt]  [xscale=0.75,yscale=0.75]  {$\operatorname{ord\ \beta }$};

\end{tikzpicture}
\end{center}

\begin{remark} The reader may notice that the definition of the bundle $S^a[j]$ does not depend on the mappings $\widetilde \sharp_2,\dots, \widetilde \sharp_{j-1}$, and may correctly point out that an extension of $\widetilde \sharp_j$ defined in the whole bundle $S^a[j]$ that is compatible with the previous extensions may not exist. This would mean that an extension of the brackets (following our construction) compatible with the previous extensions is not possible.
\end{remark}

Let us give some initial results concerning these vector subbundles.

\begin{lemma}
\label{lemma:wedge_of_identities}
Let $V$ be a finite dimensional vector space and let $\mathds{1}_a$ denote the identity on $\bigwedge^a V^\ast \otimes \bigwedge^a V$. Then, 
\[
\mathds{1}_a \wedge \mathds{1}_b = \binom{a+b}{a} \mathds{1}_{a+b}\,. 
\]
\end{lemma}
\begin{proof} Let $\{e_i\}$ denote a basis of $V$ so that
\[
\mathds{1}_a = e^{i_1, \dots, i_a} \otimes e_{i_1, \dots, i_a}\,, \quad\mathds{1}_b = e^{i_1, \dots, i_b} \otimes e_{i_1, \dots, i_b}\,.
\]
Hence, 
\begin{align*}
    \mathds{1}_a \wedge \mathds{1}_b &= e^{i_1, \dots, i_a, j_1, \dots, j_b} \otimes e_{i_1, \dots, i_a, j_1, \dots, j_b} = \delta ^{\widetilde I \widetilde J}_{\widetilde \mu}\delta_{\widetilde I\widetilde J}^{\widetilde \nu}e^{\widetilde \mu} \otimes e_{\widetilde \nu}\\
    &= \binom{a+b}{a} \mathds{1}_{a+b}\,,
\end{align*}
which ends the proof.
\end{proof}

\begin{proposition} For $(M, S^a, \sharp_a)$ a regular graded Dirac manifold we have the following:
\begin{enumerate}[\rm (i)]
\label{prop:relation_among_extensions}
    \item 
    \label{prop:unique_extension_j_to_j+1}
    Let
    $
    \widetilde \sharp_j : S^a[j] \longrightarrow \bigwedge^{a - j}M \otimes \bigvee_{n+1-j} M/K_{n+1-j}
    $ be a vector bundle map satisfying $\iota_{\widetilde \sharp_j(\alpha)} \beta = \iota_{\widetilde \sharp_1 (\alpha)} \beta$ for every $\beta \in S^{n}$. Then,
    there exists an unique extension $\widetilde \sharp_i : S^a[j] \longrightarrow \bigwedge^{a-i} M \otimes \bigvee_{n+1-i} M \big / K_{n+1-i}$, for $1 \leq i < j$, satisfying the following two properties
    \begin{enumerate}[\rm (a)]
        \item For every $\gamma \in S^{n+1-i}$ and $\alpha \in S^{a}[j]$, we have
        \[
    \iota_{\widetilde \sharp_i(\alpha)} \gamma = \binom{j-1 }{j - i}\iota_{\widetilde \sharp_j(\alpha)} \gamma\,;
    \]  
        \item For every $\beta \in S^n$, and $\alpha \in S^a[j]$, we must have
        $\iota_{\widetilde \sharp_i(\alpha)} \beta = \iota_{\widetilde \sharp_1(\alpha)} \beta$\,.
    \end{enumerate}

    \item \label{prop:relation_among_extensions_i}$S^a[{j+1}] \subseteq S^a[j],$ for $1 \leq j < n$ and $1 \leq a$. 
    \item \label{prop:relation_among_extenions_ii} For $X \in \mathfrak{X}(M)$, we have $\iota_X : S^a[j] \longrightarrow S^{a-1}[j-1]$, for every $a > n$ and $2 \leq j \leq n$. 
    %Furthermore, for $j = n$, we have $\iota_X\colon S^a[n] \longrightarrow S^{a-1} [n]$.
\end{enumerate}
\end{proposition}
\begin{proof}\begin{enumerate}[\rm (i)]
    \item Let $\alpha \in S^a[j]$. Then, we can define a map $S^{n+1-i} \longrightarrow \bigwedge^{a-i} M$ by $\gamma \mapsto \binom{j-1}{j-i}^{-1} \iota_{\widetilde \sharp_{j}(\alpha)} \gamma$. Dualizing this map we obtain an element $\widetilde \sharp_i(\alpha) \in \bigwedge^{a-i} M \otimes \bigvee_{n+1-i} M \big / K_{n+1-i}$. This clearly satisfies the first property. In order to check the second, let us recall that the compatibility conditions reads as 
    \[\widetilde \sharp_1 = \widetilde \sharp_i  \wedge \mathds{1}_{i-1}\,,\]
    which holds in this case since we have $\widetilde \sharp_i = \binom{j-1}{j-i}^{-1}\widetilde \sharp_j \wedge \mathds{1}_{j-i}$ (by Lemma \ref{lemma:wedge_of_identities}), $\mathds{1}_{j-i} \wedge \mathds{1}_{i-1} = \binom{j-1}{j-i} \mathds{1}_{j-1}$, and $\widetilde \sharp_1 = \widetilde \sharp_j \wedge \mathds{1}_{j-1}$.
    \item The construction applied in item \ref{prop:unique_extension_j_to_j+1} clearly proves the inclusion.
    \item Let $\alpha \in S^a[j]$, so that there exists $\widetilde \sharp_j(\alpha) \in \bigwedge^{a-j} M \otimes \bigvee_{n+1 - j} M \big / K_{n+1-j}$ such that $\iota_{\widetilde \sharp_j(\alpha)} \beta = \iota_{\widetilde \sharp_1(\alpha)} \beta$, for every $\beta \in S^n$. Define $\widetilde \sharp_{j-1}(\iota_X \alpha)  := -\iota_X \widetilde \sharp_{j-1}(\alpha) - (-1)^{a-j}\widetilde \sharp_j(\alpha) \wedge X$, where $\widetilde \sharp_{j-1}(\alpha)$ denotes the unique extension of $\widetilde \sharp_j(\alpha)$ from item \ref{prop:unique_extension_j_to_j+1}. Then, for arbitrary $\beta \in S^n$ we have
    \begin{align*}
        \iota_{\widetilde \sharp_j(\iota_x \alpha)} \beta &= -\iota_{\iota_X \widetilde \sharp_j(\alpha)} \beta - (-1)^{a-j} \iota_{\widetilde \sharp_j(\alpha) \wedge X} \beta = -\iota_{\iota_X \widetilde \sharp_j(\alpha)} \beta - (-1)^{n+1+a} \iota_{\widetilde \sharp_j(\alpha)} \iota_X \beta \\
        &= -\iota_X \left( \iota_{\widetilde \sharp_j(\alpha)} \beta\right) + (-1)^{n+1+a} \iota_{\widetilde \sharp_{j}(\alpha)} \iota_X \beta - (-1)^{n+1+a} \iota_{\widetilde \sharp_j(\alpha)} \iota_X \beta\\
        &= -\iota_X\left( \iota_{\widetilde \sharp_1(\alpha)} \beta\right) = -\iota_{\iota_X \widetilde \sharp_1(\alpha)} \beta \\
        &= \iota_{\widetilde \sharp_1(\iota_X \alpha)} \beta\,,
    \end{align*}
    where in the last equality we have used Proposition \ref{prop:sharp_1_contraction}.
    
\end{enumerate}
\end{proof}

Let $m = \dim M $. Then, Proposition \ref{prop:relation_among_extensions} gives the following picture of the subbundles $S^a[j]$ and the relations between them: 
\[\begin{tikzcd}[ column sep={6.3em,between origins},
  row sep={3em,between origins},]
	{S^{m}[n]} & {S^{m}[n-1]} & \cdots & {S^{m}[j]} & {S^{m}[j-1]} & \cdots & {S^{m}[1]} \\
	\vdots & \vdots & \ddots & \vdots & \vdots & \ddots & \vdots \\
	{S^a[n]} & {S^a[n-1]} & \cdots & {S^a[j]} & {S^{a}[j-1]} & \cdots & {S^a[1]} \\
	{S^{a-1}[n]} & {S^{a-1}[n-1]} & \cdots & {S^{a-1}[j]} & {S^{a-1}[j-1]} & \cdots & {S^{a-1}[1]} \\
	\vdots & \vdots & \ddots & \vdots & \vdots & \ddots & \vdots \\
	{S^{n+1}[n]} & {S^{n+1}[n-1]} & \cdots & {S^{n+1}[j]} & {S^{n+1}[j-1]} && {S^{n+1}[1]}
	\arrow[hook, from=1-1, to=1-2]
	\arrow[hook, from=1-2, to=1-3]
	\arrow[hook, from=1-3, to=1-4]
	\arrow[hook, from=1-4, to=1-5]
	\arrow[hook, from=1-5, to=1-6]
	\arrow[hook, from=1-6, to=1-7]
	\arrow["{\iota_{\T M}}"{pos=0.6}, from=2-1, to=3-2]
	\arrow["{\iota_{\T M}}"{pos=0.6}, from=2-3, to=3-4]
	\arrow["{\iota_{\T M}}"{pos=0.6}, from=2-4, to=3-5]
	\arrow["{\iota_{\T M}}"', from=2-6, to=3-7]
	\arrow[hook, from=3-1, to=3-2]
	\arrow["{\iota_{\T M}}"{pos=0.6}, from=3-1, to=4-2]
	\arrow[hook, from=3-2, to=3-3]
	\arrow[hook, from=3-3, to=3-4]
	\arrow["{\iota_{\T M}}"{pos=0.6}, from=3-3, to=4-4]
	\arrow[hook, from=3-4, to=3-5]
	\arrow["{\iota_{\T M}}"{pos=0.6}, from=3-4, to=4-5]
	\arrow[hook, from=3-5, to=3-6]
	\arrow[hook, from=3-6, to=3-7]
	\arrow["{\iota_{\T M}}"{pos=0.6}, from=3-6, to=4-7]
	\arrow[hook, from=4-1, to=4-2]
	\arrow[hook, from=4-2, to=4-3]
	\arrow[hook, from=4-3, to=4-4]
	\arrow[hook, from=4-4, to=4-5]
	\arrow[hook, from=4-5, to=4-6]
	\arrow[hook, from=4-6, to=4-7]
	\arrow["{\iota_{\T M }}"{pos=0.6}, from=5-1, to=6-2]
	\arrow["{\iota_{\T M}}"{pos=0.7}, from=5-3, to=6-4]
	\arrow["{\iota_{\T M}}"{pos=0.7}, from=5-4, to=6-5]
	\arrow["{\iota_{\T M}}"{pos=0.6}, from=5-6, to=6-7]
	\arrow[hook, from=6-1, to=6-2]
	\arrow[hook, from=6-2, to=6-3]
	\arrow[hook, from=6-3, to=6-4]
	\arrow[hook, from=6-4, to=6-5]
	\arrow[hook, from=6-5, to=6-7]
\end{tikzcd}.\]
We have horizontal arrows representing inclusions from item \ref{prop:relation_among_extensions_i}, and diagonal arrows representing contraction by vector fields from item \ref{prop:relation_among_extenions_ii}.

Therefore, by Proposition \ref{prop:relation_among_extensions}, we can only hope of finding a global extension if we define a final vector bundle map
$
\widetilde \sharp_n: S^a[n] \longrightarrow \bigwedge^{a-n}M \otimes \T M/K_1
$,
for $a\geq n$. Then, we may obtain the associated extensions as
\[
\widetilde \sharp_i = \binom{n-1}{n-i}^{-1} \widetilde \sharp_n \wedge \mathds{1}_{n-i}
\]
and defining the space of Hamiltonian forms as: 
\[
\Omega^a_H(M)[n] := \begin{cases} 
\Omega^a_H(M)\, & a \leq n - 1\,\\
\{\alpha \in \Omega^a(M) \colon \dd \alpha \in S^{a+1}[n]\}
\, & a \geq n\end{cases}\,,
\]
we may define the following bracket
\begin{align*}
   &\Omega^a(M)[n] \otimes \Omega^b(M)[n] \xrightarrow{\{\cdot, \cdot\}_{\widetilde \sharp_n}} \Omega^{a+ n - (n+1)}(M)\,,\\\\
    &\{\Theta_1, \Theta_2\}_{\widetilde \sharp_n} := 
\begin{cases}
    \{\Theta_1,\Theta_2\} & \text{if } \deg \Theta_1, \deg \Theta_2 \leq n-1\\\\
    (-1)^{\deg_H \Theta_2}{\iota_{\widetilde \sharp_n(\dd \Theta_2)} \dd \Theta_1} & \text{if } \deg \Theta_2 \geq n
\end{cases}\,.
\end{align*}
%\textcolor{red}{Lo números binomiales es para seguir el caso de la ecuación \ref{equation_condition_of_extension}.}\\

\begin{remark} We will give dynamical interpretation of all possible extensions (compatible with a fibered structure and restricted to certain subfamily) in the following section (see Theorem \ref{thm:Characterization_of_fibered_Extensions}).
\end{remark}

Suppose that a final extension is chosen, $\widetilde \sharp_n\colon S^a[n] \longrightarrow \bigwedge^{a-n} M \otimes \T M \big / K_1$. We now turn into the discussion of how many properties of Definition \ref{def:graded_Poisson_bracket} remain valid for the bracket $\{\cdot, \cdot\}_{\widetilde \sharp_n}$, depending on the particular extension $\widetilde \sharp_n$ chosen. Unlike the bracket induced from $\widetilde \sharp_1$, these extensions do not satisfy all the properties of Definition \ref{def:graded_Poisson_bracket}. Nevertheless, we can give sufficient conditions on the map $\widetilde \sharp_n$ so that these are satisfied:

\begin{theorem}
\label{thm_properties_final_extension}
Let $\widetilde \sharp_n : S^a[n] \longrightarrow \bigwedge^{n-a}M \otimes \T M / K_1$ be an extension of $\sharp_n$ and denote by $\{\cdot, \cdot\}_{\widetilde \sharp_n}$ the induced bracket on $\Omega^a_H(M)[n]$. Then it satisfies the following properties
\begin{enumerate}[\rm (i)]
    \item We have $\{\alpha, \Theta\}_{\widetilde \sharp_n} = -(-1)^{\deg_H \alpha \deg_H \Theta}\{\Theta, \alpha\}_{\widetilde \sharp_n}$ for $\alpha \in \Omega^a_H(M)$ usual Hamiltonian form and $\Theta \in \Omega^{b}_H(M)[n]$. 
    \item If $\widetilde \sharp_n$ satisfies $\widetilde \sharp_{n-1}(\iota_X \alpha)  := -\iota_X \widetilde \sharp_{n-1}(\alpha) - (-1)^{a-n}\widetilde \sharp_n(\alpha) \wedge X$, where $\widetilde \sharp_{n-1}$ is the unique extension given by item \ref{prop:unique_extension_j_to_j+1} of Proposition \ref{prop:relation_among_extensions}, it satisfies invariance by symmetries: For $X \in \mathfrak{X}(M)$, if $\pounds_X \Theta = 0$, then $\{\Theta_1, \iota_X \Theta_2\} = - \iota_X \{\Theta_1, \Theta_2\}$, for every $\Theta_1 \in \Omega^{a}(M)[n]$, with $a \geq 2$.
    \item If $\widetilde \sharp_n \left (\dd\{\alpha, \Theta_1\}\right  ) =  -\pounds_{{ \sharp_n(\dd \alpha)}} \widetilde \sharp_n(\dd \Theta_1)$, then it satisfies a Jacobi identity:
    \[
    \{\alpha, \{\Theta_1, \Theta_2\}\} + \{\Theta_2,\{ \alpha, \Theta_1\}\} + \{\Theta_1, \{\Theta_2, \alpha\}\} = \text{exact form},
    \]
    for $\alpha \in \Omega^{n-1}_H(M)$, $\Theta_1 \in \Omega^{a_1}(M)[n]$ and $\Theta_2 \in \Omega^{a_2}(M)[n]$.
\end{enumerate}
\end{theorem}
\begin{proof}
\begin{enumerate}[\rm (i)]
    \item Graded skew-symmetry is given by definition.
    \item Follows from a similar computation to that done in item (iii) of Proposition \ref{prop:relation_among_extensions}.
    \item Again, follows from a similar computation from item (iii) of Theorem \ref{thm:Properties_of_first_extension}.
\end{enumerate}
\end{proof}

%-------------------------------------------
% Section 4: Dynamics on fibered Dirac manifolds
%-------------------------------------------
\section{Dynamics on fibered graded Dirac manifolds}
\label{section:dynamics}

In this section we apply the constructions of the previous section to define dynamics and study the evolution of arbitrary forms.

Throughout this section we fix an arbitrary fibered graded Dirac manifold $\tau\colon M \longrightarrow X$ (see Definition \ref{def:fibered_graded_Dirac}).

\subsection{Fibered extensions, Hamiltonians and Hamilton--De Donder--Weyl equations}
\label{subsection:Hamilton--DeDonder--Weyl_eq}
%\textcolor{red}{Hay que pedir que las extensiones satisfagan una compatilibilidad con la estructura fibrada: deben tomar valores verticales.}

Let us first give a description of dynamics on graded Dirac manifolds, defining the equations of classical field theories in terms of Poisson brackets. We first need to define what we understand by a Hamiltonian, the object that will define the dynamics.

\begin{Def}[Hamiltonian]
\label{def:Hamiltonian}
A \textbf{Hamiltonian} is an $n$-form $\mathcal{H} \in \Omega^{n}_H(M)[n]$ such that
$
-\widetilde \sharp_1(\dd \mathcal{H}) + \mathds{1}_n
$
is a semi-basic form that satisfies $\iota_{-\widetilde \sharp_1(\dd \mathcal{H}) + \mathds{1}_n} \omega = \omega$, for every semi-basic $n$-form $\omega$.
\end{Def}

\begin{remark} Notice that the notion of semi-basic (as well as that of $r$-horizontal) can be translated to forms taking values in vector bundles.
\end{remark}

\begin{remark}
\label{remark:preserve_semi-basic-forms}
The last condition $\iota_{-\widetilde \sharp_1 (\dd \mathcal{H}) +\mathds{1}_n} \omega = \omega$ is equivalent to $\iota_{\widetilde \sharp_1(\dd \mathcal{H})} \omega = 0$, for every semi-basic $n$-form $\omega$.
\end{remark}

\begin{remark} Notice that $S^n$ contains all semi-basic $n$-forms over $X$, so that $K_n$ is contained in the space of $1$-vertical multivectors, namely those multivectors taking values in $\left \langle  \ker \dd \tau \wedge \bigvee_{p-1} M\right \rangle$. Hence, it still makes sense to consider vertical valued sections of $\bigvee_n M / K_n$.
\end{remark}

\begin{example}
\label{ex:multisymplectic_Hamiltonian}
Let $\pi \colon Y \longrightarrow X$ be a fibered manifold. Let $\bigwedge^n_2 Y$ be endowed with its canonical multisymplectic structure and reduce it, following Example \ref{example:multisymplectic_field_theory}, to a fibered graded Dirac (in fact, Poisson) structure on $\tau \colon \bigwedge^n_2 Y \big / \bigwedge^n_1 Y \longrightarrow X$. Then, given a section $h\colon \bigwedge^n_2 Y \big / \bigwedge^n_1 Y  \longrightarrow \bigwedge^n_2 Y$, we have that $\mathcal{H} := h^\ast \Theta$ is a Hamiltonian for the corresponding graded Poisson structure on $ M= \bigwedge^n_2 Y \big / \bigwedge^n_1 Y$. Indeed, in canonical coordinates, we have $S^n = \langle \dd^n x, \dd y^i \wedge \dd ^{n-1}x_\mu, \dd p^\mu_i \wedge \dd^{n-1}x_\mu \rangle$ and
\[
\sharp_n(\dd^n x) = 0\,, \quad \sharp_n(\dd y^i \wedge \dd^{n-1}x_\mu) = -\pdv{p^\mu_i}\,, \quad \sharp_n(\dd p^\mu_i \wedge \dd^{n-1}x_\mu) = \pdv{y^i}\,.
\]
An immediate computation shows $\mathcal{H}= H \dd^n x - p^\mu_i \dd y^i \wedge \dd^{n-1}x_\mu$ so that 
\[
\dd \mathcal{H} =  \pdv{H}{y^i} \dd y^i \wedge \dd^n x + \pdv{H}{p^\mu_i}\dd p^\mu_i\wedge \dd ^n x - \dd p^\mu_i \wedge \dd y^i \wedge \dd^{n-1}x_\mu\,.
\]
We have (for complete details regarding this computations we refer to Section \ref{section:Applications}) 
\begin{align*}
   \widetilde \sharp_1 (\dd \mathcal{H}) =& \frac{1}{n} \pdv{H}{y^i}\dd^n x \otimes \pdv{p^\mu_i} \wedge \pdv{^{n-1}}{^{n-1} x_\mu} -\pdv{H}{p^\mu_i} \dd^n x \otimes \pdv{y^i}\wedge \pdv{^{n-1}}{^{n-1} x_\mu} + \\& \dd y^i \wedge \dd^{n-1} x_\mu \otimes\pdv{y^i}\wedge \pdv{^{n-1}}{^{n-1} x_\mu} + \frac{1}{n} \dd p^\mu_i \wedge \dd^{n-1}x_\mu \otimes   \pdv{p^\mu_i} \wedge \pdv{^{n-1}}{^{n-1} x_\mu}\,.
\end{align*}
Therefore, $-\widetilde \sharp_1 (\dd \mathcal{H}) + \mathds{1}_n$ takes the expression
\[
-\frac{1}{n} \pdv{H}{y^i}\dd^n x \otimes \pdv{p^\mu_i} \wedge \pdv{^{n-1}}{^{n-1} x_\mu} +\pdv{H}{p^\mu_i} \dd^n x \otimes \pdv{y^i}\wedge \pdv{^{n-1}}{^{n-1} x_\mu} +\dd^n x \otimes \pdv{^n}{^n x}\,,
\]
which is clearly semi-basic and satisfies $\iota_{-\widetilde \sharp_1 (\dd \mathcal{H}) + \mathds{1}_n} \dd^n x = \dd^n x$, proving that it is a Hamiltonian in the sense of Definition \ref{def:Hamiltonian}.
\end{example}

Now, since $-\widetilde \sharp_1(\dd \mathcal{H}) + \mathds{1}_n$ is semi-basic over $X$, this allows us to define dynamics as follows:

\begin{Def}[Hamilton-De Donder-Weyl equations for $\mathcal{H}$] We define the \textbf{Hamilton--De Donder--Weyl equations} for a section $\psi: X \longrightarrow M$ as 
\[
\psi^\ast(\dd \alpha) = \left(\dd \alpha + \{\alpha, \mathcal{H}\} \right) \circ \psi = \left(\iota_{- \widetilde \sharp_1(\dd \mathcal{H} ) + \mathds{1}_n} \dd \alpha \right) \circ \psi\,,
\]
for every $\alpha \in \Omega^{n - 1}_H(M)$. 
\end{Def}

\begin{remark} Notice that the condition $\iota_{-\widetilde \sharp_1(\dd \mathcal{H}) + \mathds{1}_n} \omega = \omega$ guarantees that for basic $(n-1)$-forms $\omega$, we obtain $\dd \omega + \{\omega, \mathcal{H}\} = \dd\omega$, so that the evolution of basic forms is always determined, as it should be.
\end{remark}
\begin{remark} Notice the immediate similarity with the equation
$\dv{f}{t} = \pdv{f}{t} + \{f, H\}$
from classical mechanics.
\end{remark}

\begin{example}
    Let us write the equations for the Hamiltonian of Example \ref{ex:multisymplectic_Hamiltonian}. Notice that we have 
\begin{align*}
     \{\alpha, \mathcal{H}\} = (-1)^{\deg_H \mathcal{H}} \iota_{\widetilde \sharp_1(\dd \mathcal{H})} \dd \alpha= -\iota_{\widetilde \sharp_1( \dd \mathcal{H})} \dd \alpha\,.
\end{align*}
Taking $\alpha$ to be $y^i \dd^{n-1}x_\mu$ and $p^\mu_i \dd^{n-1}x_\mu$ we obtain
\begin{align*}
    \{y^i \dd^{n-1}x_\mu, \mathcal{H}\} = \pdv{H}{p^\mu_i} \dd^n x - \dd y^i \wedge \dd^{n-1}x_\mu\,, \quad
    \{ p^\mu_i \dd^{n-1}x_\mu, \mathcal{H}\} = - \pdv{H}{y^i} \dd^n x - \dd p^\mu_i \wedge \dd^{n-1}x_\mu \,,
\end{align*}
so that the equations determined by $\mathcal{H}$ for a section $\psi \colon X \longrightarrow M$ given locally by $y^i(x)$ and $p^\mu_i(x)$ read 
\begin{align*}
    \pdv{y^i}{x^\mu} \dd^n x &=   \psi^\ast \left( \dd y^i \wedge \dd^{n-1}x_\mu \right) = \left( \dd y^i \wedge \dd^{n-1}x_\mu + \{ y^i \dd^{n-1}x_\mu, \mathcal{}H\}\right) \circ \psi = \pdv{H}{p^\mu_i} \dd^n x\\
    \pdv{p^\mu_i}{x^\mu} \dd^n x &=   \psi^\ast \left( \dd p^\mu_i \wedge \dd^{n-1}x_\mu \right) = \left( \dd p^\mu_i \wedge \dd^{n-1}x_\mu + \{ p^\mu_i\dd^{n-1}x_\mu, \mathcal{H}\}\right) \circ \psi =- \pdv{H}{y^i} \dd^n x\,,
\end{align*}
recovering the well known Hamilton--De Donder--Weyl equations of classical field theories.
\end{example}

In many cases, it will be interesting to have a linearized version of a solution. This corresponds to the following:

\begin{Def}[Linear solution] A \textbf{linear solution} is a connection (given by a horizontal lift) on $\tau\colon M \longrightarrow X$, $h$, such that \[
h^\ast(\dd \alpha) =  \dd \alpha + \{ \alpha, \mathcal{H}\}_{\widetilde \sharp _1},
\]
for every $\alpha \in \Omega^{n-1}_H(M).$
\end{Def}

Regarding the extensions studied in Subsection \ref{subsection:last_step}, we have the following:

\begin{proposition}
\label{prop:extensions_preverve_Hamiltonian}
Let $\mathcal{H}$ be a Hamiltonian on a regular graded Dirac manifold of order $n$, $(M, S^a, \sharp_a).$ Then, given any extension
\[
\widetilde \sharp_n: S^{n+1}[n] \longrightarrow \bigwedge^{a- n} M \otimes \T M/K_1,
\]
such that $\widetilde \sharp_n(\dd \mathcal{H})$ takes vertical values, we have that $-\widetilde \sharp_n(\dd \mathcal{H}) + \mathds{1}_1$ is semi-basic and that $\iota_{-\widetilde \sharp_n(\dd \mathcal{H}) + \mathds{1}_1} \omega = \omega$, for every semi-basic $1$-form $\omega$.
\end{proposition}
\begin{proof} Let $\alpha \in S^n$ and $u \in\bigvee_{n-1} M$. Then, since $-\widetilde \sharp_1(\dd \mathcal{H}) + \mathds{1}_{n-1}$ is semi-basic, $-\iota_{\widetilde \sharp_n( \dd \mathcal{H})} \alpha +  \alpha$ is semi-basic (as $\iota_{\widetilde \sharp_1(\dd \mathcal{H})} \alpha = \iota_{\widetilde \sharp_n(\dd \mathcal{H})} \alpha$). Contracting by $u$ preserves the semi-basic character, so that
\begin{align*} 
    \iota_u \left( -\iota_{\widetilde \sharp_n( \dd \mathcal{H})} \alpha +  \alpha\right) &= -\iota_u \iota_{\widetilde \sharp_n(\dd \mathcal{H}) } \alpha+ \iota_u \alpha\\
    &= -\iota_{\widetilde \sharp_n (\dd \mathcal{H})} \iota_u \alpha + \iota_u\alpha +  \iota_{\iota_{\widetilde \sharp_n(\dd\mathcal{H})} u } \alpha
\end{align*}
is semi-basic. The proof will be finished once we show that the last term $\iota_{\iota_{\widetilde \sharp_1(\dd \mathcal{H})} u } \alpha$ also is. Indeed, since $\sharp_1(\dd \mathcal{H})$ takes vertical values, ${\iota_{\widetilde \sharp_1(\dd \mathcal{H})} u }$ is at least $1$-vertical and, together with the fact that $\alpha$ is at most $(n - 1)$-horizontal, implies that $\iota_{\iota_{\widetilde \sharp_1(\dd \mathcal{H})} u } \alpha$ is semi-basic, which finishes the proof. The fact that $\iota_{- \widetilde \sharp_n(\dd \mathcal{H}) +  \mathds{1}} \omega = \omega$, for every semi-basic $1$-form $\omega$ follows from $\widetilde \sharp_n$ taking vertical values.
\end{proof}

We now deal with the problem of interpreting the possible extensions of the brackets in term of the dynamics determined by a Hamiltonian.

Fix a Hamiltonian $\mathcal{H}$ and an extension $\widetilde \sharp_n$ that satisfies the hypotheses of Proposition \ref{prop:extensions_preverve_Hamiltonian}. Then, we get a map
\[S^1 \longrightarrow \tau^\ast \left(\T^\ast X \right)\,, \quad \theta \mapsto \theta -\iota_{\widetilde \sharp_n(\dd \mathcal{H})} \theta\,.\]
that is the identity on semi-basic forms. Dualizing this map, we obtain a vector bundle map
\[
\gamma_{\mathcal{H}}: \tau^\ast \left(\T X \right) \longrightarrow \T M \big / K_1 \,,
\]
such that $\tau_\ast \circ \gamma_{\mathcal{H}} = \operatorname{id}_{\T X}$ (recall that the projection is well defined modulo $K_p$). This map can be thought of as the horizontal lift of a connection ``up to $K_1$''. This connection will solve the Hamilton--De Donder--Weyl equations. We have the following result formalizing the previous idea:
\begin{proposition} Let $h$ be the horizontal lift of an Ehresmann connection on $\tau\colon M \longrightarrow X$. If its horizontal lift modulo $K_1$ coincides with $\gamma_\mathcal{H}$, we have that $h$ solves the Hamilton--de Donder--Weyl equations determined by $\mathcal{H}$.
\end{proposition}
\begin{proof} Since the horizontal lift of $h$ coincides with $\gamma_\mathcal{H}$ modulo $K_1$, we have $h^\ast \theta = \theta - \iota_{\widetilde \sharp_n(\dd \mathcal{H})} \theta$, for every $\theta \in S^1$. Therefore, we can write 
\begin{align*}
    h^\ast \left( \theta_1 \wedge \cdots \wedge \theta_a\right) &=  h^\ast \theta_1 \wedge \cdots \wedge h^\ast \theta_a \\
    &= (\theta_1 - \iota_{\widetilde \sharp_n(\dd \mathcal{H})}\theta_1) \wedge \cdots \wedge(\theta_a - \iota_{\widetilde \sharp_n(\dd \mathcal{H})}\theta_a)\,,
\end{align*}
for $\theta_1, \dots, \theta_a \in S^1$. Hence, since $\widetilde \sharp_n$ takes vertical values, if every form $\theta_1, \dots, \theta_a$ but possibly one is semi-basic, we may write
\[
h^\ast \left( \theta_1 \wedge \cdots \wedge \theta_a\right) = \theta_1 \wedge \cdots \wedge \theta_a - \iota_{\widetilde \sharp_n(\dd \mathcal{H})} \left( \theta_1 \wedge \cdots \wedge \theta_a\right)\,.
\]
In particular, letting $a  = n$ and $\alpha \in S^n$ we have $h^\ast \alpha = \alpha - \iota_{\widetilde \sharp_n(\dd \mathcal{H})} \alpha = \alpha - \iota_{\widetilde \sharp_1(\dd \mathcal{H})}$, which clearly implies $h^\ast (\dd \alpha) = \dd \alpha + \{ \alpha, \mathcal{H}\}$, for every $\alpha \in \Omega^{n-1}_H(M)$, finishing the proof.
\end{proof}
Notice that the previous computation works for any form $\alpha \in S^a$, and for arbitrary $a$, so that we have that suitable extensions of the brackets give a possible evolution of arbitrary forms, namely:
\begin{theorem} 
\label{thm:evolution_of_arbitrary_forms}
Let $\widetilde \sharp_n$ be an extension that {takes vertical values}, $\mathcal{H} \in \Omega^{n}(M)[n]$ be a Hamiltonian and $h$ be a horizontal lift of an Eheresmann connection whose vertical lift is $\gamma_\mathcal{H}$ modulo $K_1$. Then, $h$ solves the Hamilton--De Donder--Weyl equations determined by $\mathcal{H}$ and we have
\[
h_\mathcal{H}^\ast (\dd \alpha) = \dd \alpha +\{\alpha, \mathcal{H}\}_{\widetilde \sharp_n},
\]
for any $\alpha \in \Omega^{a}_H(M),$ with $0 \leq a \leq n - 1$.
\end{theorem}
\begin{proof}
    From the previous computation, for $\alpha \in \Omega^a_H(M)$ we have 
    \[
    h^\ast(\dd \alpha) = \dd \alpha - \iota_{\widetilde \sharp_n(\dd \mathcal{H})}\dd \alpha = \dd \alpha + \{\alpha, \mathcal{H}\}\,.
    \]
\end{proof}

We would like to end this discussion with the following result, which gives the relation between the possible extensions and the possible choices of solutions. Before stating and proving it first notice that the space of Hamiltonians is affine, and so is the space of Ehresmann connections. Furthermore, if we are exclusively interested in studying the equations defined by Hamiltonians, we can restrict to those forms whose exterior derivative takes values in the following subbundle
\[
\Gamma := \{\alpha \in S^{n+1}[n] : -\widetilde\sharp_1(\alpha) + \mathds{1}_n = \text{semi-basic}\}\,,
\]
and study extensions of the $\sharp$ maps defined only on $\Gamma$.

\begin{theorem}\label{thm:Characterization_of_fibered_Extensions} Suppose that, locally, $S^n = \langle \dd \alpha_1, \dots, \dd \alpha_k\rangle$, for certain Hamiltonian forms $\alpha_1, \dots, \alpha_k$. Then, there is a bijective correspondence between vertical valued extensions $\widetilde \sharp_n$ defined on $\Gamma$ and the family of affine maps 
\[
\gamma:\{\text{Hamiltonians} \} \longrightarrow \{\text{Horizontal lifts modulo } K_1\}
\]
satisfying the following two properties:
\begin{enumerate}[\rm (i)]
    \item If $\dd \mathcal{H}_1|_x = \dd \mathcal{H}_2|_x$, then $\gamma_{\mathcal{H}_1}|_x = \gamma_{\mathcal{H}_2}|_x$.
    \item If $h$ is a connection such that its horizontal lift coincides with $\gamma(\mathcal{H})$ modulo $K_1$, then $h$ solves the Hamilton--De Donder--Weyl equations. 
\end{enumerate}
\end{theorem}

\begin{proof} First, given the vertical extension of $\widetilde \sharp_n$, we obtain the corresponding map $\gamma$ as before, that is, noticing that the map
\[
\dd f  \mapsto \dd f + \{f, \mathcal{H}\} = \dd f - \iota_{\widetilde \sharp_n(\dd \mathcal{H})} \dd f
\]
defines a map $S^1 \longrightarrow S^1$ that takes values in semi-basic forms. Dualizing we obtain the corresponding horizontal lift. Conversely, let $\gamma$ be such a map and, for every Hamiltonian $\mathcal{H} \in \Omega^n_H(M)[n]$ define $\widetilde \sharp_n(\dd \mathcal{H})$ to be the $(1,1)$ tensor field defined by 
\[
\iota_{\widetilde \sharp_n(\dd \mathcal{H})} \dd f = \dd f - h^\ast \dd f\,, 
\]
where $h$ is the horizontal lift associated to the connection $\gamma(\mathcal{H})$. The first requiered propery on $\gamma$ implies that it defines an extension 
\[
\widetilde \sharp_n : \Gamma \longrightarrow \T^\ast M \otimes \T M / K_1\,,
\]
and the fact that $\gamma(\dd \mathcal{H})$ solves the Hamilton--De Donder--Weyl equations, together with the assuption that $S^n$ is generated by the exterior differential of Hamiltonian forms, implies that 
\[
\widetilde \sharp_1 = \widetilde \sharp_n \wedge\mathds{1}_{n-1} \,,
\]
so that it defines an extension of the brackets.
\end{proof}
\subsection{Differential forms with defined evolution}

We have shown in Subsection \ref{subsection:Hamilton--DeDonder--Weyl_eq} that extensions of brackets define linear solutions for every possible choice of Hamiltonian in a way such that the expression $\dd \alpha + \{\alpha, \mathcal{H}\}_{\widetilde \sharp_n}$ gives a possible evolution of the form $\alpha$. Thus, if we identify the subfamily of Hamiltonian forms for which $\dd \alpha + \{\alpha, \mathcal{H}\}_{\widetilde \sharp_n}$ is independent of the particular extension $\widetilde \sharp_n$, we are identifying the family of forms that have defined evolution, for every solution of Hamilton--De Donder--Weyl equations:

\begin{theorem}
\label{thm:defined_evolution_implies_Hamiltonian}
Let $(\tau:M \longrightarrow X, S^a, \sharp_a)$ be a fibered graded Dirac manifold. Let $\mathcal{H} \in \Omega^{n}_H(M)[n]$ be a Hamiltonian and let $\gamma \in \Omega^a(M)$ be a form such that there exists a semi-basic form $\beta \in \Omega^{a+1}(M)$  with
\[h^\ast (\dd \gamma ) = \beta,\]
for every linear solution of Hamilton-De Donder-Weyl equations $h$. Then, $\gamma \in \Omega^a_H(M)$ and it satisfies:
\[\gamma \wedge \varepsilon \in \Omega^{n-1}_H(M),\] for every basic closed form $\varepsilon \in \Omega^{n - 1 - a}(M).$
\end{theorem}
\begin{proof} Let us begin with the case $a = n-1.$ 
\begin{lemma}
\label{lemma:Lemma_3}
Let $\gamma \in \Omega^{n-1}(M)$ be an $(n-1)$-form and $\beta \in \Omega^n(M)$be a semi-basic $n$-form such that 
$h^\ast (\dd \gamma) = \beta,$ for every solution of the Hamilton--de Donder--Weyl equations determined by a Hamiltonian $\mathcal{H} \in\Omega^n_H(M)[n]$. Then $\gamma \in \Omega^{n-1}_H(M)$.
\end{lemma}
\begin{proof} The space of linear solutions is defined as those connections (given by a horizontal lift) $h$ such that 
$h^\ast (\dd \alpha) = \dd \alpha + \{\alpha, \mathcal{H}\}$, for every $\alpha \in \Omega^{n-1}_H(M)$. Let us define, for each $\alpha \in S^n$
\[
\Phi_\alpha: \{\text{horizontal lifts}\} \rightarrow \bigwedge^n \T^\ast X\,, \quad h \mapsto  h^\ast(\alpha)- \left(\alpha - \iota_{\widetilde \sharp_n(\dd \mathcal{H})}\alpha\right)\,.
\]
Hence, if $h$ is a horizontal lift such that $\Phi_\alpha(h) = 0$, for every $\alpha \in S^n$, we have that $h$ solves the equations defined by $ \mathcal{H}$. Also notice that each $\Phi_\alpha$ is an affine map. Now let $\gamma \in \Omega^{n-1}(M)$ and $\beta \in \Omega^n(M)$ be such that $h^\ast(\dd \gamma) = \beta$, for every solution of the Hamilton--De Donder--Weyl equations. Define 
\[\varphi: \{\text{horizontal lifts}\} \longrightarrow \bigwedge^n \T^\ast X\,, \quad h \mapsto h^\ast(\dd \alpha) - \beta\,.\]
Then $\varphi$ is identically zero on $\bigcap_{\alpha \in S^a} \Phi_\alpha^{-1}(0)$. Since $\varphi$ is affine and $\bigcap_{\alpha \in S^a} \Phi_\alpha^{-1}(0) \neq \varnothing$, we have that $\varphi \in \langle\Phi_\alpha: \alpha \in S^n \rangle $, which means that there is an $\alpha \in S^n$ such that $\varphi = \Phi_\alpha$. This immediately implies $\dd \gamma = \alpha \in S^n$, proving that $\gamma$ is Hamiltonian.
\end{proof}
Now, for arbitrary $a$, let $\gamma \in \Omega^a(M)$ and $\beta \in \Omega^{a+1}(M)$ semi-basic such that $h^\ast (\dd \gamma) = \beta$, for every solution of the Hamilton--De Donder--Weyl equations. Let $\varepsilon \in \Omega^{n - 1 - a}(M)$ be closed and basic. Then, 
\[h^\ast(\dd(\gamma  \wedge \varepsilon)) = \beta \wedge \varepsilon\] so that, by Lemma \ref{lemma:Lemma_3}, $\gamma  \wedge \varepsilon \in \Omega^{n-1}_H(M)$, for every closed and basic $(n-1-a)$-form $\varepsilon$. Hence, $\dd \gamma \wedge \varepsilon \in S^{n}$, for every basic (not necessarily closed) $(n-1-a)$-form $\varepsilon$. The proof will be finished once we show that $\dd \alpha \in S^{a+1}$, which will be a consequence of the following Lemma:

\begin{lemma} Let $V$ be a finite dimensional vector space, together with a surjective linear mapping $\tau: V \longrightarrow X$, with $\dim X = n$. Let $S^n \in \bigwedge^n V^\ast$ be a subspace consisting of $(n-1)$-horizontal forms and such that $\tau^\ast \left(\bigwedge^n X^\ast\right) \subseteq S^n$. If $\gamma  \in \bigwedge^a V^\ast$ is such that $\gamma  \wedge \varepsilon \in S^n$, for every basic $\varepsilon$, then $\gamma  \in \left\langle \iota_{\bigwedge^p V} S^n \right\rangle$.
\end{lemma}
\begin{proof}
Let $e_\mu, e_i$ denote a basis of $V$ adapted to $\tau$, so that $\tau(e_\mu) = e_\mu$ and $\tau(e_i) = 0$. Suppose 
\[S^n = \langle \tau^\ast\omega, \alpha_1, \dots, \alpha_k\rangle\,\]
where $\omega$ is the volume form on $X$ satisfying $\omega(e_1, \dots, e_n) = 1$, and $\alpha_k = A_{j i}^ \mu \omega_\mu \wedge e^i$, where $\omega_\mu = \iota_{e_\mu} \omega$. Without loss of generality, we can assume that $\gamma$ takes the expression
\[\gamma = \gamma^{\mu_1, \dots, \mu_{n+1-a}}  \omega_{\mu_1, \dots, \mu_{n+1-a}} \wedge e^i\,,\]
where $\omega_{\mu_1, \dots, \mu_{n+1-a}} = \iota_{e_{\mu_1} \wedge \cdots \wedge e_{\mu_{n+1-a}}} \omega$. The condition that $\varepsilon \wedge \gamma \in S^n$, for each basic $(n-a)$ $\varepsilon$ reads using the basis as $\gamma^{\mu_1, \dots, \mu_{n+1-a}} \omega_{\mu_{n+1-a}} \wedge e^i \in S^n$ which is equivalent to the existence of certain constants $B^{j, \mu_1, \dots, \mu_{n-a}}$ such that
\[\gamma^{\mu_1, \dots, \mu_{n+1-a}} \omega_{\mu_{n+1-a}} \wedge e^i = B^{j, \mu_1, \dots, \mu_{n-a}} \alpha = B^{j, \mu_1, \dots, \mu_{n-a}} A^{\mu_{n+1-a}}_{ji} \omega_{\mu_{n+1-a}} \wedge e^i\,.\]
Let $u^j: = B^{j, \mu_1, \dots, \mu_{n-a}} e_{\mu_1} \wedge \dots \wedge e_{\mu_{n-a}}$. Then, the previous conditions reads as $\gamma = \iota_{u^j} \alpha$, which finishes the proof.
\end{proof}
\end{proof}

In fact:
\begin{theorem} 
\label{thm:determined_evolution}
If $\alpha \in \Omega^a_H(M)$ is a Hamiltonian form such that \[\alpha \wedge \varepsilon \in \Omega^{n-1}_H(M),\] for every closed basic $\varepsilon \in \Omega^{n - 1 - a}(M),$ then $
\{\alpha, \mathcal{H}\}_{\widetilde \sharp_n}
$
does not depend on the choice of $\widetilde \sharp_n,$ for any Hamiltonian $\mathcal{H}$, if $\widetilde \sharp_n$ takes vertical values.
\end{theorem}

\begin{proof} Observe that 
\[\{\alpha \wedge \varepsilon, \mathcal{H}\} = \{\alpha \wedge \varepsilon, \mathcal{H}\}_{\widetilde \sharp_n}  = \{\alpha, \mathcal{H}\}_{\widetilde \sharp_n} \wedge \varepsilon,\]
for any closed and basic $\varepsilon \in \Omega^{n- 1 - a}(M).$ Since an $(a+1)$-horizontal form is determined by its wedge product with all possible $(n-1 - a)-$basic forms, we conclude that \[\{\alpha, \mathcal{H}\}_{\widetilde \sharp_n}\]
is independent of the choice of $\widetilde \sharp_n,$ as it only depends on $\widetilde \sharp_1.$
\end{proof}

The previous two results motivate the following definition:
\begin{Def}[Special Hamiltonian form] A Hamiltonian form $\alpha \in \Omega^a_H(M)$ is called \textbf{special Hamiltonian} if $\alpha \wedge \varepsilon \in \Omega^{n-1}_H(M)$, for every basic closed $(n - 1 - a)$-form $\varepsilon$. We denote by $\widetilde \Omega^a_H(M)$ the space of special Hamiltonian forms.
\end{Def}

\begin{remark} This gives a method of finding arbitrary $(a-1)-$horizontal $a$-forms that have a defined evolution. In particular, it allows us to look for forms that are closed on arbitrary solutions of the classical field equations, since we have clear candidates. We will apply this method in the following section.
\end{remark}

To end this section, we give a sufficient condition for the space of special Hamiltonian forms to define a subalgebra. This condition will be manually checked in the examples treated at the end of the paper.

\begin{theorem}
\label{thm_sufficient_condition_subalgebra}
    Let $(\tau: M \longrightarrow X, S^a, \sharp_a)$ be a fibered graded Dirac manifold. Suppose that for each special Hamiltonian form $\alpha \in \widetilde \Omega^{a}_H(M)$ there is a multivector field $U_\alpha \in \mathfrak{X}^{n-a}(M)$ such that $\sharp_{a+b+1}(\dd \alpha \wedge \varepsilon) = C(\alpha, \varepsilon)\iota_{\varepsilon} U_\alpha + K_{n - (a+b)}$, for each closed and basic $b$-form $\varepsilon$ and for certain nonzero constant $C(\alpha, \varepsilon)$. Further suppose that $\sharp_{a+1}(\dd \alpha) = U_\alpha + K_{n-a}$. Then, special Hamiltonian forms define a subalgebra.
\end{theorem}

\begin{proof} Let $\alpha \in \widetilde \Omega^a_H(M)$ and $\beta \in \widetilde \Omega^b_H(M)$ be special Hamiltonian forms. Let us show that $\{\alpha, \beta\} \wedge \varepsilon$ takes values in $\Omega^{n-1}_H(M)$, for every closed and basic $\varepsilon$ of suitable degree. First notice that for a closed and basic $1$-form we have the following:
\begin{align*}
    \{\alpha \wedge \varepsilon, \beta\} &= (-1)^{\deg_H \beta} \iota_{\sharp_{b+1}(\dd \beta)} (\dd \alpha \wedge \varepsilon)\\
    &= (-1)^{\deg_H \beta} \left( \iota_{\sharp_{b+1}(\dd \beta)}  \dd \alpha\right) \wedge \varepsilon + (-1)^{\deg_H \beta + \deg \alpha} \iota_{\iota_{\varepsilon} \sharp_{b+1} (\dd \beta)} \dd \alpha
    \\&= \{\alpha, \beta\} \wedge \varepsilon  + (-1)^{\deg \alpha} C(\beta, \varepsilon)  \iota_{\sharp_{b+2}(\dd \beta \wedge \varepsilon)} \dd \alpha \\
    &= \{\alpha, \beta\} \wedge \varepsilon - (-1)^{\deg_H \beta + \deg \alpha}C(\beta, \varepsilon) \{\alpha, \beta \wedge \varepsilon\}\,.
\end{align*}
Notice that, since both $\alpha \wedge \varepsilon$ and $\beta \wedge \varepsilon$ are Hamiltonian, and $C(\beta, \varepsilon)$ is nonzero, so is $\{\alpha, \beta\}\wedge \varepsilon$. Iterating this argument and using that both $\alpha \wedge \varepsilon_1 \wedge \cdots \wedge \varepsilon_j$ and $\alpha \wedge \varepsilon_1 \wedge \cdots \wedge \varepsilon_j$ are Hamiltonian, we conclude that $\{\alpha,\beta\} \wedge \varepsilon$ is Hamiltonian, for arbitrary closed and basic $\varepsilon$, which shows that special Hamiltonian forms define a subalgebra.
\end{proof}

\begin{remark} This sufficient condition is easy to apply in practice, as we will show in Section \ref{section:Applications}. Roughly speaking, we need to find a multivector field associated to each form that recovers the multivector fields through interior product (up to kernel) of all possible exterior products with closed and basic forms.
\end{remark}

%-------------------------------------------
% Section 5: Applications to classical field theory
%-------------------------------------------
\section{Applications to classical field theories}
\label{section:Applications}

Here we discuss to several applications of the theory developed in the previous sections: the study of conserved quantities for regular and almost regular Lagrangians.

\subsection{Hamiltonian formalism for regular field theories}

Let us go through the construction of bracket extensions in the case of the reduced covariant formalism of classical field theories explained in Example \ref{example:multisymplectic_field_theory}. Recall that the graded Dirac structure is given by the map
\[
 \sharp_n(\dd^n x) = 0\,, \quad \sharp_n(\dd y^i \wedge \dd^{n-1}x_\mu) = -\pdv{p^\mu_i}\,, \quad
    \sharp_n(\dd p^\mu_i \wedge \dd^{n-1}x_\mu) =  \pdv{y^i}\,.
\]
The $\sharp_1$ (modulo $K_n$) map is given by 
\[
\sharp_1 (\dd x^\mu) = K_n\,, \quad \sharp_1 (\dd y^i) = \frac{(-1)^{n}}{n} \pdv{p^\mu_i} \wedge \pdv{^{n-1}}{^{n-1}x_\mu} + K_n \,, \quad \sharp_1(\dd p^\mu_i) = (-1)^{n-1} \pdv{y^i} \wedge \pdv{^{n-1}}{^{n-1}x_\mu} + K_n\,,
\]
and $\widetilde \sharp_1$ is obtained applying the anti-derivation rule, as mentioned in Remark \ref{remark:anti_derivation}. Let us determine the space $S^{n+1}[n]$, which will be the space on which the exterior derivative of Hamiltonians may take values. 
\begin{theorem} $S^{n+1}[n]$ is generated by the following set of forms
\[
\dd y^i \wedge \dd^{n} x \,, \quad \dd p^\mu_i \wedge \dd^n x\,, \quad \dd p^\mu_i \wedge \dd y^j \wedge \dd ^{n-1}x_\mu\,.
\]
\end{theorem}
\begin{proof} Recall that $S^{n+1}[n]$ is defined as the space of forms $\Theta \in \bigwedge^{n+1} M$ for which there exists $\widetilde \sharp_n(\Theta) \in \T^\ast M \otimes \T M$ with $\iota_{\widetilde \sharp_n(\Theta)} \alpha = \iota_{\widetilde \sharp_1(\Theta)} \alpha =  \iota_{\widetilde \sharp_n(\alpha)} \Theta$, for every $\alpha \in S^n$. Notice that $\alpha$ is $(n-1)$-horizontal, so that $\iota_{\widetilde \sharp_1(\Theta)} \alpha$ is $(n-2)$-horizontal. Therefore, $\iota_{\widetilde \sharp_n(\alpha)} \Theta$ is $(n-2)$-horizontal and, since $\sharp_n(\alpha)$ is vertical, $\Theta$ must be $(n-2)$-horizontal. In Table \ref{table:contractions}, the reader may find represented all possible values of $\iota_{\sharp_n(\alpha)} \Theta$ for all possible candidates $\Theta \in \bigwedge^{n+1}M$:

\setlength{\tabcolsep}{2pt} % Default value: 6pt

\begin{table}[ht]
\caption{Values of $\iota_{\sharp_n(\alpha)} \Theta$, for $(n+1)$-forms $\Theta$ (on the left) and forms $\alpha \in S^n$ (on top)}
\centering 
%\rule{0pt}{4ex}
\renewcommand{\arraystretch}{2}
  \begin{tabular}{c|ccc}
    & $\dd y^i \wedge \dd^{n-1}x_\mu$& $\dd p^\mu_i \wedge \dd^{n-1}x_\mu$  & $\dd^n x$\\
    \hline
    $ \dd y^j \wedge \dd^n x$ & $0$ & $\delta^j_i \dd^n x$ & $ 0$\\
     \hline
    $\dd p^\alpha_j \wedge \dd^n x$ & $\delta^\alpha_\mu \delta^i_j \dd^n x$ & $0$ & $0$ \\
     \hline
     $\dd y^j \wedge \dd p^\alpha_k \wedge \dd^{n-1}x_\beta$ & $\delta^\alpha_\mu \delta^i_k \dd y^j \wedge \dd^{n-1}x _\beta$& $\delta^j_i \dd p^\alpha_k \wedge \dd^{n-1}x_\beta$ & $0$ \\
     \hline
     $\dd y^j\wedge \dd y^k \wedge \dd^{n-1}x_\alpha $ & $0$ & $\delta^j_i \dd y^k \wedge \dd^{n-1}x_\alpha - \delta^k_i \dd y^j \wedge \dd^{n-1}x_\alpha$ & $0$\\
     \hline
     $\dd p^\alpha_j \wedge \dd p^\beta_k \wedge \dd^{n-1}x_\gamma$ &
     $\delta^\alpha_\mu \delta^i_j \dd p^\beta _k \wedge \dd^{n-1}x_\gamma - \delta^\beta_\mu \delta^i_k \dd p^k_j \wedge \dd^{n-1}x_\gamma$ & $0$ & $0$
     
  \end{tabular}
  \label{table:contractions}
\end{table}
\noindent An elementary computation shows that any linear combination of the forms on the left, $\Theta$, for which there is a $(1,1)$-tensor, $\widetilde \sharp_n(\Theta)$ satisfying $\iota_{\widetilde \sharp_n(\Theta)} \alpha = \iota_{\sharp_n(\alpha)} \Theta$, must have null coefficients on the last two elements ($\dd y^j\wedge \dd y^k \wedge \dd^{n-1}x_\alpha $ and $\dd p^\alpha_j \wedge \dd p^\beta_k \wedge \dd^{n-1}x_\gamma$). Furthermore, the first two elements are clearly in $S^{n+1}[n]$, as we may choose 
\[
\widetilde \sharp_n(\dd y^j \wedge \dd^n x) := \frac{1}{n} \dd x^\mu \otimes \pdv{p^\mu_i}\,, \quad \widetilde \sharp_n(\dd p^\alpha_j \wedge \dd^n x):= \dd x^\alpha \otimes \pdv{y^j}\,.
\]
Hence, it only remains to study which possible linear combinations $\Theta = A^{k \beta}_{j \alpha} \dd y^j \wedge \dd p^\alpha_k \wedge \dd^{n-1}x_\beta$ belong to $S ^{n+1}[n]$. We look for a $(1,1)$-tensor field
\begin{align*}
    \widetilde \sharp_n(\Theta) = \,&  \dd x^\mu \otimes \left(C^\nu_\mu \pdv{x^\nu} + D^j_\mu\pdv{y^j} + E^\nu_{\mu j} \pdv{p^\nu_j}\right)+\\
    & \dd y^i \otimes \left( C^\nu_i \pdv{x^\nu} + D^j_i\pdv{y^j} + E^\nu_{i j} \pdv{p^\nu_j}\right)+\\
    &  \dd p^\mu_i \otimes \left( C^{i\nu}_\mu \pdv{x^\nu} + D^{ij}_\mu\pdv{y^j} + E^{i\nu}_{\mu j} \pdv{p^\nu_j}\right)
\end{align*}
satisfying $\iota_{\widetilde \sharp_1(\Theta)} \alpha = \iota_{\sharp_n(\alpha)} \Theta$, for every $\alpha \in S^n$. We have:
\begin{enumerate}[\rm (i)]
    \item {For $\alpha = \dd^n x$}, we have $\iota_{\sharp_n(\alpha)} \Theta = 0$, which implies $C^\mu_\mu = C^\nu_i = C^{i \nu}_{\mu} = 0$.
    \item {For $\alpha = \dd y^i \wedge \dd^{n-1}x_\mu$}, we have $\iota_{\sharp_n(\alpha)} \Theta = - A^{i \beta}_{j \mu} \dd y^j \wedge \dd^{n-1}x_\mu$, which implies that $C^\mu_\nu = 0$ and 
    \[
    - A^{i \beta}_{j\mu} \dd y^j \wedge \dd ^{n-1}x_\beta = D^i_\mu \dd^n x + D^i_j \dd y^j \wedge \dd ^{n-1}x_\mu + D^{ji}_\nu \dd p^\nu_j \wedge \dd^{n-1}x_\mu\,,
    \]
    from which we obtain $D^i_\mu = 0$, $D^{ji}_\nu = 0$ and $- A^{i \beta}_{j \mu} = \delta^\beta_\mu D^{i}_j$ and hence $\Theta = - D^k_j \dd y^j \wedge \dd p^\alpha_k \wedge \dd^{n-1}x_\alpha$.
\end{enumerate}
 It only remains to show that we can define $\widetilde \sharp_n(\Theta)$ for $\Theta = \dd p^\mu_i \wedge \dd y^j \wedge \dd^{n-1}x_\mu$. Indeed, notice that we may choose 
 \[\widetilde \sharp_n(\Theta) := \dd y^j \otimes \pdv{y^k} + \dd p^\mu_k \otimes \pdv{p^\mu_k}\,,\]
 finishing the proof.
\end{proof}

Now let us look the possible expression of a Hamiltonian (Definition \ref{def:Hamiltonian}).

\begin{theorem} Let $\mathcal{H} \in \Omega^{n}_H(M)[n]$ be a Hamiltonian. Then, up to a closed term it has the following local expression
\[\mathcal{H} = H \dd^n x  - p^\mu_i  \dd y^i \wedge \dd^{n-1}x_\mu\,,\]
for certain function $H(x^\mu, y^i, p^\mu_i)$.
\end{theorem}

\begin{proof} Recall that $\widetilde \sharp_1(\dd \mathcal{H}) + \mathds{1}_n$ has to be a semi basic form. A quick computation shows that this implies
\[
\dd \mathcal{H} = A_i \dd y^i \wedge \dd^nx + B^i_\mu \dd p^\mu_i \wedge \dd ^n x - \dd p^\mu_i \wedge \dd y^i \wedge \dd^{n-1}x_\mu\,,
\]
for certain functions $A_i$ and $ B^i_\mu$. Closedness of the previous form implies that there is a function $H$ such that $A_i = \pdv{H}{y^i}$ and $B^i_\mu = \pdv{H}{p^\mu_i}$, which finishes the proof.
\end{proof}

Finally, let us study the algebra of special Hamiltonian forms:

\begin{theorem}
\label{thm_ideal}
With the canonical graded Poisson structure on $M = \bigwedge^n_2 Y \big / \bigwedge^n_1, Y$ special Hamiltonian forms close an subalgebra and:
\begin{enumerate}[\rm(i)]
    \item On degree $a = 0, \dots, n-2$, if $\alpha \in \widetilde \Omega^a_H( M)$, then
    \[
    \dd \alpha \in \langle \dd^a x_{\widetilde \mu}\,,  \dd y^i \wedge \dd^{a-1} x_{\widetilde \nu}\, \colon  |\widetilde \mu| = n - a\,, |\widetilde \nu| = n - a - 1\rangle\,.
    \]
    \item On degree $a = n-1$, $\widetilde \Omega^{n-1}_H(M) = \Omega^{n-1}_H(M)$.
\end{enumerate}
\end{theorem}
\begin{proof} The second claim is clear. It only remains to show that if $a = 1, \dots, n-1$ and $\alpha \in \widetilde \Omega^{a-1}_H(M)$, then $\dd \alpha$ is necessarily in $ \langle \dd^a x_{\widetilde \mu}\,,  \dd y^i \wedge \dd^{a-1} x_{\widetilde \nu}\, \colon  |\widetilde \mu| = n - a\,, |\widetilde \nu| = n - a - 1\rangle$. Suppose that 
\[
\dd \alpha = A^{\widetilde \mu} \dd^a x_{\widetilde \mu} + A_i^{\widetilde \mu} \dd y^i \wedge \dd^{a-1} x_{\widetilde \mu} + A^{i\widetilde \mu} \dd p^\mu_i \wedge \dd ^{a-1} x_{\widetilde \mu \mu}\,.
\]
Imposing the condition of being special Hamiltonian with $\varepsilon = \dd^{n-a}x^{\widetilde\nu}$, we get that
\[
\varepsilon \wedge \dd \alpha = A^{\widetilde \nu} \dd^n x + (-1)^{n-a} A^{\widetilde\mu}_i\delta^{ \widetilde \nu \mu}_{\widetilde \mu} \dd y^i \wedge \dd^{n-1}x_\mu + (-1)^{n-1} A^{i \widetilde \mu} \delta^{\widetilde \nu \nu}_{\widetilde \mu \mu}\dd p^\mu_i \wedge \dd^{n-1}x_\nu\in S^n\,.\]

The first two terms clearly satisfy this condition. The last one can be expanded as
\[
(-1)^{n-a}A^{i \widetilde \mu} \delta^{\widetilde \nu}_{\widetilde \mu} \dd p^\mu_i \wedge \dd^{n-1}x_\mu + (-1)^{n-a}\sum_{\mu \neq \nu}A^{i \widetilde \mu} \delta^{\widetilde \nu \nu}_{\widetilde \mu \mu}\dd p^\mu_i \wedge \dd^{n-1}x_\nu\,,
\]
and since $a < n$, we can choose $\widetilde \nu$ so that $\delta^{\widetilde\nu}_{\widetilde\mu} = 0$ so we obtain $A^{i \widetilde \mu} = 0$, proving the result. In order to show that they close a subalgebra, we use Theorem \ref{thm_sufficient_condition_subalgebra} and assign to the form $\dd y^i \wedge \dd^{a-1}x_{\mu_1, \dots, \mu_{n+1-a}}$ the multivector
\[
U = \frac{1}{n+1-a} \sum_{j = 1}^{n+1-a} (-1)^{j-1}\pdv{p^{\mu_j}_i} \wedge \pdv{^{n-a}}{^{n-a}x^{\mu_1, \dots, \widehat{\mu_j}, \dots, \mu_{n+1-a}}}\,,
\]
which satisfies the required property. Indeed, contracting (without loss of generality) by $\varepsilon = \dd x^{\mu_1}$ on $U$ gives 
\[
\iota_{\varepsilon} U = \frac{1}{n+1-a} \sum_{j= 2}^{n+1-a} (-1)^j \pdv{p^{\mu_j}_i} \wedge \pdv{^{n-a}}{^{n-a}x^{\mu_2, \dots, \widehat{\mu_j}, \dots, \mu_{n+1-a}}}\,,
\]
which is a Hamiltonian multivector field for $\varepsilon \wedge \dd y^i \wedge \dd^{a-1}x_{\mu_1, \dots, \mu_{n+1-a}}$ up to a constant.
\end{proof}

\subsection{Bracket formalism for almost regular Lagrangians: Yang-Mills theories}

Let $\pi \colon Y \longrightarrow X$ denote a fibered manifold over an $n$-dimensional base with coordinates $(x^\mu, y^i)$, and let $\mathcal{L} : J^1 Y \longrightarrow \bigwedge^n X$ be a first order Lagrangian density (for higher order theories we refer to \cite{lr1985}), locally given by $\mathcal{L} = L(x^\mu, y^i, z^i_\mu) \dd^n x$.
The multisymplectic formalism is obtained by pullbacking the natural multisymplectic form on $\bigwedge^n_2 Y$ by the Legendre transformation, 
$
\Leg_\mathcal{L}: J^1 \pi \longrightarrow \bigwedge^n_2 Y\,
$
, which is locally given by 
\[
\Leg_\mathcal{L}^\ast p = L - \pdv{L}{z^i_\mu}z^i_\mu\,, \quad \Leg_\mathcal{L} ^\ast p^\mu_i = \pdv{L}{z^\mu_i}\,.
\]
The reduced Legendre transformation is obtained by composing with the projection $\bigwedge^n_2 Y \longrightarrow \bigwedge^n_2 Y \big /\bigwedge^n_1 Y$.
% https://q.uiver.app/#q=WzAsMyxbMCwwLCJKXjEgXFxwaSJdLFsyLDAsIlxcYmlnd2VkZ2Vebl8yIFkiXSxbMiwxLCJcXGJpZ3dlZGdlXm5fMiBZIFxcYmlnIC8gXFxiaWd3ZWRnZV5uXzEgWSJdLFsxLDJdLFswLDEsIlxcTGVnX1xcbWF0aGNhbHtMfSJdLFswLDIsIlxcbGVnX1xcbWF0aGNhbHtMfSIsMix7ImN1cnZlIjoyfV1d
\[\begin{tikzcd}
	{J^1 \pi} && {\bigwedge^n_2 Y} \\
	&& {\bigwedge^n_2 Y \big / \bigwedge^n_1 Y}
	\arrow["{\Leg_\mathcal{L}}", from=1-1, to=1-3]
	\arrow["{\leg_\mathcal{L}}"', curve={height=12pt}, from=1-1, to=2-3]
	\arrow[from=1-3, to=2-3]
\end{tikzcd}.\]
When the reduced Legendre transformation defines a local diffeomorphism, we say that $\mathcal{L}$ is a \textit{regular} Lagrangian. In this case, locally, we are dealing with the case treated in the previous subsection. However, a great deal of field theories (defined by the corresponding Lagrangian density) are not regular. Nevertheless, when the image $\leg_\mathcal{L} \left(J^1 \pi \right)$ is a submanifold for which $\leg_\mathcal{L} \colon J^1 \pi \longrightarrow \Ima \leg_\mathcal{L}$ defines a submersion, then we may apply the constructions of the present paper:

\begin{enumerate}[\rm (i)]
    \item Use the construction of Theorem \ref{thm:uniqueness_of_pullback_and_pushforward} to induce a graded Dirac structure on $\Ima \leg_\mathcal{L}$ from the graded Dirac structure on $\bigwedge^n_2 Y \big / \bigwedge^n_1 Y$. This first step defines a bracket theory on the constraint submanifold.
    \item Let $h \colon \Ima \leg_\mathcal{L} \longrightarrow \Ima \Leg_\mathcal{L}$ be the diffeomorphism given by the inverse of the natural projection $\Ima \Leg_\mathcal{L} \longrightarrow \Ima \leg_\mathcal{L}$.
    \item Set $\mathcal{H} := h^\ast \Theta$ as the Hamiltonian, and study the dynamics employing the techniques developed above: existence of of special Hamiltonian forms and conserved quantities of arbitrary order.
\end{enumerate}
In order to illustrate this, let us study the particular case of Yang-Mills theories, where the Lagrangian is almost regular.

Let $X$ be an $n$-dimensional pseudo-Riemannian manifold, and let $P \longrightarrow X$ be a principal bundle with structure group $G$ acting freely and transitively on the fibers on the right. Suppose there is an $\operatorname{Ad}$-equivariant non-degenerate bilinear form $\langle \cdot , \cdot \rangle$ on $\mathfrak{g}$. Recall that a principal connection on $P$ is an Ehresmann connection that is $G$ invariant, namely a smooth choice of horizontal subspace $p \mapsto H_p$ satisfying $(R_g)_\ast H_p = H_{p \cdot g}$. Such connections can be identified with connection $1$-forms $\theta \in \Omega^{1}(P, \mathfrak{g})$, namely connections satisfying 
\begin{enumerate}[\rm(i)]
    \item For every $\xi \in \mathfrak{g}$, $\theta(\widetilde \xi) = \xi$, where $\widetilde \xi$ denotes the fundamental vector field associated to the right action of $G$ on $P$.
    \item For every $g \in G$, $(R_g)^\ast \theta = \operatorname{Ad}_{g^{-1}} \cdot\theta$.
\end{enumerate}

The curvature of $\theta$ is given by the $\mathfrak{g}$-valued $2$-form $\omega_\theta = D_\theta \,\theta$, where $(D_\theta \theta) (X, Y) = (\dd \theta) (X^h, Y ^h)$, $X^h$ and $Y^h$ denoting the horizontal components of $X$ and $Y$, respectively defined by the connection. Since $\omega_\theta$ is a semi-basic $2$-form on $P$, and it satisfies $(R_g)^\ast \omega_\theta = \operatorname{Ad}_{g^{-1}} \cdot \omega_\theta$, it induces a $2$-form $\widetilde \omega_\theta \in \Omega^2(M, \operatorname{Ad} P)$, where $\operatorname{Ad}_P = P \times_G \mathfrak{g}$ is the associated bundle to the adjoint representation of $G$. Since the chosen metric on $\mathfrak{g}$ is $\operatorname{Ad}$-equivariant, it induces a non-degenerate bilinear form on each fiber of $\operatorname{Ad}_P$, which we denote by $\langle\cdot , \cdot \rangle$ as well. This metric, together with the pseudo-Riemannian structure on $M$, allow us to introduce a scalar quantity \[
- \frac{1}{4}\langle \omega_\theta, \omega_\theta \rangle.
\]
Then one tries to solve the variational problem on the space of principal connection given by the following action
\[
- \frac{1}{4}\int_X \langle \omega_\theta, \omega_\theta \rangle \mu\,,
\]
where $\mu$ denotes the canonical volume form on $X$ induced by the pseudo-Riemannian metric.

In our setting, principal connections are given by the sections of the bundle of principal connections: $ Y = J^1(P \longrightarrow M) \big / G \longrightarrow X$ and the Lagrangian is given by 
\[\mathcal{L}: J^1 Y \longrightarrow \bigwedge^n X\, \qquad (j^1 \theta) \mapsto - \frac{1}{4}\langle \omega_\theta, \omega_\theta\rangle \mu\,.\]

If we have a local trivialization over an open subset $U \subseteq X$ given by certain section $s: U \longrightarrow P$, we have a natural identification
\[
Y|_U  = J^1(P \longrightarrow X) / G |_U \cong \T\,^\ast U \otimes \mathfrak{g}
\]
given by $\theta \mapsto s^\ast \theta$. Then, if $e_i$ denotes a basis on $\mathfrak{g}$ and $(x^\mu)$ denotes coordinates on $U$, we have induced coordinates on $Y |_U$ given by $(x^\mu, A^i_\mu)$ representing the $1$-form \[
\theta = A^i_\mu \dd x^\mu \otimes e_a\,.
\]
With this identification, $J^1 Y$carries natural coordinates $(x^\mu, A^i_\mu, A^i_{\mu, \nu})$. Then, the curvature reads
\[\omega_\theta = F^i_{\mu \nu}\dd x^\mu \wedge \dd x^\nu \otimes e_i =  \left( \pdv{A^i_\nu}{x^\mu} - \pdv{A^i_\mu}{x^\nu} + f^i_{jk} A^j_\mu A^k_\nu \right) \dd x^\mu \wedge \dd x^\nu \otimes e_i\,,\]
where $f^i_{jk}$ are the structure constants of the Lie algebra $\mathfrak{g}$,namely $[e_j, e_k] = f^{i}_{jk} e_i$. Then, the Lagrangian is locally given by 
\[\mathcal{L} = - \frac{1}{4} F^{\mu \nu}_i F^{i}_{\mu \nu} \sqrt{|g|} \dd^n x\,.\]
An elementary computation shows that the Legendre transformation is given by 
\[\Leg_\mathcal{L}^\ast p^{\mu \nu}_i  = \pdv{L}{A^i_{\mu, \nu}} =  F^{\mu \nu}_i \sqrt{|g|}\,, \quad \Leg_{\mathcal{L}}^\ast p = \pdv{L}{A^i_{\mu,\nu}} A^i_{\mu, \nu} - L =\left( - \frac{1}{4} F^{\mu \nu}_i F^i_{\mu \nu} + \frac{1}{2} f^i_{jk} F^{\mu \nu}_i A_\nu^j A_\mu^k\right) \sqrt{|g|} \] so that it defines the constraint submanifold $N = \{p^{\mu \nu}_i + p^{\nu \mu}_i = 0\}$ on $\bigwedge^n_2 Y$. We refer to \cite{sardanashvily_gauge_1993,Sardanashvily1995GeneralizedHamiltonian} for more information regarding the multisymplectic treatment of gauge theories.

Let us compute the induced graded Dirac structure on the image of the Legendre transformation, as well as the induced Hamiltonian and the algebra of special Hamiltonian forms:

\subsubsection{The induced graded Dirac structure on \texorpdfstring{$N = \Ima \Leg_{\mathcal{L}}$}{}}

As we know from Example \ref{example:multisymplectic_field_theory}, the graded Dirac structure on $\bigwedge^n_2 Y$ is given by $S^n = \langle\dd^n x\,, \dd A^i_\mu \wedge \dd^{n-1}x_\nu\,, \dd p^{\mu \nu}_i \wedge \dd^{n-1}x_\nu\rangle$ with
\[
\sharp_n(\dd^n x) = 0\,, \quad \sharp_n(\dd A^i_\mu \wedge \dd^{n-1}x_\nu) = - \pdv{p_a^{\mu \nu}}\, , \quad \sharp_n(\dd p^{\mu \nu}_i \wedge \dd^{n-1}x_\nu) =  \pdv{A^i_\mu}\,. 
\]
Recall that the pullback of a graded Dirac structure is given by the following steps:
\begin{enumerate}[\rm (i)]
    \item Identify those forms $\alpha \in S^n$ with $\sharp_n(\alpha)$ tangent to $N$,
    \item Define $\widetilde S^n$ to be generated by the pullback of said forms and let $\widetilde \sharp_n(\alpha) := \sharp_n(\alpha)$.
\end{enumerate}
In our case the family of forms $\alpha \in S^n$ such that $\sharp_n(\alpha)$ is tangent to the submanifold defined by the equations $p^{\mu\nu}_i + p^{\nu \mu}_i = 0$ is generated by the following forms:
\[
\dd^n x\,, \quad \dd A^i_\mu \wedge \dd^{n-1}x_\nu - \dd A^i_\nu \wedge \dd^{n-1}x_\mu\,,\quad \dd p^{\mu \nu}_i \wedge \dd^{n-1}x_{\nu}\,.
\]
Let us introduce coordinates $(x^\mu, A^i_\mu,\widetilde p^{\mu \nu}_i)$ on $N$, with $\widetilde p^{\mu \nu}_i = p^{\mu \nu}_i$, if $\mu < \nu$ and $\widetilde p^{\mu \nu}_i = - p^{\mu \nu}_i$, if $\mu > \nu$, so that \[\pdv{\widetilde p^{\mu \nu}_i} = \pdv{p^{\mu \nu}_i} - \pdv{p^{\nu \mu}_i}\,,\]
for $\mu < \nu$. Then, the graded Dirac structure on $N$ is locally given by $\widetilde S^n = \langle  \dd ^n x\,, \dd A^i_\mu \wedge \dd^{n-1}x_\nu - \dd A^i_\nu \wedge \dd^{n-1}x_\mu\,, \dd \widetilde p^{\mu \nu}_i \dd^{n-1}x_\nu - \dd \widetilde p^{\nu \mu}_i \wedge \dd^{n-1}x_\mu \rangle\,,$
where sum in the last term is assumed to happen for pairs of indices such that $\mu < \nu$ (and $\nu < \mu$). Let us make an abuse notation by writing
\[
 \dd \widetilde p^{\mu \nu}_i\wedge \dd^{n-1}x_\nu - \dd \widetilde p^{\nu \mu}_a \wedge \dd^{n-1}x_\mu = \dd \widetilde p^{\mu \nu}_i \wedge \dd^{n-1}x_\mu\,,
\]
and assuming that $\widetilde p^{\nu \mu}_i = - \widetilde p^{\mu \nu}_i$. Then, the $\widetilde \sharp_n$ maps are given by
\[
\widetilde \sharp_n(\dd^n x) = 0\,, \quad \widetilde \sharp_n(\dd A^i_\mu \wedge \dd^{n-1}x_\nu - \dd A^i_\nu \wedge \dd^{n-1}x_\mu) = -\pdv{\widetilde p^{\mu \nu}_i}\,, \quad \widetilde \sharp_n(\dd \widetilde p^{\mu \nu}_i \wedge \dd^{n-1}x_{\nu}) =  \pdv{A^i_\mu}\,.
\]

\subsubsection{Special Hamiltonian forms and the Hamiltonian}

From the computations of the previous subsection we conclude that the Hamiltonian is 
\[
\mathcal{H} = \left( - \frac{1}{4} \widetilde p^{\mu \nu}_i \widetilde p^{i}_{\mu \nu}+ \frac{1}{2} f^i_{jk}\widetilde p^{\mu \nu}_i A_\nu^j A_\mu^k\right) \dd^n x - \widetilde p^{\mu \nu}_i \dd A^{i}_\mu \wedge  \dd ^{n-1}x_\nu\,,
\]
and thus the equations of motion in coordinates read as
\begin{enumerate}[\rm (i)]
    \item Taking $\alpha = A^i_\mu \dd^{n-1}x_\nu - A^i_\nu \dd^{n-1}x_\mu$ we obtain
    \[ \left(\pdv{A^i_\mu}{x^\nu} - \pdv{A^i_\nu}{x^\mu} \right) \dd^n x = \psi^\ast (\dd \alpha) = \dd \alpha + \{\alpha, \mathcal{H}\} = \left( - \widetilde p^{i}_{\mu \nu} + f^i_{jk}A^j_\mu A^k_\nu\right) \dd^n x\,;\]
    \item Taking $\alpha = \widetilde p^{\mu \nu}_i \dd^{n-1}x_{nu}$ we obtain
    \[
    \pdv{\widetilde p^{\mu \nu}_i}{x^\mu} \dd^n x = \dd \alpha + \{\alpha, \mathcal{H}\} = - f^j_{ik} \widetilde p^{\mu \nu}_j A^j_k \dd^n x\,, 
    \]
\end{enumerate}
which are the well-known Yang--Mills equations. 

Let us now focus on obtaining (what we will show to be) the subalgebra of special Hamiltonian forms, by determining what values can their exterior derivatives take:
\begin{theorem} The exterior differentials of special Hamiltonian $a$-forms are generated by:
\begin{enumerate}[\rm (i)]
    \item For $0 \leq b \leq n - 3$:
    \[\sum_{j = 1}^a (-1)^{j-1} \dd A^i_{\nu_j} \wedge \dd^{a-1}x_{\nu_1, \dots, \widehat{\nu_j}, \dots, \nu_a}\,, \quad \dd^{b+1}x_{\mu_1 \cdots \mu_{n-b}}\,.\]
    \item For $b = n - 2$:
    \[\dd A^i_\alpha \wedge \dd^{n-1}x_{\beta \gamma} + \dd A^i_\gamma \wedge \dd^{n-1}x_{\alpha \beta}+ \dd A^i_\beta \wedge \dd^{n-1}x_{\gamma \alpha}\,, \quad \dd^{n-1}x_{\mu}\,, \quad \dd \widetilde p^{\mu \nu}_i \wedge \dd^{n-1}x_{\mu \nu} \,.\]
    \item For $b = n-1$: by the whole $\widetilde S^n$.
\end{enumerate}
\end{theorem} 

\begin{proof} The last item is clear. Suppose $\alpha \in \widetilde \Omega^{a-1}_H(M)$, for $1 \leq a \leq n-1$ and write
\[
\dd \alpha = C^{\mu \widetilde\mu}  \dd A^i_\mu \wedge \dd ^{a-1}x_{\widetilde \mu} + D^{i \widetilde \mu}_{\mu\nu} \dd \widetilde p^{\mu \nu}_i \wedge \dd^{a-1} x_{ \widetilde \mu}\,.
\]
Let $\varepsilon = \dd ^{n-a} x^{\widetilde \nu}$. Then, since $\alpha$ is special Hamiltonian we must have 
\begin{align*}
    \varepsilon \wedge \dd \alpha =& (-1)^{n-a} C^{\mu  \widetilde \mu} \delta^{\widetilde \nu \alpha}_{ \widetilde \mu} \dd A^i_\mu \wedge \dd^{n-1}x_ \alpha + (-1)^{n-a} D^{i \widetilde \mu}_{\mu\nu} \delta^{\widetilde \nu \alpha}_{\widetilde \mu} \dd \widetilde p^{\mu \nu}_i \wedge \dd^{n-1}x_\alpha \in S^n\,.
\end{align*}
If $\varepsilon \wedge\dd \alpha\in S^n$, we necessarily have that $\iota_U \left(\varepsilon \wedge\dd \alpha \right) = 0$ for every $U \in (S^n)^{\circ, n}$. Hence, taking
\[
U = \pdv{A^i_\alpha} \wedge \pdv{^{n-1}}{^{n-1}x_\beta} + \pdv{A^i_\beta} \wedge \pdv{^{n-1}}{^{n-1}x_\alpha}
\]
and imposing $\iota_U(\varepsilon \wedge \dd \alpha) = 0$ we get that $C^{\alpha \widetilde \mu}_i \delta^{\widetilde \nu \beta}_{\widetilde \mu} + C^{\beta \widetilde \mu}_i \delta^{\widetilde \nu \alpha}_{\widetilde \mu} = 0$. Taking $\alpha = \mu_0$, $\widetilde \nu = \mu_1, \dots \mu_{a}$, and $\beta = \mu_{a+1}$, the previous equation reads as $C^{\mu_0, \dots, \mu_{a+1}}_i = (-1)^{a+1} C^{\mu_{a+1} \mu_0, \dots, \mu_a}_i$, which means (using that the multi-index $\widetilde \mu$ is skew-symmetric) that $C^{\mu \widetilde \mu}$ is skew-symmetric in all its indices and thus we may write
\[
C^{\mu \widetilde \mu}_i \dd A^i_\mu \wedge \dd^{a-1}x_{\widetilde \mu} = C^{\widetilde \nu}_i \delta^{\mu \widetilde \mu}_{\widetilde \nu} \dd A^i_\mu \wedge \dd^{a-1}x_{\widetilde \mu} = C^{\nu_1, \dots,\nu_a}_i \sum_{j = 1}^a (-1)^{j-1} \dd A^i_{\nu_j} \wedge \dd^{a-1}x_{\nu_1, \dots, \hat{\nu_j}, \dots, \nu_a}\,.
\] 
Now, contracting with 
\[ U = \pdv{\widetilde p_i^{\alpha \beta}} \wedge \pdv{^n}{^{n-1}x_\gamma}\,,\]
where $\gamma \neq \alpha$ and $\gamma \neq \beta$ we get the condition $D^{i\widetilde \mu}_{\alpha \beta} \delta^{\widetilde \nu \gamma}_{\widetilde \mu} = 0$. This means that $D^{i \widetilde \mu}_{\alpha, \beta} = 0$ whenever $\widetilde \mu$ contains an index different from $\alpha$ and $\beta$. This last observation clearly implies that all coefficients are zero when $|\widetilde \mu| \geq 3$ or, equivalently, when $a < n-1$. This implies the first item. Now, for the second, we only need to worry about the case $n = n-1$, and a quick computation shows that the only possibility is to have $C^{i \mu \nu}_{\alpha \beta} = \delta^{\mu \nu}_{\alpha \beta}C$, for certain constant $C$, so that $C^{i \mu \nu}_{\alpha \beta}  \dd\widetilde p^{\alpha \beta}_i \wedge \dd^{n-2}x_{\mu \nu}$ is necessarily a multiple of $\dd \widetilde p^{\alpha \beta}_i \wedge \dd^{n-2}x_{\alpha \beta}$, finishing the proof.
\end{proof}

\begin{theorem} In the Yang-Mills theory, special Hamiltonian forms define a subalgebra in the algebra of Hamiltonian forms.
\end{theorem}
\begin{proof} We just check the sufficient condition given by Theorem \ref{thm_sufficient_condition_subalgebra}. Indeed, as a quick computation shows, it is enough to choose:
\begin{enumerate}[\rm (i)]
    \item For the form $\sum_{j=1}^a (-1)^{j+1} \dd A^i_{\nu_j} \wedge \dd^{a-1} x_{\nu_1, \dots, \widehat{\nu_j}, \dots, \nu_a}$ we need to take
    \[
    U = \frac{1}{n+2-a} \sum_{i < j} (-1)^{i+j} \pdv{\widetilde p^{\nu_i \nu_j}_i} \wedge \pdv{^{n-a}}{^{n-a} x^{\nu_1, \dots, \widehat{\nu_i}, \dots, \widehat{\nu_j}, \dots, \nu_{n+2-a}}}\,,
    \]
    as an easy computation shows.
    \item For the form $\dd \widetilde p^{\alpha \beta}_i \wedge \dd^{n-2}x_{\alpha \beta}$, we need to take
    \[
    U = \pdv{A^i_\alpha} \wedge \pdv{x^\alpha}\,, 
    \]
   which another quick computation shows that satisfies the conditions of Theorem \ref{thm_sufficient_condition_subalgebra}.
\end{enumerate}
\end{proof}

Finally, to end our discussion, we compute the graded Poisson brackets of some local generators of the algebra of special Hamiltonian forms. The non zero brackets are given by 
\begin{align*}
    &\left\{\sum_{j=1}^{n+2-a} (-1)^{j+1} A^i_{\nu_j} \dd^{a-1} x_{\nu_1, \dots, \widehat{\nu_j}, \dots, \nu_{n+2-a}} , p^{\alpha \beta}_j \dd^{n-2}x_{\alpha \beta}\right\} &=& - \sum_{j= 1}^{n+2-a} (-1)^{j+1} \dd^{a-2} x_{\nu_1, \dots, \widehat{\nu_j}, \dots, \nu_{n+2-a}, \nu_j}\\&&
    =& (-1)^a (n+2-a) \delta^i_j\dd^{a-2}x_{\nu_1, \dots, \nu_{n+2-a}}\,,\\
    &\left\{\sum_{j=1}^{n+2-a} (-1)^{j+1} A^i_{\nu_j} \dd^{a-1} x_{\nu_1, \dots, \widehat{\nu_j}, \dots, \nu_{n+2-a}},  \widetilde p^{\mu \nu}_i \dd^{n-1}x_\nu\right\}&=& \sum_{j= 1}^{n+2-a} (-1)^{j+1} \delta^\nu _{\nu_j}\dd^{a-1} x_{\nu_1, \dots, \widehat{\nu_j}, \dots, \nu_{n+2-a}}\,.
\end{align*}

\section{Conclusions and further work}
\label{section:conclusions}

In this paper we have analyzed the extensions of graded Poisson brackets in classical field theories in order to give a description of the evolution of forms. This study leads to the definition of special Hamiltonian forms, a particular subfamily of Hamiltonian forms that is proved to be a subalgebra under certain hypotheses (which are showed to hold in regular field theories and Yang--Mills theories). This subalgebra is the subalgebra of forms with defined evolution. To showcase the theory we compute the algebra of special Hamiltonian forms (computing the possible exterior derivatives) and the brackets of some generators of this algebra in the case of regular Lagrangians and in the case of Yang--Mills theories, although the theory works with any almost regular Lagrangian.

There are several aspects that remain unsolved, and that we propose as interesting future questions which we were not able to answer at this moment:
\begin{enumerate}[\rm (i)]
    \item As we showed, the bracket of special Hamiltonian forms and Hamiltonians is independent of the extension chosen, so they seem a natural algebra (when in the hypotheses of Theorem \ref{thm_sufficient_condition_subalgebra}) to restrict the Poisson brackets to. We propose to investigate the properties of this algebra in general because, as we showed, this algebra can be associated to any almost regular classical field theory.
    \item Determining the evolution of arbitrary Hamiltonian forms suggests to study the relationship between these brackets and the constraint algorithm (see \cite{de_leon_pre-multisymplectic_2005} for the general geometric theory and \cite{Gomis_2023} for several examples). We believe that the theory behind the constraint algorithm would benefit from an analysis employing the extension of the brackets presented in our study.
    \item We would also find interesting to investigate the relationship of these constructions with the Poisson brackets of the canonical formalism (see \cite{ks1976commun.math.phys., Margalef_Bentabol_2022}).
    \item As we mentioned in Section \ref{section:fibered_graded_Dirac}, the tools developed may be of interest to the theory of reduction and momentum maps, which is still in full development in the case of classical fields (see \cite{Blacker_2021,blacker_reduction_2024} for recent advances in the case of reduction and \cite{callies_homotopy_2016, fregier_cohomological_2015} for advances in momentum maps).
\end{enumerate}

%-------------------------------------------
%Acknoledgements
%-------------------------------------------
\section{Acknowledgements}
The authors acknowledge financial support from the Spanish Ministry of Science, Innovation and Universities under grants PID2022-137909NB-C21 and the Severo Ochoa Program for Centers of Excellence in R\&D (CEX2023-001347-S).

%\textcolor{blue}{MÁS TUYAS? UNIR?}
\phantomsection
{\Large \bf Data availability statement}

\addcontentsline{toc}{section}{Data availability statement}

This manuscript has no associated data.

\phantomsection
{\Large \bf Conflict of interest statement}

\addcontentsline{toc}{section}{Conflict of interest statement}

The authors have no conflict of interest to disclose.
%-------------------------------------------
%Bibliography
%-------------------------------------------
\phantomsection
\addcontentsline{toc}{section}{References}
\printbibliography
\end{document}